\documentclass[aps,twocolumn,preprintnumbers,amsmath,amssymb,floatfix,superscriptaddress,nofootinbib]{revtex4}

\usepackage{amsmath,amssymb,amsfonts,stmaryrd,wasysym,graphicx,multirow,color,textcomp}
\usepackage{comment,bm,euscript,amsbsy,slashed,mathtools,natbib}
\usepackage{url}
\usepackage[colorlinks=true,allcolors=blue]{hyperref}

\usepackage[latin1]{inputenc}

\numberwithin{equation}{section}

\newcommand{\ket}[1]{\ensuremath{|#1  \rangle }}

\newcommand{\set}[1]{\ensuremath{ \{ #1 \}  }}

\usepackage{physics}
\usepackage{amsthm}

\newtheorem{lemma}{Lemma}

\begin{document}

\title{Disentangling $(2+1)$d Topological States of Matter with the Entanglement Negativity}
\author{Pak Kau Lim} 
\author{Hamed Asasi} 
\affiliation{Department of Physics and Astronomy, University of California, Riverside, CA 92511, USA.}
\author{Jeffrey C.Y. Teo}
\affiliation{Department of Physics, University of Virginia, Charlottesville, VA22904, USA.}
\author{Michael Mulligan}
\affiliation{Department of Physics and Astronomy, University of California, Riverside, CA 92511, USA.}

\begin{abstract}
We use the entanglement negativity, a bipartite measure of entanglement in mixed quantum states, to study how multipartite entanglement constrains the real-space structure of the ground state wavefunctions of $(2+1)$-dimensional topological phases.
We focus on the (Abelian) Laughlin and (non-Abelian) Moore-Read states at filling fraction $\nu=1/m$.
We show that a combination of entanglement negativities, calculated with respect to specific cylinder and torus geometries, determines a necessary condition for when a topological state can be disentangled, i.e., factorized into a tensor product of states defined on cylinder subregions.
This condition, which requires the ground state to lie in a definite topological sector, is sufficient for the Laughlin state.
On the other hand, we find that a general Moore-Read ground state cannot be disentangled even when the disentangling condition holds.
\end{abstract}

\maketitle

\tableofcontents

\section{Introduction}

\subsection{Background}

One of the defining characteristics of a topological phase of matter is the sensitivity of its ground state to the topology of the space on which it's placed \cite{wen-niu-1990} (see \cite{2007AnPhy.322.1477O} for a review).
For instance, the Laughlin state at filling fraction $\nu=1/m$ has ground state degeneracy $m^g$ where $g$ is the genus of space.
Topological phases with robust ground state degeneracy, such as the Laughlin state at $m > 1$, are said to be {\it long-range entangled} \cite{2012arXiv1210.1281W}.
On the other hand, short-range entangled (topological) states \cite{PhysRevB.87.155114, PhysRevB.86.125119}, which occur in the integer quantum Hall effect, have a unique ground state when placed on any closed manifold, but share other defining topological characteristics such as protected gapless boundary modes \cite{wengaplessboundary}. 

The entanglement entropy is a useful diagnostic for these two classes of states.
The entanglement entropy between subsystems $A$ and $B$ of a state $\rho \in {\cal H}_A \otimes {\cal H}_B$ equals the von Neumann entropy $S_A = -\trace_A{ \rho_A\ln \rho_A }$ of its reduced density matrix $\rho_A = \trace_{B} \rho$.
(Here we are denoting pure and mixed states by $\rho$.)
In a topological phase, the entanglement entropy scales with the linear size $L \rightarrow \infty$ of region $A$ as \cite{hamma-2005-kitaev, levin-wen-2006, kitaev-preskill-2006}
\begin{align}
\label{entropyscaling}
S_A = \alpha L - \gamma.
\end{align}
The coefficient $\alpha$ is nonuniversal and UV divergent, while the topological entanglement entropy $\gamma$ is a universal, geometry-dependent constant that characterizes the phase.\footnote{For notational simplicity, we do not indicate the dependence of $\gamma$ on $A$, $B$, or the state.}

For instance, if $A$ is a disk, $\gamma = \frac{1}{2} \log{\sum_a d_a^2}$, where the sum is over all superselection sectors of the phase and $d_a \geq 1$ is the quantum dimension of quasiparticle $a$ \cite{kitaev-preskill-2006,levin-wen-2006}.\footnote{$d_a$ controls the Hilbert space dimension $d_a^N$ of $N$ quasiparticles $a$ as $N \rightarrow \infty$. Abelian phases have $d_a = 1$ for all $a$; non-Abelian phases have at least one quasiparticle with $d_a > 1$.} 
Short-range entangled phases have a single superselection sector (corresponding to its unique ground state) with $d_0 = 1$; long-range entangled phases, which include both Abelian states like the toric code \cite{kitaev-z2} and non-Abelian states like the Moore-Read state \cite{Moore1991}, have at least two superselection sectors and, consequently, $\gamma > 0$.
There can be other ``boundary" contributions to $\gamma$ due to interactions localized along the border of $A$ in both short-range and long-range entangled states \cite{cano-2015, santos-2018} (see also \cite{2015JSMTE..04..010O, ChandranKhemaniSondhi}).
Importantly, for long-range entangled states {\it only} and when $A$ is non-contractible, $\gamma$ can receive an additional contribution---that we generally refer to as the {\it topological sector correction}---that depends on the amplitude $\psi_a$ to be in the sector $a$ degenerate ground state \cite{2008JHEP...05..016D, 2012PhRvB..85w5151Z}.
For example, consider the ground state of a topological phase on the torus: $| \Psi \rangle = \sum_a \psi_a |\overline{\Psi}_a \rangle$, where $\psi_a$ is the amplitude to be in the ground state $|\overline{\Psi}_a \rangle$ of sector $a$. 
If the torus is divided into two cylinders $A$ and $B$, then the topological entanglement entropy of region $A$ is $\gamma = \log{\sum_a d_a^2} - \sum_a |\psi_a|^2 \log{\frac{|\psi_a|^2}{d_a^2}}$.

To better understand the distinct forms of entanglement that these different contributions to $\gamma$ reflect in a topological ground state, Lee and Vidal \cite{lee-vidal-2013}, Castelnovo \cite{castel-2013}, and Wen {\it et~al.} \cite{wen-matsuura-ryu-2016} employed the {\it entanglement negativity} \cite{vidal-werner-2002}.
Unlike the entanglement entropy, which only quantifies the quantum correlations between a subsystem and its complement when $\rho$ is pure \cite{vidal-werner-2002, 2005quant.ph..4163P, PhysRevA.72.032317}, the entanglement negativity is a mixed state entanglement measure \cite{peres-1996} that can thereby distinguish multipartite features of entanglement (e.g., \cite{dur-vidal-2000}), for instance if $\rho = \tr_C |\Psi_{ABC} \rangle \langle \Psi_{ABC} |$ obtains by tracing out degrees of freedom in a third subsystem $C$.

The entanglement negativity\footnote{This quantity is also known as the logarithmic negativity. See below for the definition of the negativity.} is motivated by Peres's \cite{peres-1996} necessary condition for a mixed state $\rho \in {\cal H}_A \otimes {\cal H}_B$ to be separable.
This criterion says that a separable state $\rho$ has positive partial transpose $\rho^{T_A}$ with respect to subsystem $A$, where
\begin{align}
\langle i_A j_B | \rho^{T_A} | k_A l_B \rangle = \langle k_A j_B | \rho | i_A l_B \rangle,
\end{align}
and $| i_A \rangle, | k_A \rangle$ ($| j_B \rangle, | l_B \rangle$) are basis states for ${\cal H}_A$ (${\cal H}_B$).
The negativity ${\cal N}_{A:B}(\rho) = (|| \rho^{T_A} ||_1 - 1)/2$ sums (the absolute value of) any negative eigenvalues of $\rho^{T_A}$ and thereby measures the degree of nonseparability of $\rho$.  
Here, $||{\cal \rho}||_1 \equiv \tr \sqrt{\rho^\dagger \rho}$ is the trace norm of $\rho$.
The entanglement negativity ${\cal E}_{A:B}(\rho)$ is a closely related measure defined as
\begin{align}
\label{entanglementnegativitydef}
{\cal E}_{A:B}(\rho) = \log|| \rho^{T_A} ||_1 = \log \big( 1 + 2 {\cal N}_{A:B}(\rho) \big) .
\end{align}
In contrast to ${\cal N}_{A:B}(\rho)$, the entanglement negativity has an operational meaning as an upper bound to the amount of pure state entanglement contained in a general mixed state \cite{vidal-werner-2002}.
For pure states, ${\cal E}_{A:B}(\rho)$ reduces to the $q=1/2$ Renyi entropy of $\rho$ \cite{lee-vidal-2013}. 
Other situations in which the entanglement negativity has been measured include conformal field theory \cite{calabrese-2012}, holography \cite{2014JHEP...10..060R, 2021arXiv210111029D}, thermal phase transitions \cite{2015JPhA...48a5006C, 2019JSMTE..04.3106S, PhysRevResearch.2.043345}, topological systems with symmetry \cite{PhysRevA.98.032302} or at nonzero temperature, \cite{PhysRevB.97.144410, PhysRevLett.125.116801} non-equilibrium systems \cite{Coser_2014, Eisler_2014, HOOGEVEEN201578, PhysRevB.92.075109, 2020arXiv201101277S}, and recently at measurement-driven phase transitions \cite{2020arXiv201200031S, 2020arXiv201200040S}.

In this paper, we use the entanglement negativity to study how multipartite entanglement constrains the structure of the manybody wave function of a topological phase.
In particular, we show how topological degeneracy can prevent the {\it disentanglement} \cite{he-vidal-2015} of a topological ground state.

In general, a state $\rho \in {\cal H}_A \otimes {\cal H}_B \otimes {\cal H}_C$ is said to satisfy the {\it disentangling condition}\footnote{He and Vidal \cite{he-vidal-2015} introduced an equality like \eqref{squaremonogamy} in terms of the negativity $\mathcal{N}$ instead of entanglement negativity ${\cal E}$. 
These two forms are equivalent when ${\cal N}_{A:C}(\rho_{AC}) = 0$.} with respect to ${\cal H}_A$ and ${\cal H_C}$ if
\begin{align}
\label{disentangling}
{\cal E}_{A:B C}(\rho) = {\cal E}_{A:B}(\rho_{AB}),
\end{align}
where $\rho_{AB} = \trace_C \rho$.
Notice that $\rho_{AB}$ is necessarily mixed ($\rho$ could also be a mixed state) and so the entanglement negativity is an appropriate measure to use to compare the the quantum correlations in $\rho_{AB}$ and $\rho$.
To appreciate \eqref{disentangling}, we can heuristically view it as a special case of the monogamy-like relation,\footnote{Monogamy-like relations such as these depend on the entanglement measure and aren't generally satisfied for all states in a given Hilbert (sub-)space.
For example, this inequality isn't satisfied generally if ${\cal N}^2$ is replaced by ${\cal N}$ \cite{he-vidal-2015}.}
\begin{align}
\label{squaremonogamy}
{\cal N}^2_{A:B C}(\rho) \geq {\cal N}^2_{A:B}(\rho_{AB}) + {\cal N}^2_{A:C}(\rho_{AC}),
\end{align}
which expresses how entanglement is shared between $A$, $B$, and $C$ subsystems \cite{coffman-2000,osborne-2006}.
Since the entanglement negativity is a monotonic function of the negativity \eqref{entanglementnegativitydef}, the disentangling condition obtains when ${\cal N}_{A:C}(\rho_{AC}) = {\cal E}_{A:C}(\rho_{AC}) = 0$, i.e., there are no quantum correlations between degrees of freedom in $A$ and $C$.
In three-qubit systems, for instance, only product states such as $\ket{\Psi_{ABC}}=\ket{\Psi_{AB}}\otimes \ket{\Psi_C}$ satisfy the disentangling condition \cite{ou-fan-2007}. 
When the Hilbert space of subsystem $B$ further factorizes as
$\mathcal{H}_B=\mathcal{H}_{B_L}\otimes \mathcal{H}_{B_R}$, pure states satisfying \eqref{disentangling} can be disentangled as
\begin{align}\label{eq-disent-decomp-pure}
  \ket{\Psi_{ABC}} = \ket{\Psi_{AB_L}} \otimes \ket{\Psi_{B_R C}},
\end{align}
a result known as the disentangling theorem \cite{he-vidal-2015}.
A more general set of states that fulfill the disentangling condition are those that saturate the strong subadditivity of the entanglement entropy
\cite{gour-guo-2018}, i.e., $I_{A: B C} = I_{A:B}$, where the mutual information $I_{A: B} = S_A + S_B - S_{A B}$.
For such states, Hayden {\it et~al.}~\cite{hayden-2004} showed there exists a decomposition of the Hilbert space as
\begin{align}
  \mathcal{H}_B = \bigoplus_j \mathcal{H}_{B_L^j}\otimes \mathcal{H}_{B_R^j}
\end{align}
such that $\rho$ is separable:
\begin{align}\label{eq-ssa-struct}
\rho = \sum_j p_j \rho_{AB_L^j}\otimes \rho_{B_R^j C}.
\end{align}
Here $\set{p_j}$ are probabilities. 

\subsection{Summary of Results}
\begin{figure}[t]
  (a)\includegraphics[width=.2\textwidth]{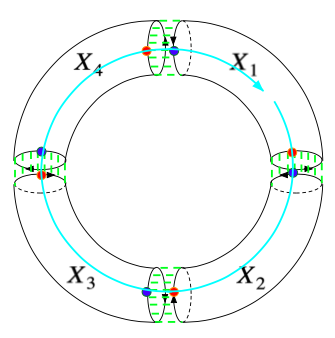}%
    \hspace{0.2cm}
  (b)\includegraphics[width=.2\textwidth]{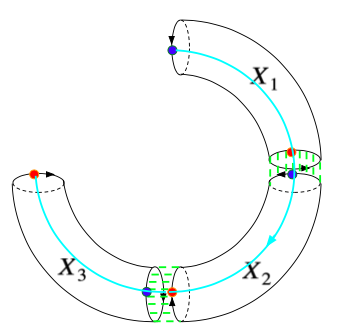}%
\caption{(a) ({\it Torus geometry}) Decomposition of a torus $X$ into $2M = 4$ cylinders $X_1, X_2, X_3, X_4$ with $X_{\rm odd} = X_1 \cup X_3$ and $X_{\rm even} = X_2 \cup X_4$; (b) ({\it Cylinder geometry}) Degrees of freedom in cylinder $\bar{Y} = X_4$ have been traced over; the remaining cylinders $Y = Y_{\rm odd} \cup Y_{\rm even}$ with $Y_{\rm odd} = X_1 \cup X_3$ and $Y_{\rm even} = X_2$ have $R=2$ shared interfaces. Interactions between low-energy boundary modes at cylinder interfaces are indicated by dashed green lines. Superselection sector $a$ is represented by the blue (quasiparticle) threading the center of the (solid) torus.}
\label{4cylinders}
\end{figure}
We study the disentangling condition \eqref{disentangling} for the (Abelian) Laughlin and (non-Abelian) Moore-Read states at filling fractions $\nu = 1/m$.
We show explicitly how Laughlin states satisfying this condition can be disentangled according to either Eqs.~\eqref{eq-disent-decomp-pure} or \eqref{eq-ssa-struct}.
Interestingly, we find that a general Moore-Read ground state cannot be disentangled even when \eqref{disentangling} is satisfied. 

To do this, we use the cut and glue construction of these states \cite{Elitzur:1989nr, doi:10.1142/S0217979292000840, Qi-Katsura-Ludwig-2012, lundgrenentanglement, Teo-Kane-2014} to calculate the entanglement negativity in two related geometries (see Fig.~\ref{4cylinders}).
(When there is overlap, our results agree with \cite{lee-vidal-2013, castel-2013, wen-matsuura-ryu-2016}.)
In the first, we partition a torus into $2M$ cylinders $X_i$ ($i \in \{1, \ldots, 2M$\}) and perform partial transposition with respect to degrees of freedom on the ``odd" cylinders $X_{\rm odd} \equiv X_1 \cup X_3 \cup \cdots \cup X_{2M-1}$ (i.e., cylinders $X_1$ and $X_3$ in Fig.~\ref{4cylinders}a).
We find that the entanglement negativity is
\begin{align}
\label{torusgeneralstatement}
    {\cal E}_{X_{\rm odd}: X_{\rm even}} = 2 \log \sum_a |\psi_a| \zeta_a^{2M},
\end{align}
where $X_{\rm even} \equiv X_2 \cup X_4 \cup \cdots \cup X_{2M}$ and $\psi_a$ is the unit-normalized amplitude to be in the sector $a$ torus ground state. 
$\zeta_a$ is a ratio of sector $a$ edge state partition functions at inverse ``temperatures" $\beta = 1/2$ and $\beta = 1$: 
\begin{align}
\label{zetadef}
    \zeta_a = \frac{\tr e^{-H_a/2}}{\sqrt{\tr e^{-H_a}}}
\end{align}
with entanglement Hamiltonian $H_a$.
This dependence of the entanglement negativity on the spectrum of the entanglement Hamiltonian is reminiscent of a similar dependence ($\rho_A \propto e^{- H_a}$) of the entanglement entropy, e.g., \cite{lihaldane, Regnault2009, Thomale2010, lauchli2010, papic2011, chandran2011, hermanns2011, Simon2013, Pollman2010, Fidkowski2010, Prodan2010, fang2013}.
In contrast to the entanglement entropy, the entanglement negativity measures the system at two different ``temperatures" \cite{2014JHEP...10..060R}. 
For the fully chiral topological phases that we study, i.e., when all the edge modes move in the same direction, $H_a$ is proportional to the edge state Hamiltonian.
In general, there can be a different $H_a$ for each of the $2M$ interfaces \cite{cano-2015}; here we only consider torus states where the interactions are the same at each interface.

The second geometry that we consider is obtained by tracing over the degrees of freedom on $N \leq M$ cylinders $\bar{Y} \subset X$ [for example, $X_4$ in Fig.~\ref{4cylinders}(b)].
We show that the entanglement negativity of the resulting state is
\begin{align}
\label{cylindergeneralstatement}
    {\cal E}_{Y_{\rm odd}:Y_{\rm even}} = \log \sum_a \left( |\psi_a| \zeta_a^{R} \right)^2,
\end{align}
where $R$ is the number of shared interfaces between the remaining cylinders $Y_{\rm odd}$ and $Y_{\rm even}$ whose degrees of freedom have not been traced over [e.g., $R=2$ in Fig.~\ref{4cylinders}(b)] and $\zeta_a$ is again given in \eqref{zetadef}.

Thus, the entanglement negativities \eqref{torusgeneralstatement} and \eqref{cylindergeneralstatement} are determined by ratios of entanglement Hamiltonian partition functions.
For the Laughlin and Moore-Read states, we show that the above entanglement negativities take the form:
\begin{align}
\label{torusresult}
    {\cal E}_{X_{\rm odd}: X_{\rm even}} & = M \alpha L - M \log {\cal D}^2 + 2 \log \sum_a |\psi_a| d^M_a, \\
    \label{cylinderresult}
        {\cal E}_{Y_{\rm odd}: Y_{\rm even}} & = {\frac{R}{2}} \alpha L - {\frac{R}{2} } \log {\cal D}^2 + \log \sum_a |\psi_a|^2 d^{R}_a,
\end{align}
where $\alpha$ is nonuniversal, $d_a$ is the quantum dimension of quasiparticle $a$, and ${\cal D} = \sqrt{\sum_a d_a^2}$ is the total quantum dimension of the phase.
The  Laughlin state has $m$ Abelian anyons each with quantum dimension $d_a = 1$; the Moore-Read state has $2m$ Abelian anyons ($d_a = 1$) and $m$ non-Abelian anyons with quantum dimensions $d_a = \sqrt{2}$, corresponding to the Majorana quasiparticle.

We use these entanglement negativities \eqref{torusresult} and \eqref{cylinderresult} to test the disentangling condition \eqref{disentangling} for the geometries in Fig.~\ref{4cylinders}.
For a general topological state on the torus, we find\footnote{In \eqref{torusresult} we set $M = 1$ for the two cylinders $A = X_2$ and $B \cup C =(X_1 \cup X_3)\cup X_4$; in \eqref{cylinderresult} we set $R=2$ for the two cylinders $A = X_2$ and $B = X_1 \cup X_3$; and we use $\mathcal{E}_{A:B} = {\cal E}_{B:A}$.}
\begin{align}
\label{longrangeentanglementgeneral}
{\cal E}_{A:B C}(\rho) - {\cal E}_{A:B}(\rho_{AB}) = 
\log \frac{\big(\sum_a |\psi_a| d_a \big)^2}{\sum_a |\psi_a|^2 d_a^2}.
\end{align}
Thus, the disentangling condition is only satisfied when the torus state lies in a specific topological sector with $\psi_a = 1$ for some $a$ and all other amplitudes equal to zero.
For topological states on the cylinder, the disentangling condition is always satisfied.

We find the disentangling condition \eqref{disentangling} is generally only a necessary condition to allow the disentanglement of a topological state.
Specifically, we show that Laughlin and untwisted sector Moore-Read states
can be disentangled according to Eqs.~\eqref{eq-disent-decomp-pure} and \eqref{eq-ssa-struct} when \eqref{disentangling} holds;
on the other hand, twisted sector Moore-Read states cannot be disentangled even when the disentangling condition is satisfied.
(As we review later, the Moore-Read state decomposes into so-called untwisted and twisted sectors, associated to Abelian and non-Abelian bulk quasiparticles.)
These results provide a precise illustration for how entanglement and non-Abelian topological order constrain a manybody wave function.

The remainder of this paper is organized as follows.
In \S \ref{sectiontwo}, we review the edge-state theories for the Laughlin and Moore-Read states at filling fraction $\nu = 1/m$ and how the torus or cylinder ground state is built out of topological states on sub-cylinders (e.g., according to the geometry in Fig.~\ref{4cylinders}).
In \S \ref{sectionthree}, we derive the entanglement negativities in Eqs.~\eqref{torusresult} and \eqref{cylinderresult}.
In \S \ref{sectionfour}, we discuss the implications of these results for disentangling topological states.
In \S \ref{sectionfive}, we conclude and discuss possible directions of future study.

\section{Cut and Glue Approach to Torus Ground States}
\label{sectiontwo}

In this section we review the edge-state theories for the Laughlin and Moore-Read states and how topological states on the torus can be decomposed into states on the sub-cylinders using the corresponding edge states.
In the next section we study the multipartite entanglement properties of these torus and cylinder states.

\subsection{Laughlin Interface Ground State}
\label{SecIIB}

We start by discussing the construction of the Laughlin state at filling fraction $\nu=1/m$ on the torus. 
One approach is to ``glue" together a collection of parallel 1d wires each hosting a single, nonchiral electron by suitable sine-Gordon inter-wire couplings \cite{Teo-Kane-2014}. 
An equivalent approach \cite{Elitzur:1989nr, doi:10.1142/S0217979292000840, Qi-Katsura-Ludwig-2012, lundgrenentanglement, Teo-Kane-2014}, which we follow here, is to construct the torus state by ``gluing" together a collection of cylinder states in the target phase of interest along their shared boundaries  by appropriate edge-state interactions.

In the Laughlin phase, each cylinder $X_i$ with $i \in \{1, \ldots, 2M \}$ hosts a pair of ${\rm U}(1)_{\rm m}$ chiral edge modes $\phi_i^\sigma$ with Lagrangian density,
\begin{align}
\label{laughlinlagrangian}
\mathcal{L}^\sigma_i = \frac{m}{4\pi}\partial_x\phi^\sigma_i\left(\sigma\partial_t-v_c\partial_x\right)\phi^\sigma_i.
\end{align}   
Here, $\phi^\sigma_i \sim \phi^\sigma_i + 2\pi \mathbb{Z}$ with $\sigma = L(R) = +1(-1)$ is a real, boson field that takes values on a circle of unit radius and $v_c > 0$ is the common\footnote{This simplification does not affect our conclusions; it merely simplifies the presentation.} velocity of the edge modes.
The charge density on each edge is $\rho^\sigma_i=\partial_x\phi^\sigma_i/(2\pi)$ in units where $e=1$.
The Lagrangian implies the equal-time commutation relations, 
\begin{align} \left[\phi^\sigma_i(x),\partial_{x'}\phi^\sigma_i(x')\right]=\frac{2\pi i \sigma}{m}\delta(x-x').
\label{ETCRLaughlin}\end{align} 
The primary fields of the theory are the vertex operators $e^{i r \phi^\sigma_i}$ for $r \in \{0, 1, \ldots, m-1\}$.
They carry charge $\sigma r/m$ and spin\footnote{The spin of an operator with left and right scaling dimensions $(h_L, h_R)$ equals $|h_R - h_L|$.} $r^2/2m$.
For $r > 0$, these operators create/destroy for $\sigma = L/R$ fractionally-charged Laughlin quasiparticles at a point along the edge.
The monodromy braiding phase between bulk quasiparticles, corresponding to operators $e^{ir\phi}$ and $e^{ir'\phi}$, equals $e^{2 \pi irr'/m}$. 
Local quasiparticles correspond to products of the fundamental electronic operator $e^{i m \phi^\sigma_i}$ carrying unit charge and integer (half-integer) spin when $m$ is even (odd). 
The braiding phase between mutually local quasiparticles is trivial, i.e., equal to one.
(For example, when $m$ is odd, $e^{i m \phi^\sigma_i}$ creates/destroys an electron on the edge.)

Take the boundary circles on each cylinder to have circumference $L$.
Then $\phi^\sigma_i$ has the mode expansion: 
\begin{align}
\begin{split} \phi^R_i &= \phi^R_{i,0} + 2\pi N^{RX_i}\frac{x}{L} \\
&\;\;\;\;+ \sum_{k>0}\sqrt{\frac{2\pi}{mL|k|}}\left(a_{i,k}e^{ikx}+(a_{i,k})^\dagger e^{-ikx}\right), \\
\phi^L_i &= \phi^L_{i,0} + 2\pi N^{LX_i}\frac{x}{L} \\
&\;\;\;\;+ \sum_{k<0}\sqrt{\frac{2\pi}{mL|k|}}\left( a_{i,k}e^{ikx} + 
(a_{i,k})^\dagger e^{-ikx} \right)\end{split}\label{ModeExpansionMRboson}
\end{align} with integer quantized momenta $k=2\pi j/L$ and $j\in \mathbb{Z}\backslash \{0\}$. 
Here, $k>0$ ($k<0$) corresponds to a right (left) mover.
The superscript $RX_i$ ($LX_i$) refers both to the right (left) edge and the right-moving (left-moving) edge mode of cylinder $X_i$.
The equal-time commutation relations imply the mode operators in \eqref{ModeExpansionMRboson} satisfy the following commutation relations:
\begin{align}\begin{split}
\left[a_{i,k}, (a_{i,k'})^\dagger\right] &= \delta_{k,k'}, \quad \big[a_{i,k}, a_{i,k'}\big]=0,\\
\left[\phi^R_{i,0},N^{RX_i}\right]&=-\left[\phi^L_{i,0}, N^{LX_i}\right]=-\frac{i}{m}.
\end{split}
\end{align} 

The winding number $N^{\sigma X_i}$ measures the total charge of the $\sigma X_i$ edge state since
\begin{align}N^{\sigma X_i} = \int_0^L \frac{\partial_x\phi_i^\sigma}{2\pi}dx.
\label{Laughlinwindingnumberdef}
\end{align}
The local operator $e^{im\phi^\sigma_i}$ obeys periodic boundary conditions (in the absence of any additional fields). 
For this requirement to be consistent with Eq.~\eqref{Laughlinwindingnumberdef},
\begin{align}
e^{im\phi_i^\sigma(x+L)} = e^{im\phi_i^\sigma(x)}e^{im2\pi N^{\sigma X_i}},
\end{align} the winding number must be quantized as $N^{\sigma X_i} -\sigma\frac{a}{m}\in\mathbb{Z}$ \cite{sohal-2020}. 
Thus, $a=0,1,\ldots,m-1$ (mod m) specifies $m$ inequivalent boundary conditions for $\phi_i^\sigma$.
As the notation suggests, these boundary conditions are in 1:1 correspondence with the different anyon types.
In particular, boundary condition $a$ can be viewed as resulting from threading the flux of anyon $a$ through the cylinder (see Fig.~\ref{fig:cylinder}). 
Each of these boundary conditions corresponds to a Wilson line of type $a$ connecting the two edges, obtained by the creation of an anyon of type $a$ on, say, the left edge and its subsequent destruction on the right edge.
\begin{figure}[t]
\centering
\includegraphics[width=0.3\textwidth]{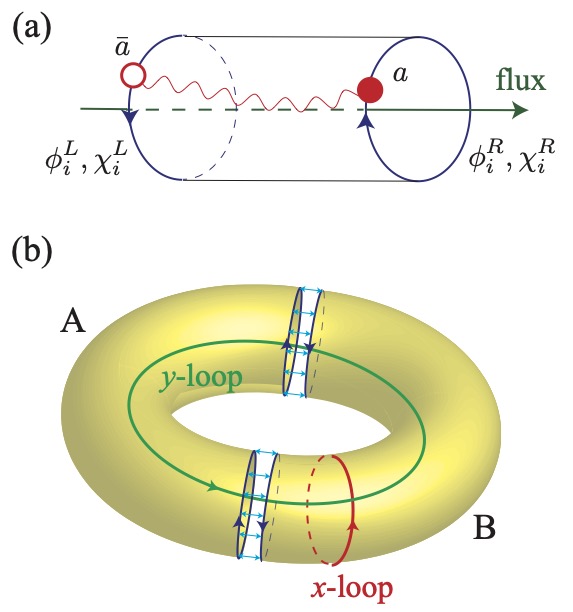}\caption{(a) Anyon flux threading continuously across a cylinder with no bulk excitation. $\phi_i^{L,R}$ refer to bosonic edge modes of the Laughlin and Moore-Read states; $\chi_i^{L,R}$ refer to fermionic edge modes only in the Moore-Read state. 
(b) Wilson string operators in the $x$ or $y$ directions parallel or perpendicular to interfaces between cylinders $A$ and $B$.
}\label{fig:cylinder}
\end{figure}

We are interested in ``gluing" together the right edge states of cylinder $X_{i-1}$ to the left edge states of cylinder $X_i$ to form the torus state.
(The subscripts are $2M$ periodic: $X_0 \equiv X_{2M}$ and therefore $X_{2M + 1} \equiv X_1$.)
This means we want to add a suitable interaction between edge modes on the right edge of cylinder $X_{i-1}$ and the left edge of cylinder $X_i$ that results in a gapped state along their shared interface $i$.
According to \eqref{laughlinlagrangian}, before the interaction is added, the relevant edge modes are controlled by the Hamiltonian,
\begin{align} H^{(0)}_{i} = \frac{mv_c}{4\pi}\int_0^L dx\left[(\partial_x\phi^R_{i-1})^2 + (\partial_x\phi^L_i)^2\right].
\end{align} 
The edges are glued together by an interaction that tunnels a local boson or fermion between nearby edges.
This is accomplished by the sine-Gordon interaction, 
\begin{align} H^{(1)}_{i} = -\frac{2g}{\pi}\int_0^L dx \cos\left[m\left(\phi^R_{i-1} +\phi^L_i \right)\right].
\end{align} 
We take coupling of the interaction $g>0$ to be independent of $i$. 
The total Hamiltonian at interface $i$ is therefore
\begin{align} 
H_{i} = H^{(0)}_{i}+H^{(1)}_{i}.
\end{align} 
The resulting torus Hamiltonian is then $H = \sum_i H_i$.
Upon projecting each cylinder $X_i$ into the same topological sector $a$, i.e., all edge modes obey the same boundary conditions around $L$, these decoupled $H_i$ may be considered independently.

For large coupling $g\rightarrow \infty$, we approximate the sine-Gordon potential at quadratic order in an expansion in $(\phi^R_{i-1} +\phi^L_i )$ \cite{lundgrenentanglement, Teo-Kane-2014}.
This is a dramatic simplification that enables the following exact solution to the approximated $H_i$; it relies on the ability of the sine-Gordon potential to generate a gapped interface ground state.
(We will denote and refer to the approximated Hamiltonian by $H_i$.)
Using the mode expansion \eqref{ModeExpansionMRboson} for the bosons, the total Hamiltonian decouples into zero and oscillation (osc) mode sectors: 
\begin{align} 
H_i = H^{\mathrm{zero}}_{i,b} + H^{\mathrm{osc}}_{i,b}.\label{ApproximateLaughlin}
\end{align} 
Defining $X_i = m\left(N^{RX_{i-1}}-N^{LX_i}\right)/2$ and $P_i=\phi^R_{i-1,0}+\phi^L_{i,0}$ such that $[X_i,P_i]=i$, the zero mode Hamiltonian is
\begin{align}
\label{zeromodeLaughlin}
H^{\mathrm{zero}}_{i,b} = \frac{2\pi v_c }{mL}X_i^2 + \frac{\pi \lambda v_c L}{2}P_i^2,
\end{align} 
where $\lambda \equiv 2gm^2/\pi^2v_c>0$.
This has the form of an harmonic oscillator Hamiltonian and a corresponding ground state, 
\begin{widetext}
\begin{align}
\label{GSLzm}
    |b^{\mathrm{zero}}_{a,i} \rangle & = \sum_{N_{a,i} \in \mathbb{Z} - {\frac{a}{m}}} e^{- {v_e \pi m \over 2 L} N_{a,i}^2} |N^{RX_{i-1}} = N_{a,i} \rangle_{RX_{i-1}} \otimes |N^{LX_{i}} = - N_{a,i} \rangle_{LX_{i}}.
\end{align}
\end{widetext}
$RX_{i-1}$ ($LX_{i}$) labels the Hilbert space of edge modes on the right (left) boundary of cylinder $X_{i-1}$ ($X_{i}$) with $X_{0} \equiv X_{2M}$ and $v_e = \frac{2}{\pi}\sqrt{\frac{m}{\lambda}}$ is the entanglement velocity.

The oscillation mode Hamiltonian is
\begin{align}
\begin{split} H^{\mathrm{osc}}_{i,b} &= v_c\sum_{k>0}\begin{bmatrix}
(a_{i-1,k})^\dagger & a_{i,-k}\end{bmatrix}\begin{bmatrix}
A_k & B_k \\ B_k & A_k\end{bmatrix}\begin{bmatrix}
a_{i-1,k} \\ (a_{i,-k})^\dagger\end{bmatrix} \end{split}\label{OscBH}
\end{align} with $A_k\equiv |k|+\frac{2\lambda\pi^2}{m|k|}$ and $B_k\equiv\frac{2\lambda\pi^2	}{m|k|}$. 
Using the Bogoliubov transformation,
\begin{align}
\begin{bmatrix}\beta_{i,k} \\ (\gamma_{i,k})^\dagger\end{bmatrix}
=\begin{bmatrix}\cosh\theta_k & \sinh\theta_k \\ \sinh\theta_k & \cosh\theta_k\end{bmatrix}\begin{bmatrix}
a_{i-1,k} \\ (a_{i,-k})^\dagger
\end{bmatrix},
\end{align} 
where $\cosh 2\theta_k = A_k/E_k$ and $\sinh 2\theta_k = B_k/E_k$ with $E_k=\sqrt{|k|^2+4\lambda\pi^2/m}$, the oscillation mode Hamiltonian is diagonal, \begin{align}H^{\mathrm{osc}}_{i,b}=v_c\sum_{k>0}E_k\left(\beta_{i,k}^\dagger\beta_{i,k}+\gamma_{i,k}^\dagger\gamma_{i,k}+1\right).
\end{align} 
The ground state of the diagonalized Hamiltonian is given by the coherent state \citep{sohal-2020}, 
\begin{align}
|b^{\mathrm{osc}}_i\rangle &=\prod_{k>0}e^{-\Omega_k (a_{i-1,k})^\dagger(a_{i,-k})^\dagger}|0\rangle,
\label{LCoherentState}
\end{align} 
where $|0\rangle$ is the vacuum state annihilated by all $a_{i-1, k}$ and $a_{i,-k}$. 
$|b^{\mathrm{osc}}_{i}\rangle$ satisfies $\beta_{i,k}|b^{\mathrm{osc}}_{i}\rangle = \gamma_{i,k}|b^{\mathrm{osc}}_{i}\rangle=0$ for $k > 0$ with $\Omega_k = \tanh\theta_k$. 
In the limit $|k| \ll \lambda$, $\tanh\theta_k\approx v_e k/2$.
Upon expanding the exponential in \eqref{LCoherentState},  the oscillation ground state can be rewritten as 
\begin{widetext}
\begin{align}
\label{GSLosc}
|b^{\mathrm{osc}}_{i} \rangle & = \sum_{\{n_{i,k} \in \mathbb{Z}^+ \}} e^{-\sum_{k>0} {v_e k \over 2} (n_{i,k} + 1/2)} |\{n^{RX_{i-1}}_{b, k}= n_{i,k} \}_{k>0} \rangle_{RX_{i-1}} \otimes |\{n^{LX_{i}}_{b,- k} = n_{i,k} \}_{k>0} \rangle_{LX_{i}}.
    \end{align}
    \end{widetext}
Here, $n^{RX_{i-1}}_{b,k}$ is the eigenvalue of the right-moving number operator $(a_{i-1,k})^\dagger a_{i-1,k}$ on cylinder $X_{i-1}$ and $n^{LX_{i}}_{b,k}$ the eigenvalue of the left-moving number operator $(a_{i,-k})^\dagger a_{i,-k}$ on cylinder $X_i$.
The coherent state form for $|b^{\mathrm{osc}}_i\rangle$ in \eqref{LCoherentState} ensures these two eigenvalues coincide in each interface oscillator state.

Putting together these results, we find the unnormalized torus state in sector $a$ equals
\begin{align}
|\Psi_a\rangle = \bigotimes_i |b^{\mathrm{zero}}_{a,i}\rangle\otimes|b^{\mathrm{osc}}_i\rangle.
\label{LaughlinGroundState}
\end{align} 
The topological sector label $a=0,1,\ldots,m-1$ coincides with the $m$-fold ground state degeneracy of the Laughlin phase on the torus.
Notice that each cylinder is in the same topological sector $a$.
This follows from our assumption that there are no bulk excitations inside any cylinder. 
Consequently, {\em all} cylinders are threaded by the same anyon flux $a$ and $N^{RX_i} = -N^{LX_i} =a/m$ mod 1 for all $X_i$ (see Fig.~\ref{fig:cylinder}). 
Using \eqref{GSLzm}, $N_{a,i}=-N^{LX_i} = N^{RX_i}=N_{a,i+1}$ mod $1$ and therefore $N_{a,i}\equiv a/m$ mod 1 for all cylinder $i$. 
A general (unnormalized) ground state on the torus is the linear combination of states $|\Psi_a\rangle$ with different anyon fluxes $a$.

\subsection{Moore-Read Interface Ground State}

The Moore-Read 
state at filling fraction $\nu=1/m$ has  $({\rm U}(1)_{\rm m}\times \text{Ising})/\mathbb{Z}_2$ topological order. 
The $\mathbb{Z}_2$ symmetry couples together the ${\rm U}(1)_{\rm m}$ and $\text{Ising}$ topological orders.
The ${\rm U}(1)_{\rm m}$ sector edge states are described by the same bosonic fields $\phi^\sigma_i$ used in the construction in the Laughlin state.
In particular, the commutation relations \eqref{ETCRLaughlin} and mode expansions \eqref{ModeExpansionMRboson} still hold.  
The Ising sector, which has electrically-neutral Majorana fermion edge states,
supports bulk quasiparticles $1,\chi$, and $\xi$.
Here, $1$ labels the identity sector containing the vacuum; $\chi=\chi^\dagger$ is the neutral Majorana fermion; and $\xi$ is the non-Abelian Ising twist field. 
The Ising anyon and the Majorana fermion have mutual semionic statistics, so that the monodromy braiding phase between $\chi$ and $\xi$ is $-1$.

\begin{widetext}
\centering
\begin{table*}[htbp]
\begin{tabular}{c|cccccccccc}
&1&$e^{1/2}$&$e^1$&$e^{3/2}$&$\ldots$&$e^r$&$e^{r+1/2}$&$\ldots$&$e^{m-1}$&$e^{m-1/2}$\\\hline
1&0&$\ast$&$\frac{1}{2m}$&$\ast$&$\ldots$&$\frac{r^2}{2m}$&$\ast$&$\ldots$&$\frac{m^2+1}{2m}$&$\ast$\\
$\chi$&0&$\ast$&$\frac{1}{2}+\frac{1}{2m}$&$\ast$&$\ldots$&$\frac{1}{2}+\frac{r^2}{2m}$&$\ast$&$\ldots$&$\frac{1}{2}+\frac{m^2+1}{2m}$&$\ast$\\
$\xi$&$\ast$&$\frac{1}{16}+\frac{1}{8m}$&$\ast$&$\frac{1}{16}+\frac{9}{8m}$&$\ldots$&$\ast$&$\frac{1}{16}+\frac{(2r+1)^2}{8m}$&$\ldots$&$\ast$&$\frac{1}{16}+\frac{(2m-1)^2}{8m}$
\end{tabular}
\caption{The $3m$ anyon types of the Moore-Read topological order. Occupied entries are the spins (mod 1) of distinct (deconfined) anyons, $I^r=e^r \equiv e^{ir\phi}$, $\chi^r=\chi e^r$ and $\xi^{r+1/2}=\xi e^{r+1/2}$, for $r=0,1,\ldots,m-1$. Empty entries ($\ast$) are confined fields disallowed by electron locality.}\label{tab:MRanyons}
\end{table*}
\end{widetext}

We set the notion of locality in the Moore-Read edge-state theory by taking the fundamental electronic operator to be $\psi_{\rm el}=\chi e^{im\phi}$.
When $m$ is even, $\psi_{\rm el}$ is a fermion; when $m$ is odd, $\psi_{\rm el}$ is a boson.
Integral combinations of the fundamental electronic operator, such as $e^{\pm 2im\phi}$, belong to the identity sector. 
They are mutually local in the sense that the corresponding bulk quasiparticles have trivial monodromy braiding phases with one another.

The remaining anyons in the Moore-Read
theory correspond to the operators $I^r=e^{ir\phi}$, $\chi^r=\chi e^{ir\phi}$ and $\xi^{r+1/2}=\xi e^{i(r+1/2)\phi}$, where $r \in \{0, 1, \ldots, m -1 \}$.
The corresponding anyons have trivial braiding monodromy with linear combinations of $\psi_{\rm el}$. 
The anyons obey the fusion rules: 
\begin{align} \begin{split}&I^r\times I^{r'}=\chi^r\times\chi^{r'}=I^{r+r'},\quad I^r\times\chi^{r'}=\chi^{r+r'}\\&I^r\times\xi^{r'+1/2}=\chi^r\times\xi^{r'+1/2}=\xi^{r+r'+1/2},\\&\xi^{r+1/2}\times\xi^{r'+1/2}=I^{r+r'+1}+\chi^{r+r'+1}.\end{split}
\label{Isingfusionrules}
\end{align} 
The fusion rules imply the quantum dimensions $d_{I^r}=d_{\chi^r}=1$ and $d_{\xi^{r+1/2}}=\sqrt{2}$. 
The locality of the electronic operator dictates that fields that differ by $\psi_{\rm el}$ belong in the same anyon class, $a\times\psi_{\rm el}\equiv a$. 
Hence, the anyon types have a $m$-fold (i.e., charge $e$) periodicity \begin{align}\chi^{r+m}\equiv I^r,\quad I^{r+m}\equiv\chi^r,\quad\xi^{r+m+1/2}\equiv\xi^{r+1/2}.\label{MRanyonperiodicity}\end{align} 
In total, there are $3m$ distinct anyon classes; they are listed in Table~\ref{tab:MRanyons}.

Bulk anyonic quasiparticles are non-local excitations that must come in conjugate pairs in real space, i.e., the total anyon charge contained in a region is conserved.
Anyons in the physical Hilbert space are identified by equivalence classes of particles. Two anyons belong to the same class if they differ by a multiple of the electronic operator.
Different topological sectors on the torus are obtained by imagining a process in which an anyon-anti-anyon pair is nucleated at a point and then each is dragged around the $y$-loop in Fig.~\ref{fig:cylinder} in opposite directions until they meet again and annihilate.
Decomposing the torus into cylinders, edges of adjacent cylinders must therefore carry conjugate anyon charge (see Fig.~\ref{4cylinders}).
This constraint was imposed implicitly when we considered the Laughlin state by requiring each cylinder to lie in sector $a$; in the present case, the presence of non-Abelian quasiparticles makes this more delicate, as we discuss. 

The Moore-Read edge-states on cylinder $X_i$ are described by the Lagrangian density \cite{Milovanovic-Read-1996},   
\begin{align} \begin{split}
\mathcal{L}^\sigma_{i} &=\frac{i}{2}\chi^\sigma_i
\left(\partial_t-\sigma v_m\partial_x\right)\chi^\sigma_i\\
&\;\;\;\;+\frac{m}{4\pi}\partial_x\phi^\sigma_i\left(\sigma\partial_t -v_c\partial_x\right)\phi^\sigma_i. \end{split}\label{MRLagrangian}
\end{align} 
As before, $\phi_i^\sigma$ is a real boson with unit compactification radius and $\sigma \in \{L,R\} = \{\pm 1\}$, $\chi_i^\sigma$ is a Majorana fermion, and $v_c$ ($\tilde{v}_c$) is the velocity of the the boson (Majorana fermion). 
$\chi_i^\sigma$ satisfies the anti-commutation relations,
\begin{align}
\{\chi^\sigma_i(x),\chi^\sigma_i(x')\}&=\delta(x-x').
\label{ChargednNeutralSectors}
\end{align} 
The mode expansions \citep{sohal-2020} of the Majorana fermion fields are 
\begin{align} 
\chi^R_i &= \frac{1}{\sqrt{L}} \sum_k e^{ikx} c^R_{i,k},\quad \chi^L_i = \frac{1}{\sqrt{L}} \sum_k e^{ikx} c^L_{i,k}.
\label{ModeExpansionMajorana}
\end{align} 
The fermionic mode operators $c^\sigma_{i,k}$ obey $(c^\sigma_{i,k})^\dagger = c^\sigma_{i,-k}$ since $\chi_i^\sigma$ is real, and the anti-commutation relations \begin{align}
\{c^\sigma_{i,k}, c^{\sigma'}_{i',k'}\} = \delta_{k,-k'}\delta_{ii'}\delta^{\sigma\sigma'}.
\end{align} 

The Moore-Read
state is classified into untwisted and twisted sectors \cite{sohal-2020}. 
In the untwisted sector, the Majorana fermions obey anti-periodic boundary conditions [$\chi^\sigma_i(x+L) = - \chi^\sigma_i(x)$]. 
Consequently, the fermionic momenta are quantized in half integers: $k=\frac{2\pi}{L}(j+1/2)$ with $j\in\mathbb{Z}$. 
This sector consists of Abelian quasiparticles that correspond to vertex operators $\{e^{ir\phi^\sigma_i}, \chi^\sigma_ie^{ir\phi^\sigma_i}\}$. 
The boson winding number is quantized as $N^{\sigma X_i}-\sigma r/m\in\mathbb{Z}$. 

In the twisted sector, the Majorana fermion is periodic ($\chi_i^\sigma(x+L) = \chi_i^\sigma(x)$) and the fermionic momenta are integrally quantized: $k=\frac{2\pi j}{L}, j\in\mathbb{Z}$.
The change in boundary conditions is effected by inserting a $\pi$ flux through the cylinder. 
In addition to the fermion oscillation modes with nonzero momenta ($k > 0$), there is an additional Majorana zero mode ($k=0$) $c^\sigma_{i,0}$ due to the integral quantization of momenta in the twisted sector.
The boson winding number also changes its quantization to $N^{\sigma X_i}-\sigma\frac{r+1/2}{m}\in \mathbb{Z}$ in response to the added $\pi$ flux.
The non-Abelian bulk quasiparticles are associated to the vertex operators $\{e^{i(r+1/2)\phi^\sigma_i} \}$.

While the boson and fermion modes are decoupled in the Lagrangian \eqref{MRLagrangian}, physical states must be 
invariant under a $\mathbb{Z}_2$ internal symmetry.
This neutrality requirement introduces correlations between the bosonic and fermionic components of a physical state.
To see how this works, we first observe that the local electronic operator $(\psi_{\rm el})_i^\sigma=\chi^\sigma_ie^{im\phi^\sigma_i}$ is neutral under the following $\mathbb{Z}_2$ transformation, which is local to a given cylinder $X_i$: \begin{align}\mathbb{Z}_2(i):\quad\chi^\sigma_{i'}\to(-1)^{\delta_{ii'}}\chi^\sigma_{i'},\quad\phi^\sigma_{i'}\to\phi^\sigma_{i'}+\frac{i\pi\sigma}{m}\delta_{ii'}.
\label{internalZ2symm}
\end{align} 
Consequently, any integral combination of electron operators, such as a Wilson string that creates a conjugate pair of anyons on the two ends of $X_i$, must be even under the local $\mathbb{Z}_2$ symmetry. 
Assuming there are no bulk excitations inside any of the cylinders, the artificially extended Hilbert space in which the bosons and fermions are decoupled where \eqref{MRLagrangian} acts must be restricted to the physical Hilbert space that is invariant under all $\mathbb{Z}_2(i)$ symmetries. 
The restriction can be achieved by the projection operator $\mathcal{P}=\prod_i\mathcal{P}_{X_i}$, where \begin{align}\mathcal{P}_{X_i}=\frac{1}{2}\left(1+(-1)^{N^{LX_i}+N^{RX_i}}(-1)^{F^{LX_i}+F^{RX_i}}\right)\label{cylinderprojection}\end{align} is the projection operator for cylinder $X_i$ that ensures the corresponding edge states are even under $\mathbb{Z}_2(i)$. 
Here $N^{\sigma X_i}$ is the winding number defined in \eqref{Laughlinwindingnumberdef} and $F^{\sigma X_i}$ measures the fermion parity of a state.
In particular, $(-1)^{F^{\sigma X_i}} \chi_i^\sigma = - \chi_i^\sigma (-1)^{F^{\sigma X_i}}$.

Now consider ``gluing" the right edge of cylinder $X_{i-1}$ to the left edge of cylinder $X_i$.
The strategy is similar to that of the Laughlin case.
In the absence of any coupling, the edge modes are described by the free, decoupled Hamiltonians associated to \eqref{MRLagrangian}.
The cylinders can be pieced together at the interfaces by the electron tunneling terms, \begin{align}
\begin{split}
H^{(1)}_{i} &=-\frac{2g}{2\pi}\int_0^Ldx \left\{\left(\psi^L_{i}\right)^\dagger\psi^R_{i-1}+h.c.\right\}\\
&=-\frac{2g}{\pi}\int_0^L dx \Big\{i\chi^L_{i}\chi^R_{i-1} \cos\left[m\left(\phi^R_{i-1}+\phi^L_{i}\right)\right]\Big\}, \end{split}
\end{align} 
where the coupling constant $g>0$ is taken to be independent of the specific interface $i$.
We treat the tunneling term in a mean-field approximation \cite{sohal-2020} in which the corresponding ground state expectation values (up to $\mathbb{Z}_2$ symmetry) of the bosonic and fermionic operators are
\begin{align}\begin{split}&\langle m(\phi^L_{i}+\phi^R_{i-1})\rangle = \mbox{0 mod $2\pi$ and }\langle i\chi^L_{i}\chi^R_{i-1}\rangle>0,\\\mbox{or }&\langle m(\phi^L_{i}+\phi^R_{i-1})\rangle = \mbox{$\pi$ mod $2\pi$ and }\langle i\chi^L_{i}\chi^R_{i-1}\rangle<0.
\end{split}\end{align} 
The overall scale of the expectation value of the fermion bilinear is absorbed into $g$.
In the $g\rightarrow \infty$ limit, we once again employ the quadratic approximation to the sine-Gordon potential and pin the bosonic fields at the corresponding minima. 
This allows only neutral charge $(N^{RX_{i-1}}+N^{LX_i} =0)$ at the interfaces. 
With these approximations, the tunneling potential becomes 
\begin{align}
\begin{split} H^{(1)}_{i}&=\int_0^L dx \Big[
\frac{v_c\lambda\pi}{2}\big(\phi^R_{i-1}+\phi^L_{i}\big)^2\\
&\;\;\;\;+v_m\tilde{g}i\chi^L_{i}\chi^R_{i-1}+\mathrm{const.}+\ldots\Big], \end{split}
\end{align} where $\tilde{g}=-\frac{2g}{v_m\pi}<0$ and $\lambda>0$. 
The ellipsis denotes higher-order terms which can be ignored as $g \rightarrow \infty$.

It remains to construct torus ground state of this simplified model.
We treat the untwisted and twisted sectors in turn.

\subsubsection{Untwisted Sector}

We construct the ground state of the quadratic Hamiltonian discussed in the previous section and then project the result to the physical Hilbert space.
Since the bosonic zero and oscillation mode Hamiltonians are the same as in the Laughlin case, the bosonic parts of the unprojected ground state are given in Eqs.~\eqref{GSLzm} and \eqref{GSLosc}.
The Hamiltonian for the fermionic oscillation modes is 
\begin{align}
H_{i,f}^{\mathrm{osc}} &= v_m\sum_{k>0}\begin{bmatrix}
(c^R_{i-1,k})^\dagger & c^L_{i,-k}
\end{bmatrix}\begin{bmatrix}
k & -i\tilde{g} \\ i\tilde{g} & -k
\end{bmatrix}\begin{bmatrix}
c^R_{i-1,k} \\ (c^L_{i,-k})^\dagger
\end{bmatrix}
\end{align} 
where $k=2\pi(j+1/2)/L$ with $j$ a non-negative integer. 
For suitable $\varphi_k$, the following transformation, 
\begin{align}
\begin{bmatrix}\tilde{\beta}_{i,k} \\ (\tilde{\gamma}_{i,k})^\dagger
\end{bmatrix} =\begin{bmatrix}
\cos\varphi_k & -i\sin\varphi_k\\
\sin\varphi_k & i\cos\varphi_k
\end{bmatrix}\begin{bmatrix}
c^R_{i-1,k} \\ (c^L_{i,-k})^\dagger
\end{bmatrix},
\end{align} 
diagonalizes the Hamiltonian to 
\begin{align}
H^{\mathrm{osc}}_{i,f} = v_m\sum_{k>0}\left((\tilde{\beta}_{i,k})^\dagger\tilde{\beta}_{i,k}+ (\tilde{\gamma}_{i,k})^\dagger\tilde{\gamma}_{i,k}-1\right).
\end{align} 
We take $\cos 2\varphi_k= k/\epsilon_k$, $\sin2\varphi_k=\tilde{g}/\epsilon_k$, and $\epsilon_k=\sqrt{k^2+\tilde{g}^2}$. 
The ground state is given by the BCS coherent state,
\begin{align}
|f^{\mathrm{osc}}_i\rangle =\prod_{k>0} e^{-i\Xi_k (c^L_{i,-k})^\dagger(c^R_{i-1,k})^\dagger}|0\rangle, 
\end{align} 
where $\tilde{\beta}_{i,k}|f^{\mathrm{osc}}_i\rangle = \tilde{\gamma}_{i,k}|f^{\mathrm{osc}}_i\rangle=0$, and $c^L_{i-1,-k}|0\rangle=c^R_{i,k}|0\rangle=0$ for all $k>0$. 
In the limit of $|k| \ll |\tilde{g}|$, $\Xi_k=\tan\varphi_k\approx \tilde{v}_e k/2$ with $\tilde{v}_e\equiv2/|\tilde{g}|$ the ``entanglement velocity" in the fermionic sector. 
Similar to $|b_i^{\rm osc} \rangle$ in \eqref{GSLosc}, the ground state can be rewritten as 
\begin{widetext}
\begin{align}
        \label{GSFosc}
        |f^{\mathrm{osc}}_{i} \rangle & = \sum_{\{\tilde n_{i,k} \in\mathbb{Z}_2\}} i^{\sum_{k>0} \tilde n_{i, k}} e^{-\sum_{k>0} {\tilde v_e k \over 2} (\tilde n_{i,k}+ 1/2)} |\{n^{RX_{i-1}}_{f, k} = \tilde n_{i,k} \}_{k>0} \rangle_{RX_{i-1}} \otimes |\{n^{LX_{i}}_{f, - k} = \tilde n_{i,k} \}_{k>0} \rangle_{LX_{i}}.
    \end{align}
    \end{widetext}
where $n^{RX_{i-1}}_{f,k}$ and $n^{LX_{i}}_{f,-k}$ are the eigenvalues of the fermion number operators $(c^R_{i-1,k})^\dagger c^R_{i-1,k}$ and $(c^L_{i,-k})^\dagger c^L_{i,-k}$.  

Because the zero mode and oscillation modes are decoupled (in the artificially extended Hilbert space), the torus ground state for the approximated Hamiltonian can be written as a tensor product of \eqref{GSLzm}, \eqref{GSLosc}, and \eqref{GSFosc}:  
\begin{align}
|\hat{\Psi}_a\rangle &=\bigotimes_i|b^{\mathrm{zero}}_{r,i}\rangle\otimes|b^{\mathrm{osc}}_i\rangle\otimes|f^{\mathrm{osc}}_i\rangle\label{ApproxUntwistedGS}.
\end{align} The corresponding physical ground state that is invariant under the internal $\mathbb{Z}_2$ symmetry \eqref{internalZ2symm} is the projection:  \begin{align}\left|\Psi_a\right\rangle&=\mathcal{P}|\hat{\Psi}_a\rangle\label{untwistedMRGS}\\&=\bigotimes_iP_{a,i}|b^{\mathrm{zero}}_{r,i}\rangle\otimes|b^{\mathrm{osc}}_i\rangle\otimes|f^{\mathrm{osc}}_i\rangle,\nonumber\end{align}
where the projection operator $\mathcal{P}$ is given in \eqref{cylinderprojection}.
In the untwisted sector, the projection operator for each cylinder $X_i$ decomposes into the product of left and right edge projection operators $P_{a, i} P_{a, i+1}$ given by
\begin{align}P_{a,i}=\frac{1}{2}\left(1+(-1)^{N_{a,i}-r/m+\sum_{k>0}\tilde{n}_{i,k}}\right).
\end{align} 
(Note that $\tilde{n}_{i,k}$ denotes one of the fermion number operators $(c^R_{i-1,k})^\dagger c^R_{i-1,k}$ or $(c^L_{i,-k})^\dagger c^L_{i,-k}$, whose eigenvalues coincide at interface $i$.)
These operators restrict the winding number and fermion parity of the ground state at interface $i$ between $X_{i-1}$ and $X_i$.

\subsubsection{Twisted Sector}

In the twisted sector, we must include the Majorana zero mode excitations, which arise from the $\pi$ flux that threads across all cylinders and results in fermionic momenta that are integrally quantized as $k=2\pi j/L$.
The contributions of the bosonic modes and fermion oscillator modes to the unprojected torus state have the same form as before and so we need only discuss the novelty presented by the Majorana zero modes.

The Majorana zero mode Hamiltonian is 
\begin{align}
H^{\mathrm{zero}}_f = i\tilde{g}\sum_i c^{R}_{i-1,0}c^{L}_{i,0}=\tilde{g}\sum_i f_i^\dagger f_i,
\label{Hfzero}
\end{align}
where $\tilde{g}>0$.
This Hamiltonian is essentially the Kitaev chain~\cite{2001PhyU...44..131K} with quantum states labeled by the eigenvalues $n'_i=0,1$ of the fermion number operators $f_i^\dagger f_i$ at the interface between cylinders $X_{i-1}$ and $X_i$. 
Here, $f_i=(c^R_{i-1,0}+ic^L_{i,0})/\sqrt{2}$ is an interface Dirac fermion. 
Suppose the torus is divided into four consecutive cylinders $X_1 \cup X_2 \cup X_3 \cup X_4$.
Then the ground state of \eqref{Hfzero} is  $|f^{\mathrm{zero}}\rangle=|0'0'0'0'\rangle$, where the primes refer to the interface basis states. 

Because the $\mathbb{Z}_2$ projection operator is not diagonal with respect to this interface basis, we need to change to an appropriate cylinder basis for the Majorana zero modes. 
To this end, we define the cylinder Dirac fermions $d_i=(c^R_{i,0}+ic^L_{i,0})/\sqrt{2}$ on $X_i$ and the corresponding occupation numbers $\gamma_{i}=0,1$ of the operators $d_i^\dagger d_i$.
Notice that $f_i^\dagger f_i$ and $d_{i'}^\dagger d_{i'}$ do not commute when $i=i'$ or $i=i'+1$.
The $F$-symbols \cite{kiaev-2006} generate the basis transformation between cylinder and interface bases. 
This basis change is depicted in Fig.~\ref{fusiondiagram}. 
\begin{figure}[t]
\centering
\includegraphics[width=0.48\textwidth]{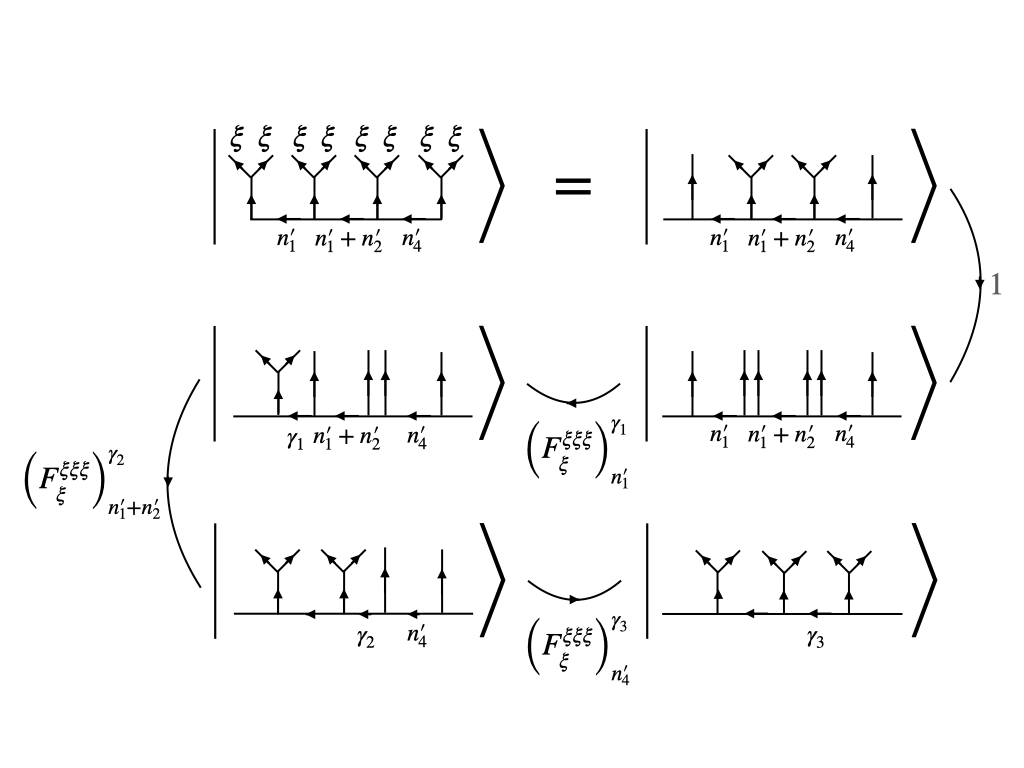}
\caption{Basis Change. Three $F$-moves are used to transform from the interface to cylinder bases.
Each $|n'_i = 0 \rangle$ in $|0'0'0'0'\rangle$ obtains by fusing two non-Abelian twist fields $\xi$ into the vacuum channel; fusion into the $\chi$ channel is denoted by ``1." The interface fermion numbers generally satisfy $n'_1 + n'_2 + n'_3 + n'_4 = 0$ (mod 2).}
\label{fusiondiagram}
\end{figure} 
For Ising topological order, the relevant $F$-move transformation is given by the $2 \times 2$ matrix,
\begin{align}
\left(F^{\xi\xi\xi}_\xi\right)_\mu^\nu=\frac{1}{\sqrt{2}}(-1)^{\mu\nu}\label{IsingFsymbol}
\end{align} 
where $\mu,\nu \in \{0,1 \}$, corresponding to the two possible fusion channels of $\xi$.
Thus, the F-moves transform $|0'0'0'0' \rangle$ to a basis written in terms of states labeled by the fusion channels of pairs of $\xi$ belonging to a particular cylinder.
The index $\mu$ is the original internal channel and $\nu$ the new internal channel after the $F$-move.
Thus, to transform the ground state $|n'_1=0,n'_2=0,n'_3=0,n'_4=0\rangle$, we use
\begin{align}
&|n_1'n_2'n_3'n_4'\rangle \nonumber\\&= \sum_{\gamma_1,\gamma_2,\gamma_3} \left(F_\xi^{\xi\xi\xi}\right)^{\gamma_3}_{n'_4}\left(F_\xi^{\xi\xi\xi}\right)^{\gamma_2}_{n'_1+n'_2} \left(F_\xi^{\xi\xi\xi}\right)^{\gamma_1}_{n'_1}\nonumber\\
&\;\;\;\;\times|\gamma_1\gamma_2\gamma_3\gamma_4\rangle,\label{general4cylinderequation}
\end{align}
where $\gamma_4 = 1+\gamma_1 + \gamma_2 + \gamma_3$ (mod 2),
with the result:
\begin{align}
|f^{\mathrm{zero}}\rangle
&=\frac{1}{\sqrt{8}}(|0001\rangle+|0010\rangle+|0100\rangle+|0111\rangle\nonumber\\&\;\;\;\;+|1000\rangle+|1011\rangle+|1101\rangle+|1110\rangle).
\end{align} 
Each of the un-primed states in $|f^{\rm zero}\rangle$ is an eigenstate of $d_i^\dagger d_i$. 
For example, $|\gamma_1 \gamma_2 \gamma_3 \gamma_4 \rangle$ has eigenvalue $(-1)^{\gamma_i}$ under $d^\dagger_i d_i$.
Here, the total fermion parity in the cylinder basis, $\sum_i\gamma_i$, is odd, while the total parity in the interface basis, $\sum_in'_i$, is even. This is because the two total parities are exactly opposite, \begin{align}\prod_i(-1)^{f_i^\dagger f_i}&=(-1)^M2^{2M}\prod_ic^R_{i-1,0}c^L_{i,0}\nonumber\\&=-(-1)^M2^{2M}\prod_ic^R_{i,0}c^L_{i,0}\nonumber\\&=-\prod_i(-1)^{d_i^\dagger d_i}.\label{totalparitycomparison}\end{align}

For $2M$ cylinders, this result generalizes to
\begin{align}
\label{twistedzero}
 | f^{\rm zero} \rangle = \frac{1}{\sqrt{2^{2M-1}}}\sum_{\vec \gamma \in \{0,1\}^{2M}|_C}|\gamma_1 \gamma_2 \cdots \gamma_{2M}\rangle.
\end{align}
Here, $\{0,1 \}^{2M}|_C$ indicates that $\vec \gamma$ takes values in $\{0,1\}^{2M}$ subject to the constraint $\sum_{i=1}^{2M}\gamma_i \equiv 1\  ({\rm mod}\ 2)$;
the overall normalization comes from the fact that there are $2^{2M-1}$ solutions to this constraint.
Physically, this constraint on $\vec \gamma$ means the overall topological charge of the $2M$ Majorana fermions on the torus is in the vacuum channel.
Similar to Eq.~\eqref{general4cylinderequation}, it will sometimes be convenient to take the sum in \eqref{twistedzero} to be over unconstrained $\gamma_i$ for $i \in \{1, \ldots, 2M-1\}$ with $\gamma_{2M}$ implicitly determined by the constraint.

Thus, the unprojected ground state for the twisted sector is \begin{align}
|\hat{\Psi}_a\rangle &=\left(\bigotimes_i|b^{\mathrm{zero}}_{r,i}\rangle\otimes|b^{\mathrm{osc}}_i\rangle\otimes|f^{\mathrm{osc}}_i\rangle\right)\otimes|f^{\mathrm{zero}}\rangle.\label{ApproxtwistedGS}\end{align} 
Here $a$ refers to the Ising twist field $\xi^{r+1/2}=\xi e^{i(r+1/2)\phi}$, for $r=0,1,\ldots,m-1$, and therefore the winding number $N_{r,i}$ of $|b^{\mathrm{zero}}_{r,i}\rangle$ takes values in $(r+1/2)/m+\mathbb{Z}$. 
The physical ground state that is invariant under the internal $\mathbb{Z}_2$ symmetry \eqref{internalZ2symm} is the projection: \begin{align}|\Psi_a\rangle=\mathcal{P}|\hat{\Psi}_a\rangle=\prod_i P_{a,X_i}|\hat{\Psi}_a\rangle\label{twistedMRGS}\end{align} where the projection operator $P_{a,X_i}$ on cylinder $X_i$ is defined by 
\begin{align}
\label{eq:twistSecProjector0}
P_{a,X_i}\nonumber&=\frac{1}{2}\left(1+(-1)^{N_{a,X_i}+\sum_{k>0}(\tilde{n}_{X_i,k})+d_i^\dagger d_i+u_i}\right),\nonumber\\
N_{a,X_i}&\equiv -N_{a,i}+N_{a,i-1}, \quad \tilde{n}_{X_i,k}\equiv\tilde{n}_{i-1,k}+\tilde{n}_{i,k},\nonumber\\
u_i&=\delta_{2M,i}=\left\{\begin{array}{*{20}l}0,&\mbox{if $i=1,\ldots,2M-1$}\\1,&\mbox{if $i=2M$}\end{array}\right..
\end{align}
Here, the additional $u_{2M}$ accounts for the total odd parity in the zero mode sector $\prod_i(-1)^{d_i^\dagger d_i}=(-1)^{\sum_iu_i}=-1$ [see \eqref{totalparitycomparison}].

\section{Entanglement Negativity}
\label{sectionthree}

We now study the entanglement negativity of the Laughlin and Moore-Read states at filling fraction $\nu = 1/m$ on the torus, constructed in the previous section.

\subsection{\texorpdfstring{$\nu = 1/m$}{v=1/m} Laughlin State}

\subsubsection{Torus Geometry}

We begin with the Laughlin state at filling fraction $\nu = 1/m$ and the torus geometry [e.g., Fig.~\ref{4cylinders}(a)].
The unnormalized torus ground state in sector $a \in \{0, \ldots, m-1\}$ factorizes as
\begin{align}
    |\Psi_a \rangle = \bigotimes_{i=1}^{2M} | \Psi_{a, i}\rangle,
\end{align}
where $i$ refers to the interface between cylinders $X_{i-1}$ and $X_i$ and 
\begin{align}
    | \Psi_{a,i} \rangle = |b^{\mathrm{zero}}_{a, i} \rangle \otimes |b^{\mathrm{osc}}_{i} \rangle.
    \end{align}
    The bosonic zero mode $|b^{\mathrm{zero}}_{a, i}\rangle$ and oscillator $|b^{\mathrm{osc}}_{i} \rangle$ states are given in Eqs.~\eqref{GSLzm} and \eqref{GSLosc}.
    Introducing the collective mode numbers,
    \begin{align}
        {\cal N}^{LX_{i}} & \equiv \big(-N^{LX_{i}}, \{n^{LX_{i}}_{b,-k} \}_{k>0} \big), \\
        {\cal N}^{RX_{i}} & \equiv \big(N^{RX_{i}}, \{n^{RX_{i}}_{b,k} \}_{k>0} \big), \\
        \label{generalN}
        {\cal N}_{a,i} & \equiv \big(N_{a,i}, \{ n_{i,k} \}_{k>0} \big),
    \end{align}
  with domains defined in Eqs.~\eqref{GSLzm} and \eqref{GSLosc},
    we write 
    \begin{widetext}
    \begin{align}
        | \Psi_{a, i} \rangle = \sum_{{\cal N}_{a,i}} \lambda({\cal N}_{a,i}) |{\cal N}^{RX_{i-1}} = {\cal N}_{a,i} \rangle_{RX_{i-1}} \otimes |{\cal N}^{LX_{i}} = {\cal N}_{a,i} \rangle_{LX_{i}},
    \end{align}
    \end{widetext}
    where
      \begin{align}
      \label{laughlinhalfeigenvalues}
        \lambda({\cal N}_{a,i}) = \exp\left[- {v_e \pi m \over 2 L} N_{a,i}^2  - \sum_{k>0} {v_e k \over 2} \left(n_{i,k} + \frac{1}{2}\right) \right].
    \end{align}

Assembling the preceding together, we have
    \begin{align}\label{eq-psi-a-4cylin}
        |\Psi_a \rangle = \bigotimes_{i=1}^{2M} \sum_{{\cal N}_{a,i}} \lambda({\cal N}_{a,i}) |{\cal N}_{a,i} \rangle_{RX_{i-1}} \otimes |{\cal N}_{a,i} \rangle_{LX_{i}}.
    \end{align}
    Equation \eqref{eq-psi-a-4cylin} shows how a product of cylinder states glue together to form the unnormalized torus state in sector $a$.
      The norm-squared of $|\Psi_a \rangle$ is
    \begin{align}
        (Z_a)^{2M} = \left(\sum_{{\cal N}_a} \lambda^2({\cal N}_a)    \right)^{2M},
    \end{align}
    with mode number ${\cal N}_a$ defined as in \eqref{generalN}.
    We identify $Z_a$ as the partition function in sector $a$ at inverse ``temperature" $\beta = 1$ of the entanglement Hamiltonian $H_a$,
    \begin{align}
    \label{laughlinsectorapartitionfunction}
        Z_a(\beta) = \tr e^{- \beta H_a},
    \end{align}
    with entanglement spectrum equal to $-2 \log \lambda({\cal N}_a)$.

  We use Lemma \ref{lemma1} below to calculate the entanglement negativity of the general torus state $|\Psi \rangle = \sum_a \psi_a |\overline{\Psi}_a \rangle$ with respect to the torus partition [e.g.,  Fig.~\ref{4cylinders}(a)], where the normalized sector $a$ state is
    \begin{align}
        | \overline{\Psi}_a \rangle = Z_a^{-M} |\Psi_a \rangle.
    \end{align}
    While the details of the proof of the analogous lemma for the non-Abelian Moore-Read state differs slightly due to the presence of fermionic zero modes (in the twisted sector), it turns out that the result of Lemma 1 continues to apply. 
    Readers uninterested in the details of the straightforward, but tedious proof of Lemma 1, may safely skip it and use the result. 
       \begin{lemma}
  \label{lemma1}
  The entanglement negativity of $\rho = |\Psi \rangle \langle \Psi |$ with respect to the torus partition $X = X_{\rm odd} \cup X_{\rm even}$ (e.g., Fig.~\ref{4cylinders}) and with partial transposition on odd cylinders $X_{\rm odd} = X_1 \cup X_3 \cup \cdots \cup X_{2M - 1}$ equals
  \begin{align}
  \label{resultoflemma1}
  {\cal E}_{X_{\rm odd}:X_{\rm even}} = 2 \log \sum_a |\psi_a| \left( {Z_a(1 /2) \over \sqrt{Z_a(1)}}\right)^{2M},
  \end{align}
  where $X_{\rm even} = X_2 \cup X_4 \cup \cdots \cup X_{2M}$ and $Z_a(\beta)$ is defined in \eqref{laughlinsectorapartitionfunction}.
  \end{lemma}
  \begin{proof}
  We directly evaluate $|| \rho^{T_{{\rm odd}}} ||_1 = \tr \sqrt{(\rho^{T_{{\rm odd}}})^\dagger \rho^{T_{{\rm odd}}}}$ to compute ${\cal E}_{X_{\rm odd}: X_{\rm even}} = \log || \rho^{T_{{\rm odd}}} ||_1$, where $\rho^{T_{{\rm odd}}}$ denotes the partial transpose of $\rho$ with respect to $X_{\rm odd}$.
  Define $\{\vec{{\cal N}}_a \} \equiv \{(\mathcal{N}_{a,1},\mathcal{N}_{a,2},\ldots,\mathcal{N}_{a,2M})\}$ with ${\cal N}_{a,i}$ in \eqref{generalN} and $c_a(\vec{{\cal N}}_a) = \psi_a \prod_{i = 1}^{2M} \lambda({\cal N}_{a, i})/\sqrt Z_a$ (no sum over $a$) with $\lambda({\cal N}_{a,i})$ given in \eqref{laughlinhalfeigenvalues}; we will sometimes denote $c_a = c_a(\vec{{\cal N}}_a)$ for brevity.
  Then we may write
\begin{align}
\begin{split}
    \rho&= \sum_{a,a'} \sum_{\vec{{\cal N}}_a, \vec{{\cal N}}'_{a'}} c^\ast_a(\vec{{\cal N}}_{a})c_{a'}(\vec{{\cal N}}'_{a})|\mathcal{N}'_{a',1}\mathcal{N}'_{a',2}\rangle\langle \mathcal{N}_{a,1}\mathcal{N}_{a,2}|_{X_1}\\
    &\;\;\;\;\otimes\ldots\\
    &\;\;\;\;\otimes 
    |\mathcal{N}'_{a',2M-1}\mathcal{N}'_{a',2M}\rangle\langle \mathcal{N}_{a,2M-1}\mathcal{N}_{a,2M}|_{X_{2M-1}}\\
    &\;\;\;\;\otimes 
    |\mathcal{N}'_{a',2M}\mathcal{N}'_{a',1}\rangle\langle \mathcal{N}_{a,2M}\mathcal{N}_{a,1}|_{X_{2M}}.
\end{split}
\end{align} 
Note that $|\mathcal{N}'_{a',i}\mathcal{N}'_{a',i+1}\rangle\langle \mathcal{N}_{a,i}\mathcal{N}_{a,i+1}|_{X_i}$ denotes the outer product of states on the edges of cylinder $X_i$: the first entry of each ket or bra refers to states on the left edge of $X_i$, while the second entry refers to states on the right edge of $X_i$.
Taking the partial transpose with respect to $X_{\rm odd}$, we have 
\begin{align}
    \rho^{T_{{\rm odd}}}&= \sum_{a,a'} \sum_{\vec{{\cal N}}_a, \vec{{\cal N}}'_{a'}} c^\ast_a c_{a'} |\mathcal{N}_{a,1}\mathcal{N}_{a,2}\rangle\langle\mathcal{N}'_{a',1}\mathcal{N}'_{a',2}|_{X_1}\cr
    &\;\;\;\;\otimes\ldots\cr 
    &\;\;\;\;\otimes|\mathcal{N}'_{a',2M-2}\mathcal{N}'_{a',2M-1}\rangle\langle \mathcal{N}_{a,2M-2}\mathcal{N}_{a,2M-1}|_{X_{2M-2}}\cr
    &\;\;\;\;\otimes 
    |\mathcal{N}_{a,2M-1}\mathcal{N}_{a,2M}\rangle\langle \mathcal{N}'_{a',2M-1}\mathcal{N}'_{a',2M}|_{X_{2M-1}}\cr
    &\;\;\;\;\otimes 
    |\mathcal{N}'_{a',2M}\mathcal{N}'_{a',1}\rangle\langle \mathcal{N}_{a,2M}\mathcal{N}_{a,1}|_{X_{2M}}.
\end{align} 
 Next we evaluate
 \begin{widetext}
 \begin{align}
   (\rho^{T_{\text{odd}}})^\dagger \rho^{T_{\rm odd}} & = \sum_{a,a', a^{''}, a^{'''}} \sum_{\vec{{\cal N}}_a, \vec{{\cal N}}'_{a'}, \vec{{\cal N}}^{''}_{a^{''}}, \vec{{\cal N}}^{'''}_{a^{'''}}} c^\ast_{a^{'''}} c_{a^{''}} c^\ast_a c_{a'} \cr 
    &\;\;\;\;\times
    |\mathcal{N}'''_{a''',1}\mathcal{N}'''_{a''',2}\rangle\langle\mathcal{N}''_{a'',1}\mathcal{N}''_{a'',2}|\mathcal{N}_{a,1}\mathcal{N}_{a,2}\rangle\langle\mathcal{N}'_{a',1}\mathcal{N}'_{a',2}|_{X_1}\cr
    &\;\;\;\;\otimes\ldots\cr
    &\;\;\;\;\otimes
    |\mathcal{N}''_{a'',2M-2}\mathcal{N}''_{a'',2M-1}\rangle\langle\mathcal{N}'''_{a''',2M-2}\mathcal{N}'''_{a''',2M-1}|\mathcal{N}'_{a',2M-2}\mathcal{N}'_{a',2M-1}\rangle\langle\mathcal{N}_{a,2M-2}\mathcal{N}_{a,2M-1}|_{X_{2M-2}}\cr
    &\;\;\;\;\otimes
    |\mathcal{N}'''_{a''',2M-1}\mathcal{N}'''_{a''',2M}\rangle\langle\mathcal{N}''_{a'',2M-1}\mathcal{N}''_{a'',2M}|\mathcal{N}_{a,2M-1}\mathcal{N}_{a,2M}\rangle\langle\mathcal{N}'_{a',2M-1}\mathcal{N}'_{a',2M}|_{X_{2M-1}}\cr
    &\;\;\;\;\otimes
    |\mathcal{N}''_{a'',2M}\mathcal{N}''_{a'',1}\rangle\langle\mathcal{N}'''_{a''',2M}\mathcal{N}'''_{a''',1}|\mathcal{N}'_{a',2M}\mathcal{N}'_{a',1}\rangle\langle\mathcal{N}_{a,2M}\mathcal{N}_{a,1}|_{X_{2M}}.
     \end{align}
 \end{widetext}
Using the orthonormality of states with different quantum numbers in the above overlaps,
\begin{align}
\begin{split}
  a''' & = a', \quad  \mathcal{N}'''_{a''',i} =\mathcal{N}'_{a',i}, \\
  a'' & = a, \quad \mathcal{N}''_{a'',i} =\mathcal{N}_{a,i},
\end{split}
\end{align} 
for all $1 \leq i \leq 2M$, we find that $(\rho^{T_{{\rm odd}}})^\dagger \rho^{T_{{\rm odd}}}$ is diagonal with entries given by $|c_{a'}(\Vec{\mathcal{N}}'_{a'})|^2 |c_a(\vec{\mathcal{N}}_a)|^2$. 
Thus, 
\begin{align}
    \begin{split}
        \sqrt{(\rho^{T_{\text{odd}}})^\dagger \rho^{T_{\rm odd}}} & =
        \sum_{a,a'}\sum_{\vec{\mathcal{N}}_{a},\vec{\mathcal{N}}'_{a'}}
        |c_{a'}||c_a|\\
        &\;\;\;\;\times |\mathcal{N}'_{a',1} \mathcal{N}'_{a',2}\rangle\langle\mathcal{N}'_{a',1} \mathcal{N}'_{a',2}|_{X_1}\\
        &\;\;\;\;\otimes |\mathcal{N}_{a,2} \mathcal{N}_{a,3}\rangle\langle\mathcal{N}_{a,2} \mathcal{N}_{a,3}|_{X_2}\\
        &\;\;\;\;\otimes\cdots\\
        &\;\;\;\;\otimes|\mathcal{N}_{a,2M} \mathcal{N}_{a,1}\rangle\langle\mathcal{N}_{a,2M} \mathcal{N}_{a,1}|_{X_{2M}}
        \label{extrastep}
    \end{split}
\end{align}
and
$\tr\sqrt{(\rho^{T_{\rm odd}})^\dagger\rho^{T_{\rm odd}}}=\sum_{a,a'}\sum_{\vec{\mathcal{N}}_{a},\vec{\mathcal{N}}'_{a'}}|c_{a'}||c_a|$.

Tracing through our definitions, we find
\begin{align}
    \begin{split}
    ||\rho^{T_{\rm odd}}||_1
     &= \Big( \sum_{a} \sum_{\vec{{\cal N}}_a} |c_a(\vec{{\cal N}}_a)| \Big)^2 \cr
     & =\left(\sum_a|\psi_a| \sum_{\vec{{\cal N}}_a} \prod_{i = 1}^{2M} {\lambda({\cal N}_{a,i}) \over \sqrt{Z_a}} \right)^2 \cr
    &=\left(\sum_a|\psi_a|\left(\frac{\sum_{\mathcal{N}_{a}}\lambda(\mathcal{N}_{a})}{\sqrt{\sum_{\mathcal{N}_{a}}\lambda^2(\mathcal{N}_{a})}}\right)^{2M}\right)^2 \cr
    &=\left(\sum_a|\psi_a|\left(\frac{Z_a(1/2)}{\sqrt{Z_a(1)}}\right)^{2M}\right)^2.
    \end{split}\label{tracenormlemma1}
\end{align}
Taking the logarithm of $||\rho^{T_{\rm odd}}||_1$, we obtain \eqref{resultoflemma1} and thereby complete the proof.
\end{proof}

   It remains to calculate $Z_a(\beta)$.
   We are specifically interested in the $L \rightarrow \infty$ limit.
   The partition function of the Laughlin edge states in sector $a$ can be written as
    \begin{align}
       Z_{a}(\beta) = \frac{\theta^{-a/m}_0(m\tau)}{\eta(\tau)}, \quad \tau=i\tau_2=\frac{i\beta v_e}{L},\label{DefineTau}
   \end{align}
where the Jacobi $\theta$ and Dedekind $\eta$ functions (see Appendix) are
\begin{align}
    \theta^{-a/m}_0(m\tau) & = \sum_{N_a\in \mathbb{Z}-\frac{a}{m}}e^{-\frac{\beta v_e\pi m}{L} N_a^2} , \\
    &=\sum_{n\in\mathbb{Z}}q^{\frac{1}{2}(n-\frac{a}{m})^2},\\
    \eta^{-1}(\tau) & = \sum_{n_{i,k}\in \mathbb{Z}^+}e^{-\sum_{k>0}\beta v_ek(n_{i,k}+1/2)}\\
    &=q^{-\frac{1}{24}}\prod_{n_{i,k}\in\mathbb{Z}^+}(1-q^{n_{i,k}})^{-1},
\end{align} 
and $q=e^{2\pi i \tau}$.
These functions have a useful transformation under the modular transformation $\tau \mapsto - 1/\tau$ that allows us to easily extract the scaling behavior of the entanglement negativity as $L \rightarrow \infty$.
Specifically,
\begin{align}
    \theta^{-a/m}_0(m\tau) & = \frac{\theta^0_{-a/m}(-1/m\tau)}{\sqrt{-im\tau}} , \\ 
    \eta(\tau)& = {\eta(-1/\tau) \over \sqrt{-i \tau}}.
\end{align}
Thus, $Z_a(\beta) = \frac{1}{\sqrt{m}}\exp\left[\frac{\pi L}{12\beta v_e}\right]$ and so
\begin{align}
    \frac{Z_a(1/2)}{\sqrt{Z_a(1)}} = \frac{1}{\sqrt[4]{m}}e^{\frac{\pi L}{8v_e}}.
\end{align}
Inserting this expression into \eqref{resultoflemma1} and taking the $L\rightarrow \infty$ limit, we find the topological entanglement negativity
\begin{align}
    {\cal E}_{X_{\rm odd}: X_{\rm even}} = M \left({\pi \over 2 v_e}\right) L - M \log m + 2 \log \sum_a |\psi_a|.
\end{align}
We see that ${\cal E}_{X_{\rm odd}:X_{\rm even}}$ receives $2M$ contributions, proportional to $\log \sqrt m$, and a single topological sector correction, equal to $2 \log \sum_a |\psi_a|$. 
Since the Laughlin phase has only Abelian quasiparticles ($d_a = 1$), ${\cal E}_{X_{\rm odd}:X_{\rm even}}$ takes the form given in \eqref{torusresult} with $\alpha = {\pi \over 2 v_e}$ and ${\cal D} = \sqrt{m}$.

\subsubsection{Cylinder Geometry}
\label{cylinderlaughlinsection}

Next we consider the entanglement negativity between subsets of $X_{\rm even}$ and $X_{\rm odd}$ when the degrees of freedom on $N \leq M$ cylinders $\bar{Y} \subset X$ have been traced over.
We denote the remaining $(2M - N)$ cylinders by $Y$ and their decomposition into ``odd" and ``even" cylinders as $Y_{\rm odd}$ and $Y_{\rm even}$.
The resulting entanglement negativity will depend on the number $R$ of shared interfaces between the remaining cylinders in $Y_{\rm odd} \cup Y_{\rm even}$.
As an example, Fig.~\ref{4cylinders}(b) represents $\Tr_{X_4} |\Psi \rangle \langle \Psi |$, i.e., the $X_1 \cup X_2 \cup X_3$ cylinder state when the degrees of freedom on $\bar{Y} = X_4$ have been traced over; we then consider the entanglement negativity of $\Tr_{X_4} |\Psi \rangle \langle \Psi |$ between degrees of freedom on $Y_{\rm odd} = X_1 \cup X_3$ and $Y_{\rm even} = X_2$ with a result that depends on $R = 2$.

Our calculation of the entanglement negativity will apply Lemma \ref{lemma2} below to the torus ground state from the previous section.
This lemma applies to both the Abelian Laughlin and non-Abelian Moore-Read states.
We will summarize the appropriate generalization of its proof in the non-Abelian case in a later section.
   
   \begin{lemma}\label{lemma2}
   Consider the reduced density matrix $\rho_{Y} = \Tr_{\bar{Y}} | \Psi \rangle \langle \Psi |$, where $|\Psi \rangle = \sum_a \psi_a |\overline{\Psi}_a \rangle$ is a general state on the torus $X = Y \cup \bar{Y}$.
   Then the entanglement negativity of $\rho_{Y}$ equals
   \begin{align}
   \label{lemma2result}
       {\cal E}_{Y_{\rm odd}: Y_{\rm even}} = \log \sum_a \left(|\psi_a| \left( {Z_a(1/2) \over \sqrt{Z_a(1)}}\right)^{R}  \right)^2,
   \end{align}
   where $R$ equals the number of interfaces shared between the remaining cylinders $Y = Y_{\rm odd} \cup Y_{\rm even}$ and $Z_a(\beta)$ is defined in \eqref{laughlinsectorapartitionfunction}.
\end{lemma}

\begin{proof} 
We use notation introduced in Lemma \ref{lemma1}.
There are four cases to consider.

(Case I) We remove cylinder $\bar{Y} = X_{2k}$ where $1<2k\leq2M$ by tracing over its left and right edge states.
Thus, $Y_{\rm odd} = X_{\rm odd}$, $Y_{\rm even} = X_2 \cup \cdots \cup X_{2k - 2} \cup X_{2k + 2} \cdots X_{2M}$, and $R = 2M - 2$.

We begin with the torus ground state $|\Psi \rangle \langle \Psi|$:
\begin{align}
    \begin{split}
        \rho &=\sum_{a,a'}\sum_{\vec{\mathcal{N}}_a,\vec{\mathcal{N}}'_{a'}}c_a(\vec{\mathcal{N}}_a) c^*_{a'}(\vec{\mathcal{N}}'_{a'})\times\cdots\\ 
        &\;\;\;\;\otimes 
        |\mathcal{N}_{a,2k-2}\mathcal{N}_{a,2k-1}\rangle\langle \mathcal{N}'_{a',2k-2}\mathcal{N}'_{a',2k-1}|_{X_{2k-2}}\\
        &\;\;\;\;\otimes 
        |\mathcal{N}_{a,2k-1}\mathcal{N}_{a,2k}\rangle\langle \mathcal{N}'_{a',2k-1}\mathcal{N}'_{a',2k}|_{X_{2k-1}}\\
        &\;\;\;\;\otimes 
        |\mathcal{N}_{a,2k}\mathcal{N}_{a,2k+1}\rangle\langle \mathcal{N}'_{a',2k}\mathcal{N}'_{a',2k+1}|_{X_{2k}}\\
        &\;\;\;\;\otimes 
        |\mathcal{N}_{a,2k+1}\mathcal{N}_{a,2k+2}\rangle\langle \mathcal{N}'_{a',2k+1}\mathcal{N}'_{a',2k+2}|_{X_{2k+1}}\\
        &\;\;\;\;\otimes 
        |\mathcal{N}_{a,2k+2}\mathcal{N}_{a,2k+3}\rangle\langle \mathcal{N}'_{a',2k+2}\mathcal{N}'_{a',2k+3}|_{X_{2k+2}}\\
        &\;\;\;\;\otimes \cdots.
    \end{split}
\end{align}
Tracing over degrees of freedom on $X_{2k}$ sets 
\begin{align}
\label{onetraceconstraint}
   a'=a, \quad \mathcal{N}'_{a',2k} = \mathcal{N}_{a,2k}, \quad \mathcal{N}'_{a',2k+1} = \mathcal{N}_{a,2k+1},
\end{align} 
and removes the corresponding outer products involving states on $X_{2k}$.
The first condition above $(a = a')$ removes any ``interference" in $\Tr_{X_{2k}} (\rho)$ between states in different topological sectors. 
Using \eqref{onetraceconstraint} and the definition of $c_a(\vec{{\cal N}}_a)$, the partial transpose of $\Tr_{X_{2k}} (\rho)$ with respect to $Y_{\rm odd}$ is
\begin{align}
    \begin{split}
        \rho^{T_{\rm odd}}_{Y}&=\sum_{a}\sum_{\vec{\mathcal{N}}_a,\vec{\mathcal{N}}'_{a}}c_a(\vec{\mathcal{N}}_a) c^*_{a}(\vec{\mathcal{N}}'_{a})\times\cdots\\ 
        &\;\;\;\;\otimes 
        |\mathcal{N}_{a,2k-2}\mathcal{N}_{a,2k-1}\rangle\langle \mathcal{N}'_{a,2k-2}\mathcal{N}'_{a,2k-1}|_{X_{2k-2}}\\
        &\;\;\;\;\otimes 
        |\mathcal{N}'_{a,2k-1}\mathcal{N}_{a,2k}\rangle\langle \mathcal{N}_{a,2k-1}\mathcal{N}_{a,2k}|_{X_{2k-1}}\\
        &\;\;\;\;\otimes 
        |\mathcal{N}_{a,2k+1}\mathcal{N}'_{a,2k+2}\rangle\langle \mathcal{N}_{a,2k+1}\mathcal{N}_{a,2k+2}|_{X_{2k+1}}\\
        &\;\;\;\;\otimes|\mathcal{N}_{a,2k+2}\mathcal{N}_{a,2k+3}\rangle\langle\mathcal{N}'_{a,2k+2}\mathcal{N}'_{a,2k+3}|_{X_{2k+2}}\\
        &\;\;\;\;\otimes \cdots ,
    \end{split}
\end{align}  
where
\begin{align}
\begin{split}
    c_a(\vec{\mathcal{N}}_a) c^*_{a}(\vec{\mathcal{N}}'_{a}) &= {|\psi_a|^2 \over Z_a^{2M}}\lambda(\mathcal{N}_{a,1})\lambda(\mathcal{N}'_{a,1})\times\ldots\\
    &\;\;\;\;\times\lambda(\mathcal{N}_{a,2k-1})\lambda(\mathcal{N}'_{a,2k-1})\\
    &\;\;\;\;\times\lambda^2(\mathcal{N}_{a,2k})\lambda^2(\mathcal{N}_{a,2k+1}) \\
    &\;\;\;\;\times\lambda(\mathcal{N}_{a,2k+2})\lambda(\mathcal{N}'_{a,2k+2})\times\ldots\\
    &\;\;\;\;\times\lambda(\mathcal{N}_{a,2M})\lambda(\mathcal{N}'_{a,2M}).
    \end{split}\label{ProductCC}
\end{align} 
There is no dependence on ${\cal N}'_{a, 2k}$ or ${\cal N}'_{a, 2k + 1}$ above because of \eqref{onetraceconstraint}.
In what follows, it will be implicitly understood that ${\cal N}'_{a, 2k}, {\cal N}'_{a, 2k +1}$ are removed in sums over $\vec{{\cal N}}'_{a}$.
The remainder of the proof follows that of Lemma \ref{lemma1}. 
Specifically, we compute
\begin{widetext}
\begin{align}
    \begin{split}
    (\rho^{T_{\rm odd}}_{Y})^\dagger \rho^{T_{\rm odd}}_{Y} &= \sum_{a,a'}
    \sum_{\vec{\mathcal{N}}_a,\vec{\mathcal{N}}'_{a},\vec{\mathcal{N}}''_{a'},\vec{\mathcal{N}}'''_{a'}} c^*_{a'}(\vec{\mathcal{N}}''_{a'})
    c_{a'}(\vec{\mathcal{N}}'''_{a'})
    c_a(\vec{\mathcal{N}}_a) c^*_{a}(\vec{\mathcal{N}}'_{a})\times\cdots\\ 
    &\;\;\;\;\otimes
    |\mathcal{N}'''_{a',2k-2}\mathcal{N}'''_{a',2k-1}\rangle\langle\mathcal{N}''_{a',2k-2}\mathcal{N}''_{a',2k-1}|\mathcal{N}_{a,2k-2}\mathcal{N}_{a,2k-1}\rangle\langle\mathcal{N}'_{a,2k-2}\mathcal{N}'_{a,2k-1}|_{X_{2k-2}}\\
    &\;\;\;\;\otimes
    |\mathcal{N}''_{a',2k-1}\mathcal{N}''_{a',2k}\rangle\langle\mathcal{N}'''_{a',2k-1}\mathcal{N}''_{a',2k}|\mathcal{N}'_{a,2k-1}\mathcal{N}_{a,2k}\rangle\langle\mathcal{N}_{a,2k-1}\mathcal{N}_{a,2k}|_{X_{2k-1}}\\
     &\;\;\;\;\otimes
    |\mathcal{N}''_{a',2k+1}\mathcal{N}''_{a',2k+2}\rangle\langle\mathcal{N}''_{a',2k+1}\mathcal{N}'''_{a',2k+2}|\mathcal{N}_{a,2k+1}\mathcal{N}'_{a,2k+2}\rangle\langle\mathcal{N}_{a,2k+1}\mathcal{N}_{a,2k+2}|_{X_{2k+1}}\\
    &\;\;\;\;\otimes
    |\mathcal{N}'''_{a',2k+2}\mathcal{N}'''_{a',2k+3}\rangle\langle\mathcal{N}''_{a',2k+2}\mathcal{N}''_{a',2k+3}|\mathcal{N}_{a,2k+2}\mathcal{N}_{a,2k+3}\rangle\langle\mathcal{N}'_{a,2k+2}\mathcal{N}'_{a,2k+3}|_{X_{2k+2}}\\
    &\;\;\;\;\otimes\cdots.
    \end{split}\label{PartialSpectrumProduct}
\end{align} 
\end{widetext}
Note that ${\cal N}'_{a, 2k}, {\cal N}'_{a, 2k +1}$ and ${\cal N}^{'''}_{a, 2k}, {\cal N}^{'''}_{a, 2k +1}$ are absent in the sums over $\vec{{\cal N}}'_a$ and $\vec{{\cal N}}^{'''}_a$.
The above overlaps fix $a = a'$ and 
\begin{align}
    \begin{split}
        \mathcal{N}'''_{a',i}&=\mathcal{N}'_{a,i}, \quad    \mathcal{N}''_{a',i}=\mathcal{N}'_{a,i}, 
    \end{split}\label{RemainedCylindersCond2}
\end{align}
for $1 \leq i \leq 2M$.
Analogous to \eqref{extrastep}, we may now read off $|| \rho^{T_{\rm odd}}_{Y} ||_1 = \tr \sqrt{(\rho_{Y}^{T_{{\rm odd}}})^\dagger \rho_{Y}^{T_{{\rm odd}}}}$ to find
\begin{align}
|| \rho^{T_{\rm odd}}_{Y} ||_1 & = \sum_a \sum_{\vec{{\cal N}}_a, \vec{{\cal N}}'_a} |c_a(\vec{{\cal N}}_a)| |c_a(\vec{{\cal N}}'_a)|  \cr
& = \sum_a |\psi_a|^2 \bigg|\frac{Z^{2M - 2}_a(1/2)}{Z^{M - 2}_a(1)} \bigg|\bigg|\frac{Z^{2M - 2}_a(1/2)}{Z^M_a(1)} \bigg| \cr
& = \sum_a \left( |\psi_a| \left(\frac{Z_a(1/2)}{\sqrt{Z_a(1)}}\right)^{2M-2}\right)^2.
\end{align}
The second equality follows from \eqref{ProductCC} and the definition of the $c_a(\vec{{\cal N}}_a)$.
Taking the logarithm of $|| \rho^{T_{\rm odd}}_{Y} ||_1$, we obtain \eqref{lemma2result} and thereby complete the proof of the lemma when $\bar{Y} = X_{2k}$, i.e., a single even cylinder has been removed.

This dependence of the entanglement negativity on the number $R$ of shared interfaces of the remaining cylinders $Y$ not traced over continues in the other cases.

(Case II) If $X_{2k}$ and an additional ``even" cylinder $X_{2k'}$ are traced out, then the generalization of \eqref{onetraceconstraint} will also remove any dependence on ${\cal N}'_{a, 2k'}$ and ${\cal N}'_{a, 2k' + 1}$.
Proceeding with the remaining steps outlined for Case 1, we find $R = 2M - 4$, reflective of the number of remaining shared interfaces.
In the special case when $k' = k+1$, degrees of freedom on $X_{2k +1}$ become disconnected from those on the remaining cylinders; because $X_{2k +1}$ has no shared interface with the remaining cylinders, we conclude that it effectively makes no contribution to the entanglement negativity.

(Case III) The proof proceeds identically if instead $\bar{Y} = X_{2k-1}$.
Then \eqref{onetraceconstraint} removes the dependence on ${\cal N}'_{a, 2k'-1}$ and ${\cal N}'_{a, 2k'}$, and the remainder of the proof proceeds as before, obtaining $R = 2M -2$ in this case.

It is straightforward to generalize the above reasoning to the situation when more than two nonconsecutive cylinders, e.g., a subset of the ``even" cylinders, are removed.
In this situation, the generalization of the above arguments gives $R = 2M - 2Q$, where $Q$ equals the number of cylinders removed.

(Case IV) The remaining case to discuss involves the removal (by trace) of two consecutive cylinders, say, $X_{2k-1}$ and $X_{2k}$.
In this case, \eqref{onetraceconstraint} removes the dependence on $\mathcal{N}'_{a,2k-1},\mathcal{N}'_{a,2k},$ and $\mathcal{N}'_{a,2k+1}$.
There are only three mode numbers in this case because states at the interface between $X_{2k-1}$ and $X_{2k}$ share $\mathcal{N}'_{a,2k}$.
Proceeding then as above we obtain $R=2M - 3$.

The above reasoning can then be suitably generalized, as needed, to show \eqref{lemma2result} with $R$ equal to the number of shared interfaces in $Y$.
This completes our proof of Lemma \ref{lemma2}.
\end{proof}

By Lemma \ref{lemma2}, the calculation of the entanglement negativity between $Y_{\rm odd}$ and $Y_{\rm even}$ reduces to the calculation of the ratio of entanglement Hamiltonian partition functions.
Making use of the partition function results from the previous section, we find
\begin{align}
        {\cal E}_{Y_{\rm odd}: Y_{\rm even}} & = {R \over 2} \left({\pi \over 2 v_e} \right) L - {R \over 2} \log m + \log \sum_a |\psi_a|^2 \cr
        & = {R \over 2} \left({\pi \over 2 v_e} \right) L - {R \over 2} \log m,
        \end{align}
where in the last line we used $\sum_a |\psi_a|^2 = 1$.
This verifies \eqref{cylinderresult} with $\alpha = \pi/2v_e$ and ${\cal D} = \sqrt m$. 
Similar to the torus geometry, there are $R$ contributions, equal to $\log \sqrt{m}$.
However, the topological sector correction is absent: the trace that is used to construct the cylinder state removes the correlations between different topological sectors when the state is Abelian ($d_a = 1$).
We show in the next section that a topological sector correction is present for the Moore-Read state.

\subsection{\texorpdfstring{$\nu = 1/m$}{v=1/m} Moore-Read State}

\subsubsection{Torus Geometry}

We now consider the entanglement negativity of the Moore-Read state on the torus geometry [e.g., Fig.~\ref{4cylinders}(a)].
At filling fraction $\nu = 1/m$, there are $2m$ untwisted anyon sectors $a=I^r=e^r$ or $\chi^r=\chi e^r$ and $m$ twisted anyon sectors $a = \xi^{r+1/2} =  \xi e^{r+1/2}$ topological sectors for $r=0,1,\ldots,m$ (see Table~\ref{tab:MRanyons}).
For the untwisted sector, much of our presentation will mirror that of the Laughlin state; in the twisted sector, there are some differences associated to Majorana fermion zero modes that we will highlight and discuss as they arise.

We begin with the untwisted sectors.
The un-normalized torus ground state in sector $a$ can be factorized as [see \eqref{untwistedMRGS}]
\begin{align}
\label{interfacefactorization}
    |\Psi_a \rangle = 
    P_a | \hat \Psi_a \rangle = 
    \bigotimes_{i=1}^{2M} P_{a,i} | \hat \Psi_{a, i}\rangle,
\end{align}
where $i$ labels the interface between cylinders $X_i$ and $X_{i+1}$, $P_a = \otimes_i P_{a,i}$ is the decomposition of the sector $a$ projection operator into projection operators local to each interface, and the unprojected state:
\begin{align}
\label{untwistedunprojected}
    | \hat \Psi_{a,i} \rangle = |b^{\mathrm{zero}}_{r, i} \rangle \otimes |b^{\mathrm{osc}}_{i} \rangle \otimes |f^{\rm osc}_i \rangle.
\end{align}
The bosonic zero mode and oscillator states are given in Eqs.~\eqref{GSLzm} and \eqref{GSLosc}; the fermionic oscillator states are given in Eq.~\eqref{GSFosc}.
The collective mode numbers are now
    \begin{align}
        {\cal N}^{LX_{i}} & \equiv \big(- N^{LX_{i}}, \{n^{LX_{i}}_{b,- k} \}_{k>0}, \{n^{LX_{i}}_{f, -k} \}_{k>0} \big), \\
        {\cal N}^{RX_{i}} & \equiv \big(N^{RX_{i}}, \{n^{RX_{j}}_{b, k} \}_{k>0}, \{n^{RX_{i}}_{f, k} \}_{k>0} \big), \\
        \label{generalNwithfermionosc}
        {\cal N}_{a,i} & \equiv \big(N_{a,i}, \{ n_{i,k} \}_{k>0}, \{\tilde n_{i,k} \}_{k>0} \big),
\end{align}
with domains defined in Eqs.~\eqref{GSLzm}, and \eqref{GSLosc}, \eqref{GSFosc}.
When acting on $|\hat \Psi_{a, i} \rangle$, we may replace $P_{a,i}$ with its eigenvalue $P_a({\cal N}_{a,i})$, using
\begin{widetext}
 \begin{align}
  \label{untwistedprojectioneigenvalues}
        P_{a,i} \left( |{\cal N}_{a,i} \rangle_{RX_{i-1}} \otimes | {\cal N}_{a,i} \rangle_{LX_i} \right) & = {1 \over 2} \left(1 + (-1)^{N_{a,i}- {r \over m} + \sum_{k>0} \tilde n_{i,k} }\right)  |{\cal N}_{a,i} \rangle_{RX_{i-1}} \otimes | {\cal N}_{a,i} \rangle_{LX_i} \cr
        & \equiv P_a({\cal N}_{a,i}) |{\cal N}_{a,i} \rangle_{RX_{i-1}} \otimes | {\cal N}_{a,i} \rangle_{LX_i}.
 \end{align}
\end{widetext}
Putting this all together, we have
\begin{widetext}
    \begin{align}\label{untwistedtorusstate}
        |\Psi_a \rangle =  \bigotimes_{i=1}^{2M} \sum_{{\cal N}_{a,i}} P_a({\cal N}_{a,i}) \lambda({\cal N}_{a,i}) |{\cal N}_{a,i} \rangle_{RX_{i-1}} \otimes |{\cal N}_{a,i} \rangle_{LX_i},
    \end{align}
    where
    \begin{align}
    \label{halfeigenvalues}
        \lambda({\cal N}_{a,i}) = \exp\left[- {v_e \pi m \over 2 L} N_{a,i}^2  - \sum_{k>0} {v_e k \over 2} \left(n_{i,k} + \frac{1}{2}\right) - \sum_{k>0} {\tilde v_e k \over 2}\left(\tilde n_{i,k} + {1 \over 2}\right) \right].
    \end{align}\label{lambdaForMRsectors}
    \end{widetext}
    The norm-squared of $|\Psi_a \rangle$ is 
    \begin{align}
    \label{untwistednormsquared}
        \left(Z^{\rm untwisted}_a\right)^{2M} = \left(\sum_{{\cal N}_a} P({\cal N}_{a}) \lambda^2({\cal N}_a)\right)^{2M}.
    \end{align}
Similar to the Laughlin case, $Z^{\rm untwisted}_a$ defines the untwisted sector $a$ partition function at inverse ``temperature" $\beta = 1$ of the Moore-Read entanglement Hamiltonian $H_a$,
    \begin{align}
    \label{sectorapartitionfunction}
        Z^{\rm untwisted}_a(\beta) = \tr e^{- \beta H_a},
    \end{align}
    with untwisted entanglement spectrum equal to $-2 \log \lambda({\cal N}_a)$, subject to the condition on allowed states imposed by the projection operator eigenvalues in \eqref{untwistedprojectioneigenvalues}.
    (As we have already done above, we will continue to abuse notation below; however, we will make sure to specify whether we are dealing with the untwisted or twisted topological sectors.)

We next turn to the twisted sectors.
The unprojected sector $a$ torus state is
\begin{align}
    |\hat \Psi_a \rangle = \Big(\bigotimes_{i=1}^{2M} | \hat \Psi_{a,i} \rangle \Big) \otimes | f^{\rm zero} \rangle,
\end{align}
where $| \hat \Psi_{a,i} \rangle$ takes the form in \eqref{untwistedunprojected} and the Majorana zero mode state $| f^{\rm zero} \rangle$ is given in \eqref{twistedzero}.
Note that the domain of $N_{a,i} \in \mathbb{Z} + (r+1/2)/m$ and oscillator fermion momenta are shifted by a half integer.

The presence of Majorana zero modes makes the decomposition of the torus ground states into cylinder states more delicate.
In particular, the twisted sector $a$ projection operator $P_{a}$ does not factorize in terms of independent projection operators local to each interface (as in, e.g., \eqref{interfacefactorization}); instead, we can at most decompose $P_a = \otimes_{i = 1}^{2M} P_{a,X_i}$, where $P_{a,X_i}$ is the projection operator for cylinder $X_i$.
When acting on a cylinder state $|\hat \Psi_{a,i}\rangle \otimes |\gamma_i \rangle \in (\otimes_i |\hat \Psi_{a,i} \rangle) \otimes |f^{\rm zero} \rangle$, we may replace the projection operator $P_{a,X_i}$ with its eigenvalue $P_a({\cal N}_{a,i}, {\cal N}_{a,i+1}, \gamma_i)$ using,
  \begin{widetext}
  \begin{align}
  \label{twistedprojectioneigenvalues}
        P_{a,X_i} | {\cal N}_{a,i}, {\cal N}_{a,i+1}, \gamma_i \rangle|_{X_i} & = {1 \over 2} \Big(1 + (-1)^{-N_{a,i} + N_{a, i+1} + \sum_{k>0} \big( \tilde n_{i,k} + \tilde n_{i+1,k} \big) + \gamma_i} \Big) | {\cal N}_{a,i}, {\cal N}_{a,i+1}, \gamma_i \rangle|_{X_i} \cr
        & \equiv P_a({\cal N}_{a,i}, {\cal N}_{a,i+1}, \gamma_i) | {\cal N}_{a,i}, {\cal N}_{a,i+1}, \gamma_i \rangle|_{X_i}.
    \end{align}
    \end{widetext} 
Using these eigenvalues, the product of $2M$ projection operators can be reduced to a product of $(2M-1)$ operators, e.g.,
\begin{align}
\label{projectorsimplification}
    \prod_{i=1}^{2M} P_{a, X_i} | \hat \Psi_a \rangle = \prod_{i=1}^{2M-1} P_{a, X_i} | \hat \Psi_a \rangle.
\end{align}
Thus, the norm squared of twisted sector $a$ state $|\Psi_a \rangle = P_a | \hat \Psi_a \rangle$ [see \eqref{twistedMRGS}] equals
\begin{widetext}
\begin{align}
\label{twistednormsquared}
   \left (Z^{\rm twisted}_a\right)^{2M} & = {1 \over 2^{2M-1}}\sum_{\vec {\cal N}_a} \lambda^2({\cal N}_{a,2M}) \sum_{\gamma_1, \ldots, \gamma_{2M-1} \in \{0,1\}} \prod_{i = 1}^{2M - 1} P_a({\cal N}_{a,i}, {\cal N}_{a,i+1}, \gamma_i) \lambda^2({\cal N}_{a,i}) \cr
    & = {1 \over 2^{2M-1}} \left( \sum_{{\cal N}_a} \lambda^2({\cal N}_a) \right)^{2M},
\end{align}
\end{widetext}
where $\lambda({\cal N}_{a,i})$ is given in \eqref{halfeigenvalues} and we have used $\sum_{\gamma_i \in \{0,1\}} P_a({\cal N}_{a,i}, {\cal N}_{a,i+1}, \gamma_i) = 1$.
As before, we may interpret $Z_a^{\rm twisted}$ in terms of a twisted sector partition function $Z_a^{\rm twisted}(\beta)$ of a Hamiltonian $H_a$ with spectrum $- 2 \log \lambda({\cal N}_{a})$ at inverse ``temperature" equal to one.
In contrast to the untwisted sector, the projection operator eigenvalues do not appear in the norm squared of the twisted sector state $| \Psi_a \rangle$ or the corresponding partition function.

To calculate the entanglement negativity of $| \Psi \rangle = \sum_a \psi_a |\overline{\Psi}_a \rangle$, where the sum is over all topological sectors and the normalized sector $a$ state is 
\begin{align}
    \label{torusMRstate}
    |\overline{\Psi}_a \rangle = Z_a^{-M} | \Psi_a \rangle.
\end{align} 
We will again use Lemma \ref{lemma1}.
[Here and in the generalized proof below we drop the untwisted/twisted superscripts for the normalization factors in ~\eqref{untwistednormsquared} and \eqref{twistednormsquared}.]
The proof that we previously gave of this lemma was special to the Laughlin state; below we will sketch how the proof generalizes for the Moore-Read state. 
\begin{proof}[Generalized Proof of Lemma 1]
As before, we will directly evaluate $||\rho^{T_{\rm odd}}||_1 = \tr \sqrt{(\rho^{T_{\rm odd}})^\dagger \rho^{T_{\rm odd}}}$.
We begin by writing the torus state as
\begin{align}
\label{torusgeneral}
    |\Psi\rangle &=\sum_a \sum_{\vec{{\cal X}_{a}}}c_a(\vec{\mathcal{N}}_a)P_a(\vec{{\cal X}}_{a}) \bigotimes_{i=1}^{2M} |{\cal X}_{a,i}\rangle, \\
    {\cal X}_{a,i} &\equiv (\mathcal{N}_{a,i},\mathcal{N}_{a,i+1}, s_a\gamma_i),\label{calXdef}\\
    \label{cgeneraldef}
    c_a(\vec{\mathcal{N}}_a) & \equiv \psi_a \prod_{i=1}^{2M}\frac{\lambda(\mathcal{N}_{a,i})}{\sqrt{Z_{a}}},\\
    \label{projectorgeneral}
    P_a(\vec{{\cal X}_{a}}) & \equiv \prod_{i=1}^{2M} P_a({\cal X}_{a,i})
\end{align}
where $s_a=0$ if $a$ belongs to an untwisted sector and $s_a=1$ if $a$ belongs to a twisted sector.
The sum over $\vec{{\cal X}_{a}}$ is understood to be a sum over $\vec{\cal N}_{a}$ and, when $s_a = 1$, the Majorana fermion parity eigenvalues $\vec \gamma$, (i.e., $\sum_{\vec{{\cal X}}_{a}} = \sum_{\vec{{\cal N}}_{a}, s_a \vec \gamma}$).
$P_a({\cal X}_{a,i}) = P_a({\cal N}_{a,i}, {\cal N}_{a, i+1}, s_a \gamma_i)$ is defined in \eqref{twistedprojectioneigenvalues} for twisted sector $a$ where $s_a = 1$, this eigenvalue is also valid for untwisted $a$, in which case $s_a = 0$.

The density matrix $\rho = |\Psi\rangle\langle\Psi|$ and its partial transpose with respect to $X_{\rm odd}$ are then
\begin{align}
    \begin{split}
    \label{torusdensitymatrixMR}
        \rho &=\sum_{a,a'}\sum_{\vec{{\cal X}_{a}}, \vec{{\cal X}}'_{a'}} c_a(\vec{\mathcal{N}}_a)c^*_{a'}(\vec{\mathcal{N}}'_{a'})\\ 
        &\;\;\;\;\times P_{a}({\cal X}_{a,1})|{\cal X}_{a,1}\rangle\langle {\cal X}'_{a',1}|P_{a'}({\cal X}'_{a',1})\\
        &\;\;\;\;\otimes P_{a}({\cal X}_{a,2})|{\cal X}_{a,2}\rangle\langle {\cal X}'_{a',2}|P_{a'}({\cal X}'_{a',2})\\
        &\;\;\;\;\otimes \cdots\\
        &\;\;\;\;\otimes P_{a}({\cal X}_{a, 2M})|{\cal X}_{a, 2M}\rangle\langle {\cal X}'_{a', 2M}|P_{a'}({\cal X}'_{a', 2M}),
    \end{split} \\
    \begin{split}
        \rho^{T_{\rm odd}} 
        &=\sum_{a,a'}\sum_{\vec{{\cal X}_{a}}, \vec{{\cal X}}'_{a'}} c_a(\vec{\mathcal{N}}_a)c^*_{a'}(\vec{\mathcal{N}}'_{a'})\\ 
        &\;\;\;\;\times P_{a'}({\cal X}'_{a',1})|{\cal X}'_{a', 1}\rangle\langle {\cal X}_{a, 1}|P_{a}({\cal X}_{a,1})\\
        &\;\;\;\;\otimes P_{a}({\cal X}_{a,2})|{\cal X}_{a,2}\rangle\langle {\cal X}'_{a',2}|P_{a'}({\cal X}'_{a',2})\\
        &\;\;\;\;\otimes \cdots\\
        &\;\;\;\;\otimes P_{a}({\cal X}_{a,2M})|{\cal X}_{a, 2M}\rangle\langle {\cal X}'_{a', 2M}|P_{a'}({\cal X}'_{a', 2M}).
    \end{split}
\end{align}
\begin{widetext}
Thus,
\begin{align}
\begin{split}
    (\rho^{T_{\rm odd}})^\dagger\rho^{T_{\rm odd}} &=
    \sum_{a,a',a'',a'''}
    \sum_{\vec{{\cal X}_{a}}} \sum_{\vec{{\cal X}}'_{a'}}
    \sum_{\vec{{\cal X}}^{''}_{a^{''}}} \sum_{\vec{{\cal X}}^{'''}_{a^{'''}}}c_a(\vec{\mathcal{N}}_a)c^*_{a'}(\vec{\mathcal{N}}'_{a'})c_{a'''}(\vec{\mathcal{N}}'''_{a'''})c^*_{a''}(\vec{\mathcal{N}}''_{a''})\\
    &\;\;\;\;\times P_{a'''}({\cal X}^{'''}_{a''', 1})P_{a''}({\cal X}^{''}_{a'', 1})|{\cal X}^{''}_{a'', 1}\rangle\langle {\cal X}^{'''}_{a''',1}|{\cal X}'_{a', 1}\rangle\langle {\cal X}_{a, 1}|P_{a}({\cal X}_{a, 1})P_{a'}({\cal X}'_{a', 1})\\
    &\;\;\;\;\otimes P_{a''}({\cal X}^{''}_{a'', 2})P_{a'''}({\cal X}^{'''}_{a''', 2})|{\cal X}^{'''}_{a''', 2}\rangle\langle {\cal X}^{''}_{a'', 2}|{\cal X}_{a, 2}\rangle\langle {\cal X}'_{a', 2}|P_{a'}({\cal X}'_{a', 2})P_{a}({\cal X}_{a, 2})\\
    &\;\;\;\;\otimes P_{a'''}({\cal X}^{'''}_{a''', 3})P_{a''}({\cal X}^{''}_{a'', 3})|{\cal X}^{''}_{a'', 3}\rangle\langle {\cal X}^{'''}_{a''', 3}|{\cal X}'_{a', 3}\rangle\langle {\cal X}_{a, 3}|P_{a}({\cal X}_{a, 3})P_{a'}({\cal X}'_{a', 3})\\
    &\;\;\;\;\otimes \ldots\\
    &\;\;\;\;\otimes P_{a''}({\cal X}^{''}_{a'', 2M})P_{a'''}({\cal X}^{'''}_{a''', 2M})|{\cal X}^{'''}_{a''', 2M}\rangle\langle {\cal X}^{''}_{a'', 2M}|{\cal X}_{a, 2M}\rangle\langle {\cal X}'_{a', 2M}|P_{a'}({\cal X}'_{a', 2M})P_{a}({\cal X}_{a, 2M}).
\end{split}\label{RhodaggerrhoAgain}
\end{align} 
The inner products identify:
\begin{align}
    \begin{split}
        a'''&=a', \quad {\cal X}^{'''}_{a''', 2k-1} = {\cal X}'_{a', 2k-1}, \quad 
        P_{a'''}({\cal X}'''_{a''', 2k-1}) = P_{a'}({\cal X}'_{a', 2k-1}),\\
        a''&=a, \quad \;\;\;\;\; {\cal X}^{''}_{a'', 2k} = {\cal X}_{a, 2k}, \quad 
        \;\;\;\;\;\;\;\; P_{a''}({\cal X}''_{a'', 2k}) = P_{a}({\cal X}_{a, 2k}),
    \end{split}
\end{align} 
for $k \in \{1, \ldots, M\}$.
From \eqref{calXdef}, we see that  ${\cal X}^{'''}_{a''', 2k-1} = {\cal X}'_{a', 2k-1}$ means that $\mathcal{N}'''_{a''',i}=\mathcal{N}'_{a',i}$ for all $i=1,2,\ldots, 2M$ and $s_{a'''} \gamma'''_{2k-1} = s_{a'} \gamma'_{2k-1}$ for $k = 1, 2, \ldots, M$.
If the projection operators $P_{a'''}({\cal X}'''_{a''', 2k})$ and $P_{a'}({\cal X}'_{a', 2k})$ are to be nonzero simultaneously for fixed $\vec{\mathcal{N}}'''_{a'''}= \vec{\mathcal{N}}'_{a'}$, then $s_{a'''} \gamma'''_{2k} = s_{a'} \gamma'_{2k}$ for $k = 1, 2, \ldots, M$.
Thus, we conclude the above inner products identify $P_{a'''}(\mathcal{X}'''_{a''',i})=P_{a'}(\mathcal{X}'_{a',i})$ and ${\cal X}^{'''}_{a''', i} = {\cal X}'_{a', i}$ for {\it all} $i$.
Using similar logic, we likewise find ${\cal X}^{''}_{a'', i} = {\cal X}_{a, i}$ and $P_{a''}({\cal X}''_{a'', i}) = P_{a}({\cal X}_{a, i})$ for all $i$. 
Thus, $(\rho^{T_{\rm odd}})^\dagger\rho^{T_{\rm odd}}$ is again diagonal and 
\begin{align}
\label{diagonalgeneralrhodaggerrho}
    \begin{split}
        \sqrt{\left(\rho^{T_{\rm odd}}\right)^\dagger\rho^{T_{\rm odd}}}&=    \sum_{a,a'}\sum_{\vec{{\cal X}_{a}}}\sum_{\vec{{\cal X}'_{a'}}} P_a(\vec{{\cal X}_{a}}) P_{a'}(\vec{{\cal X}}'_{a'})
        |c_a(\vec{\mathcal{N}}_a)||c_{a'}(\vec{\mathcal{N}}'_{a'})|\\
        &\times\;\;\;\;|{\cal X}_{a, 1}\rangle\langle {\cal X}_{a, 1}|\otimes |{\cal X}'_{a', 2}\rangle\langle {\cal X}'_{a', 2}|\otimes\cdots\otimes|{\cal X}'_{a', 2M}\rangle\langle {\cal X}'_{a', 2M}|.
    \end{split}
\end{align}
Consequently,
\begin{align}
     \tr\sqrt{\left(\rho^{T_{\rm odd}}\right)^\dagger\rho^{T_{\rm odd}}}&= \sum_{a,a'}\sum_{\vec{{\cal X}_{a}}}\sum_{\vec{{\cal X}'_{a'}}} P_a(\vec{{\cal X}_{a}}) P_{a'}(\vec{{\cal X}}'_{a'})
        |c_a(\vec{\mathcal{N}}_a)||c_{a'}(\vec{\mathcal{N}}'_{a'})|\\
        &=\left(\sum_a \sum_{\vec{{\cal X}_{a}}} P_a(\vec{{\cal X}_{a}}) |c_a(\vec{\mathcal{N}}_a)|\right)^2.
\end{align}
Following the same logic as in \eqref{tracenormlemma1} and using the definitions in Eqs.~\eqref{untwistednormsquared} and \eqref{twistednormsquared} as well as the identity $\sum_{\gamma_i \in \{0,1\}} P({\cal X}_{a, i}) = 1$ when $s_a = 1$, we find
\begin{align}
   ||\rho^{T_{\rm odd}}||_1 & = \left(\sum_{a \in {\rm untwisted}}|\psi_a|\left(\frac{Z^{\rm untwisted}_a(1/2)}{\sqrt{Z^{\rm untwisted}_a(1)}}\right)^{2M} + \sum_{a \in {\rm twisted}}|\psi_a|\left(\frac{Z^{\rm twisted}_a(1/2)}{\sqrt{Z^{\rm twisted}_a(1)}}\right)^{2M}\right)^2.
\end{align}
\end{widetext}
We obtain Lemma \ref{lemma1} upon taking the logarithm of $||\rho^{T_{\rm odd}}||_1$.
\end{proof}

To finish the computation of the entanglement negativity, we need to evaluate the untwisted and twisted partition functions in the $L \rightarrow \infty$ limit.
The untwisted sector $a$ partition function at inverse ``temperature" $\beta$ can be written as \cite{sohal-2020}
\begin{align}\begin{split}
Z^{\rm untwisted}_a(\beta) &=\frac{1}{2}\chi^{\rm Ising}_0(\tilde{q})\left[\chi^+_{r/m}(q)+\chi^{-}_{r/m}(q)\right]\\
&\;\;\;\;+\frac{1}{2}\chi^{\rm Ising}_{1/2}(\tilde{q})\left[\chi^+_{r/m}(q)-\chi^-_{r/m}(q)\right],
\end{split}
\end{align} 
where the characters $\chi$ are
\begin{align}\begin{split}
\chi^{\rm Ising}_0(\tilde{q}) &=\frac{1}{2}\tilde{q}^{-\frac{1}{48}}
\left[ \prod_{j>0}\left(1+\tilde{q}^{j+\frac{1}{2}}\right) + \prod_{j>0}\left(1-\tilde{q}^{j+\frac{1}{2}}\right)\right]\\
\chi^{\rm Ising}_{1/2}(\tilde{q}) &=\frac{1}{2}\tilde{q}^{-\frac{1}{48}}
\left[ \prod_{j>0}\left(1+\tilde{q}^{j+\frac{1}{2}}\right) - \prod_{j>0}\left(1-\tilde{q}^{j+\frac{1}{2}}\right)\right]\\
\chi^{\pm}_{r/m}(q)&=\left(
\sum_{n\in\mathbb{Z}}(\pm1)^nq^{m\left(n-\frac{r}{m}\right)^2/2}
\right)\\
&\;\;\;\;\times q^{-\frac{1}{24}}\prod_{j>0}(1-q^j)^{-1}.
\end{split}\label{MRCharacters}
\end{align} 
Recall that $q=e^{2\pi i\tau}$ and $\tau$ is defined in \eqref{DefineTau}.
In addition, we have a new pair of modular parameters $\tilde{\tau}$ and $\tilde{q}$ defined by $\tilde{\tau}=\tau/2=i\tilde{\tau}_2$ and $\tilde{q}=e^{2\pi i \tilde{\tau}}$.
The characters can be rewritten in terms of modular functions (Appendix~\ref{ModularFunctions}) as 
\begin{align}
    \begin{split}
        \chi^{\rm Ising}_0(\tilde{q}) &=\frac{1}{2}\sqrt{\frac{\theta^0_0(\tilde{\tau})}{\eta(\tilde{\tau})}} +\frac{1}{2}\sqrt{\frac{\theta^0_{1/2}(\tilde{\tau})}{\eta(\tilde{\tau})}},\\
        \chi^{\rm Ising}_{1/2}(\tilde{q}) &=\frac{1}{2}\sqrt{\frac{\theta^0_0(\tilde{\tau})}{\eta(\tilde{\tau})}} -\frac{1}{2}\sqrt{\frac{\theta^0_{1/2}(\tilde{\tau})}{\eta(\tilde{\tau})}},\\
        \chi^{\pm}_{r/m}(q)&=\frac{\theta^{-r/m}_0(m\tau)}{\eta(\tau)}, \quad \text{or}\quad \frac{e^{\frac{i\pi r}{m}}\theta^{-r/m}_{1/2}(m\tau)}{\eta(\tau)}.
    \end{split}
\end{align}
$\theta^0_{1/2}(\tilde\tau)$ goes to zero in the $L\rightarrow\infty$ limit. 
Therefore, the untwisted $1$ and $\chi$ sector partition functions both reduce to 
\begin{align}
        Z^{\rm untwisted}_a(\beta) = \frac{1}{2}\sqrt{\frac{\theta^0_0(\tilde{\tau})}{\eta(\tilde{\tau})}}\frac{\theta^{-r/m}_0(m\tau)}{\eta(\tau)}.
        \label{PartitionFuncUntwisted}
\end{align} 
In the $L\rightarrow\infty$ limit, 
\begin{align}
    \frac{Z^{\rm untwisted}_a(1/2)}{\sqrt{Z^{\rm untwisted}_a(1)}}
    &=\frac{1}{\sqrt{2\sqrt{m}}}\exp\left[\frac{\pi L}{8}\left(\frac{1}{v_e}+\frac{1}{2\tilde{v}_e}\right)\right].
    \label{PartitionRatioUntwisted}
\end{align}

The twisted sector $a$ partition function is 
\begin{align}
    \begin{split}
        Z^{\text{twisted}}_a &= \chi^{\text{Ising}}_{1/16}(\tilde{q})\chi^+_{(r+1/2)/m}(q),
    \end{split}
\end{align}
where the characters are
\begin{align}
\begin{split}
\chi^{\text{Ising}}_{1/16}(\tilde{q}) &= 
\sum_{n_{i,k}\in\mathbb{Z}^+}e^{-\sum_{k>0}\beta\tilde{v}_e k (n_{i,k}+\frac{1}{2})}\\
&=\tilde{q}^{\frac{1}{24}}\prod_{j=1}^\infty (1+\tilde{q}^j),\quad k=\frac{2\pi j}{L}, j\in\mathbb{Z}^+\\
\chi_{(r+1/2)/m}(q) &= q^{-\frac{1}{24}}\prod_{j=1}^\infty(1-q^j)^{-1}\\
&\;\;\;\;\times\left(
\sum_{n\in\mathbb{Z}}q^{m\left(n-\frac{r+1/2}{m}\right)^2/2}
\right).\label{PartitionRatioTwisted}
\end{split}
\end{align} 
$\chi^{\text{Ising}}_{1/16}$ produces the $d_a=\sqrt{2}$ quantum dimension associated to the Majorana quasiparticle of the Moore-Read state.
In terms of modular functions (Appendix~\ref{ModularFunctions}), the characters are
\begin{align}
    \begin{split}
        \chi^{\text{Ising}}_{1/16}(\tilde{q}) &=\sqrt{\frac{\theta^{1/2}_0(\tilde{\tau})}{2\eta(\tilde{\tau})}},\\
        \chi_{(r+1/2)/m}(q)&=\frac{\theta_0^{-(r+1/2)/m}(m\tau)}{\eta(\tau)}.
    \end{split}
\end{align}
Thus, \begin{align}
    \begin{split}
        Z^{\text{twisted}}_a &= \sqrt{\frac{\theta^{1/2}_0(\tilde{\tau})}{2\eta(\tilde{\tau})}}\frac{\theta_0^{-(r+1/2)/m}(m\tau)}{\eta(\tau)},
    \end{split}
\end{align}
and for $L \rightarrow \infty$,
\begin{align}
\begin{split}
\label{PartitionRatiotwisted}
 \frac{Z^{\rm twisted}_a(1/2)}{\sqrt{Z^{\rm twisted}_a(1)}}&=\frac{1}{\sqrt[4]{2m}}\exp\left[\frac{\pi L}{8}\left(\frac{1}{v_e}+\frac{1}{2\tilde{v}_e}\right)\right].
\end{split}
\end{align}

Plugging these untwisted and twisted partition function ratios \eqref{PartitionRatioUntwisted} and \eqref{PartitionRatiotwisted} into Lemma \ref{lemma1}, we find
\begin{widetext}
\begin{align}
    \begin{split}
     ||\rho^{T_{\rm odd}}||_1 &=
     \left(
     \sum_{a\in {\rm untwisted}}|\psi_a|\left[
     \frac{Z_a^{\rm untwisted}(1/2)}{\sqrt{Z_a^{\rm untwisted}(1)}}\right]^{2M}+
     \sum_{a\in {\rm twisted}}|\psi_a|\left[
     \frac{Z_a^{\rm twisted}(1/2)}{\sqrt{Z_a^{\rm twisted}(1)}}\right]^{2M}
     \right)^2\\
     &=  \left( \left(\frac{1}{2\sqrt{m}}\right)^M\exp\left[\frac{M\pi L}{4}\left(\frac{1}{v_e}+\frac{1}{2\tilde{v}_e}\right)\right]   \sum_a |\psi_a|(d_a)^M\right)^2;
    \end{split} \\
    \begin{split}
      \mathcal{E}_{X_{\rm odd}: X_{\rm even}} &=M\left[\frac{\pi}{2}\left(\frac{1}{v_e}+\frac{1}{2\tilde{v}_e}\right)\right]L - 2 M\log \sqrt{4 m} + 2\log \sum_a|\psi_a|(d_a)^M.  \label{eq:LogNegMRTori}
    \end{split}
\end{align}
\end{widetext}  
In the second identity above, we recovered the quantum dimensions of the quasiparticles associated to each sector: for the $\{1,\chi\}$ untwisted sectors, $d_a=1$; while for the $\xi$ twisted sectors, $d_a=\sqrt{2}$. 
In addition to the $2M$ contributions, proportional to $\log\sqrt{4 m}$, in the entanglement negativity, there is a topological sector correction equal to $2\log\sum_a|\psi_a|(d_a)^M$. 
This recovers \eqref{torusresult} with the non-universal constant given by $\alpha = \frac{\pi}{2}\left(\frac{1}{v_e}+\frac{1}{2\tilde{v}_e}\right)$ and the total quantum dimension equal to $\mathcal{D}= \sqrt{4m}$.

\subsubsection{Cylinder Geometry}
\label{cylinderMRsection}

Next we calculate the entanglement negativity between subsets of $X_{\rm even}$ and $X_{\rm odd}$ when the degrees of freedom on $N \leq M$ cylinders $\bar{Y} \subset X$ of the Moore-Read state (constructed in the previous section) have been traced over [e.g., Fig.~\ref{4cylinders}(b)].
As in \S \ref{cylinderlaughlinsection}, we denote the remaining $(2M - N)$ cylinders by $Y$ and their decomposition into ``odd" and ``even" cylinders as $Y_{\rm odd}$ and $Y_{\rm even}$.
The resulting entanglement negativity will depend on the number $R$ of shared interfaces between the remaining cylinders in $Y_{\rm odd} \cup Y_{\rm even}$.

To find the entanglement negativity in this cylinder geometry, we will apply Lemma \ref{lemma2}.
Before doing so, we describe how the proof of this lemma generalizes to the Moore-Read state.
\begin{proof}[Generalized Proof of Lemma 2] 
The argument follows almost exactly the proof for the Laughlin case upon updating the notation to the Moore-Read state with the replacements: $c_a(\vec{{\cal N}}_a) \rightarrow P_a(\vec{{\cal X}}_{a}) c_a(\vec{{\cal N}}_a)$ and $|{\cal N}_{a,i}, {\cal N}_{a,i+1} \rangle \rightarrow |{\cal X}_{a,i} \rangle$.
Because of this, we will only discuss one case (Case I) below; the remaining cases (Cases II$-$IV) follow straightforwardly using the same logic as in the Laughlin case and the manipulations outlined for the generalized proof of Lemma \ref{lemma1}.
We will suppress the untwisted/twisted superscripts on $Z_a$ when convenient.

(Case I) We remove cylinder $\bar{Y} = X_{2k}$ where $1<2k \leq 2M$ by tracing over its left and right edge states.
Thus, $Y_{\rm odd} = X_{\rm odd}$, $Y_{\rm even} = X_2 \cup \cdots \cup X_{2k - 2} \cup X_{2k + 2} \cdots X_{2M}$, and $R = 2M - 2$.

The density matrix of the Moore-Read torus state in \eqref{torusgeneral} is
\begin{align}
\begin{split}
    \rho &=\sum_{a,a'}\sum_{\vec{{\cal X}_{a}}, \vec{{\cal X}}'_{a'}} c_a(\vec{\mathcal{N}}_a)c^*_{a'}(\vec{\mathcal{N}}'_{a'}) \times \cdots \\ 
        &\;\;\;\;\times P_{a}({\cal X}_{a,2k-1})|{\cal X}_{a,2k-1}\rangle\langle {\cal X}'_{a',2k-1}|P_{a'}({\cal X}'_{a',2k-1})\\
        &\;\;\;\;\otimes P_{a}({\cal X}_{a,2k})|{\cal X}_{a,2k}\rangle\langle {\cal X}'_{a',2k}|P_{a'}({\cal X}'_{a',2k})\\
        &\;\;\;\;\otimes P_{a}({\cal X}_{a, 2k+1})|{\cal X}_{a, 2k+1}\rangle\langle {\cal X}'_{a', 2k+1}|P_{a'}({\cal X}'_{a', 2k+1})\\
        &\;\;\;\;\otimes \cdots.
        \end{split}
\end{align}
The trace over degrees of freedom on cylinder $\bar{Y} = X_{2k}$ sets
\begin{align}
\label{generalconstraint}
    a' = a, \quad {\cal X}'_{a',2k} = {\cal X}_{a,2k}, \quad P_{a'}({\cal X}'_{a',2k}) = P_{a}({\cal X}_{a,2k})
\end{align}
and removes the corresponding outer products involving states on $X_{2k}$
In particular, $\Tr_{\bar{Y}} (\rho)$ is a direct sum over untwisted and twisted topological sectors.
\begin{widetext}
The partial transpose of $\rho_Y = \Tr_{\bar{Y}} (\rho)$ with respect to $Y_{\rm odd}$ is
\begin{align}
\label{generalpartialtranspose}
\begin{split}
    \rho_Y^{T_{\rm odd}} &=\sum_{a}\sum_{\vec{{\cal X}_{a}}, \vec{{\cal X}}'_{a}} c_a(\vec{\mathcal{N}}_a)c^*_{a}(\vec{\mathcal{N}}'_{a}) P^2_{a}({\cal X}_{a,2k}) \times \cdots \\ 
            &\;\;\;\;\times P_{a}({\cal X}_{a,2k-2})|{\cal X}_{a,2k-2}\rangle\langle {\cal X}'_{a,2k-2}|P_{a}({\cal X}'_{a,2k-2})\\
        &\;\;\;\;\otimes P_{a}({\cal X}_{a,2k-1})|{\cal N}'_{a,2k-1}, {\cal N}_{a,2k}, s_a \gamma'_{i-1}\rangle\langle {\cal X}_{a,2k-1}|P_{a}({\cal N}'_{a,2k-1}, {\cal N}_{a,2k}, s_a \gamma'_{i-1})\\
        &\;\;\;\;\otimes P_{a}({\cal X}_{a, 2k+1})|{\cal N}_{a, 2k+1}, {\cal N}'_{a, 2k+2}, s_a \gamma'_{2k+1}\rangle\langle {\cal X}_{a, 2k+1}|P_{a}({\cal N}_{a, 2k+1}, {\cal N}'_{a, 2k+2}, s_a \gamma'_{2k+1})\\
                &\;\;\;\;\otimes P_{a}({\cal X}_{a, 2k+2})|{\cal X}_{a, 2k+2}\rangle\langle {\cal X}'_{a, 2k+2}|P_{a}({\cal X}'_{a, 2k+2})\\
        &\;\;\;\;\otimes \cdots,
        \end{split}
\end{align}
where $c_a(\vec{\mathcal{N}}_a) c^*_{a}(\vec{\mathcal{N}}'_{a})$ takes the same form as in \eqref{ProductCC}.
There is no dependence on ${\cal X}'_{a, 2k}$ above because of \eqref{generalconstraint}.
In \eqref{generalpartialtranspose}, we have expanded out the arguments of two of the projection operator eigenvalues and kets involved in the outerproducts associated to cylinders $X_{2k-1}$ and $X_{2k+1}$ to make the identifications \eqref{generalconstraint} manifest.
Next, we compute
\begin{align}
    \begin{split}
    (\rho^{T_{\rm odd}}_{Y})^\dagger \rho^{T_{\rm odd}}_{Y} &= \sum_{a,a'}
    \sum_{\vec{\mathcal{X}}_a,\vec{\mathcal{X}}'_{a},\vec{\mathcal{X}}''_{a'},\vec{\mathcal{X}}'''_{a'}} c^*_{a'}(\vec{\mathcal{N}}''_{a'})
    c_{a'}(\vec{\mathcal{N}}'''_{a'})
    c_a(\vec{\mathcal{N}}_a) c^*_{a}(\vec{\mathcal{N}}'_{a}) P^2_{a}({\cal X}_{a,2k}) P^2_{a'}({\cal X}''_{a,2k}) \times\cdots\\ 
    &\;\;\;\;\otimes
    P_{a'}({\cal X}'''_{a,2k-2}) P_{a'}({\cal X}''_{a,2k-2}) P_{a}({\cal X}_{a,2k-2}) P_{a}({\cal X}'_{a,2k-2}) |\mathcal{X}'''_{a',2k-2}\rangle\langle\mathcal{X}''_{a',2k-2}|\mathcal{X}_{a,2k-2}\rangle\langle\mathcal{X}'_{a,2k-2}|\\
    &\;\;\;\;\otimes
    P_{a'}({\cal X}'''_{a,2k-1}) P_{a'}({\cal X}''_{a,2k-1}) P_{a}({\cal X}_{a,2k-1}) P_{a}({\cal X}'_{a,2k-1}) |\mathcal{X}''_{a',2k-1}\rangle\langle\mathcal{X}'''_{a',2k-1}|\mathcal{X}'_{a,2k-1}\rangle\langle\mathcal{X}_{a,2k-1}|\\
     &\;\;\;\;\otimes
   P_{a'}({\cal X}'''_{a,2k+1}) P_{a'}({\cal X}''_{a,2k+1}) P_{a}({\cal X}_{a,2k+1}) P_{a}({\cal X}'_{a,2k+1}) |\mathcal{X}''_{a',2k+1}\rangle\langle\mathcal{X}''_{a',2k+1}|\mathcal{X}_{a,2k+1}\rangle\langle\mathcal{X}_{a,2k+1}|\\
    &\;\;\;\;\otimes
   P_{a'}({\cal X}'''_{a,2k+2}) P_{a'}({\cal X}''_{a,2k+2}) P_{a}({\cal X}_{a,2k+2}) P_{a}({\cal X}'_{a,2k+2}) |\mathcal{X}'''_{a',2k+2}\rangle\langle\mathcal{X}''_{a',2k+2}|\mathcal{X}_{a,2k+2}\rangle\langle\mathcal{X}'_{a,2k+2}|\\
    &\;\;\;\;\otimes\cdots.
    \end{split}\label{PartialSpectrumProductgeneral}
\end{align} 
\end{widetext}
 \label{projectorgeneral2}
Note that ${\cal X}'_{a, 2k}$ and ${\cal X}^{'''}_{a, 2k}$ are absent in the sums over $\vec{{\cal X}}'_a$ and $\vec{{\cal X}}^{'''}_a$ and it is to be understood that we have imposed \eqref{generalconstraint} and the analogous constraints for ${\cal X}^{'''}_{a, 2k}$ above.
The above overlaps fix $a' = a$ and identify
\begin{align}
\begin{split}
        \mathcal{X}'''_{a',2i-1}&=\mathcal{X}'_{a,2i-1}, \quad    \mathcal{X}''_{a',2i}=\mathcal{X}_{a,2i}, \\
        P_{a'}({\cal X}'''_{2i-1}) & = P_{a}({\cal X}'_{2i-1}), \quad P_{a'}({\cal X}''_{2i}) = P_{a}({\cal X}_{2i}),
    \end{split}\label{RemainedCylindersCond}
\end{align}
for $1 < 2i \leq 2M$.
(Recall that $2M+1 \equiv 1$.)
As in the generalized proof of Lemma \ref{lemma1}, these identifications imply $\mathcal{X}'''_{a',i}=\mathcal{X}'_{a,i}$, $\mathcal{X}'_{a',i}=\mathcal{X}_{a,i}$, and equate corresponding projection operators for all $i$.  
Thus, $(\rho^{T_{\rm odd}}_{Y})^\dagger \rho^{T_{\rm odd}}_{Y}$ is again diagonal and, analogous to \eqref{diagonalgeneralrhodaggerrho}, we may read off $\tr \sqrt{(\rho^{T_{\rm odd}}_{Y})^\dagger \rho^{T_{\rm odd}}_{Y}}$ to find
\begin{align}
|| \rho^{T_{\rm odd}}_{Y} ||_1 & = \sum_a \sum_{\vec{\mathcal{X}}_a, \vec{\mathcal{X}}'_a} P_a(\vec{\mathcal{X}}_a) |c_a(\vec{{\cal N}}_a)| P_a(\vec{\mathcal{X}}'_a) |c_a(\vec{{\cal N}}'_a)|  \cr
& = \sum_a \left( |\psi_a| \left(\frac{Z_a(1/2)}{\sqrt{Z_a(1)}}\right)^{2M-2}\right)^2,
\end{align}
where the sum is over all topological sectors and the identifications in \eqref{generalconstraint} are understood and the corresponding sums are removed.
Taking the logarithm of $|| \rho^{T_{\rm odd}}_{Y} ||_1$, we complete the proof of Lemma \ref{lemma2} for the Moore-Read state when $X^{''} = X_{2k}$.
As remarked above, the remaining cases follow similarly.
\end{proof}

By Lemma 2, the entanglement negativity between $Y_{\rm odd}$ and $Y_{\rm even}$ reduces to the calculation of entanglement Hamiltonian partition functions for the Moore-Read state.
Using the partition functions calculated in the previous section, we find
\begin{align}
        {\cal E}_{Y_{\rm odd}: Y_{\rm even}} & = {R \over 2} \left(\frac{\pi}{2}\left(\frac{1}{v_e}+\frac{1}{2\tilde{v}_e}\right)\right) L - R \log \sqrt{4 m} \cr & + \log \sum_a |\psi_a|^2 d_a^R,\label{eq:LogNegCylinder}
\end{align}
where the sum is over $a$ in the last term is over all topological sectors and the quantum dimensions $d_a = 1$ for the $\{1,\chi\}$ untwisted sectors and $d_a=\sqrt{2}$ for the $\xi$ twisted sectors.
This verifies \eqref{cylinderresult} with $\alpha = \frac{\pi}{2}\big(\frac{1}{v_e}+\frac{1}{2\tilde{v}_e}\big)$ and ${\cal D} = \sqrt{4 m}$. 
Similar to the torus geometry, there are $R$ terms each proportional to $\log {\cal D}$.
In addition and in contrast to the Abelian case (where $d_a = 1$), there is a topological sector correction equal to $\log |\psi_a|^2 d_a^R$.

\section{Disentangling}
\label{sectionfour}

In this section, we discuss how the topological order of the Laughlin and Moore-Read states affects the spatial structure of their manybody wavefunctions. 
Specifically, we determine when the disentangling condition,
\begin{align}
\label{eq-disent-logneg}
  \mathcal{E}_{A:BC}(\rho) = \mathcal{E}_{A:B}(\rho_{AB}),
\end{align}
holds, for suitable choices of cylinder subsets $A$, $B$, and $C$ of the torus, and the implication of \eqref{eq-disent-logneg} for the manybody wavefunction.
We focus on two decompositions of the torus:  
\begin{enumerate}
\item $AB_1CB_2$ geometry: The torus is divided into four consecutive
  cylinders $A$, $B_1$, $C$, and $B_2$ with disjoint $B=B_1\cup B_{2}$.
  In this case, $\rho_{ABC} = | \Psi_{ABC} \rangle \langle \Psi_{ABC} |$ with $|\Psi_{ABC} \rangle$ a pure ground state on the torus.
\item $ABCD$ geometry: The torus is divided into four consecutive
  cylinders $A$, $B$, $C$, and $D$.
  In this case, $\rho_{ABC} = \tr_D |\Psi_{ABCD} \rangle \langle \Psi_{ABCD} |$ is a mixed state on cylinder $A\cup B\cup C$ and $|\Psi_{ABCD} \rangle$ is a pure torus ground state.
\end{enumerate}
The entanglement negativity results in the previous section can be used to determine when the disentangling condition \eqref{eq-disent-logneg} is satisfied.
Applying Eqs.~\eqref{torusresult} with $M = 1$ and \eqref{cylinderresult} with $R=2$ to the $AB_1CB_2$ geometry, we find
\begin{align}
\label{longrangeentanglementgeneralrepeat}
{\cal E}_{A:B C}(\rho_{ABC}) - {\cal E}_{A:B}(\rho_{AB}) = \log {\big(\sum_a |\psi_a| d_a \big)^2 \over \sum_a |\psi_a|^2 d_a^2}.
\end{align}
Consequently, only torus states in a specific topological sector, i.e., $\psi_a = 1$ for a single $a$ with all other amplitudes vanishing, satisfy the disentangling condition.
For the $ABCD$ geometry, \eqref{cylinderresult} with $R=1$ implies any mixed cylinder state on $A\cup B\cup C$ satisfies \eqref{eq-disent-logneg}.

Generally, for a tripartite Hilbert space ${\cal H_A} \otimes {\cal H_B} \otimes {\cal H_C}$, the degrees of freedom in subsystems $A$ and $C$ have no quantum correlations in states that satisfy the disentangling condition.
This allows their corresponding wavefunctions to be disentangled in the following sense.
For pure states $\ket{\Psi_{ABC}}$, He and Vidal \cite{he-vidal-2015} showed that \eqref{eq-disent-logneg} implies that there exists a
decomposition of the Hilbert space of region $B$ as
$\mathcal{H}_B=\mathcal{H}_{B_L}\otimes \mathcal{H}_{B_R}$ such that
the state can be factorized as
\begin{align}\label{eq-disent-decomp-pure-2}
  \ket{\Psi_{ABC}} = \ket{\Psi_{AB_L}}\otimes \ket{\Psi_{B_RC}}.
\end{align}
The reverse statement is also valid: The disentangling condition \eqref{eq-disent-logneg} is implied by states satisfying \eqref{eq-disent-decomp-pure-2}.
For mixed states, Gour and Guo \cite{gour-guo-2018} demonstrated that the disentangling condition \eqref{eq-disent-logneg}  is satisfied for all states that saturate the strong subadditivity of the entanglement
entropy. 
The structure of these states follows 
\begin{align}\label{eq-ssa-structrepeat}
\rho_{ABC} = \sum_j p_j \rho_{AB_L^j}\otimes \rho_{B_R^j C},
\end{align}
where the Hilbert space of $B$ decomposes into $\mathcal{H}_B=\bigoplus_j \mathcal{H}_{B_L^j}\otimes\mathcal{H}_{B_R^j}$
and $\set{p_j}$ are probabilities.
The reverse statement of this case is not necessarily true: not all mixed states that satisfy the disentangling condition have the structure of \eqref{eq-ssa-structrepeat}.

To what extent does \eqref{eq-disent-logneg} constrain the structure of the manybody ground state of a topological phase?
More specifically, how are the disentangling condition \eqref{eq-disent-logneg} and the decompositions \eqref{eq-disent-decomp-pure-2} and \eqref{eq-ssa-structrepeat} related, if at all, for the Laughlin and Moore-Read ground states?
Because the relevant 
ground state of a topological phase is generally a direct sum over distinct topological sectors, $\ket{\Psi}=\sum_a\psi_a\ket{\Psi_a}$,
the applicability of the above results is less clear. For example, the degenerate ground state Hilbert space $\mathcal{H}=\bigoplus_a\mathcal{H}_a$ does not decompose into a tensor product $\mathcal{H}_A\otimes\mathcal{H}_B\otimes\mathcal{H}_C$ of local cylinder spaces. (See Ref.~\cite{PhysRevD.89.085012} for a related discussion in the context of entanglement entropy in gauge theory.)
This provides an {\it a priori} explanation for the necessity of restriction to a single topological sector $a$ to disentangle a ground state.

In this section, we will show that the Abelian Laughlin and untwisted sector Moore-Read states can be decomposed according to Eqs.~\eqref{eq-disent-decomp-pure-2} in the $AB_1CB_2$ geometry and \eqref{eq-ssa-structrepeat}  in the $ABCD$ geometry. 
On the other hand, the twisted sector Moore-Read states fail to decompose according to \eqref{eq-disent-decomp-pure-2} or \eqref{eq-ssa-structrepeat}. In other words, even when the disentangling condition \eqref{eq-disent-logneg} is satisfied, the ground state $\ket{\Psi_a}$ cannot be disentangled, if $a=\xi^{r+1/2}$ is an Ising twist field. 

The general failure of the Moore-Read state to disentangle stems from the non-Abelian nature of the twisted sectors. The Ising twist field carries non-trivial quantum dimension $d_a=\sqrt{2}>1$, associated to each Majorana zero mode $c^\sigma_{i,0}$. The ground state fixes the fermion parity $(-1)^{n'_i}=ic^R_{i-1,0}c^L_{i,0}$ of the pair of zero modes to be even at any given interface. However, the two zero modes do not belong to the same cylinder. When decomposing the ground state into a tensor product of local cylinder states, the ground state becomes a superposition of states with different cylinder fermion parities $(-1)^{\gamma_i}=ic^R_{i,0}c^L_{i,0}$. Since $(-1)^{n'_i}$ and $(-1)^{\gamma_i}$ do not commute, they do not share simultaneous eigenstates and the basis transformations between the two bases are generated by the non-diagonal  $F$-symbol in \eqref{IsingFsymbol}. Consequently, the zero mode part of the ground state, $\ket{f^{\rm zero}}$ in \eqref{twistedzero}, is a maximally entangled state where the cylinder fermion parities $(-1)^{\gamma_i}$ are scrambled. This state $\ket{f^{\rm zero}}$ does not decompose because the total fermion parity $\sum_i\gamma_i$ has a fixed value [see \eqref{totalparitycomparison}].

The main results of this section are summarized as the follows. The ground state of a fixed Abelian (Laughlin or untwisted Moore-Read) sector $|\Psi_a\rangle$ admits the factorization \eqref{Factorization1} in the $AB_1CB_2$ geometry. The reduced density matrix $\Tr_D|\Psi_a\rangle\langle\Psi_a|$ in the $ABCD$ geometry also factorizes according to \eqref{Factorization2}. These results are in agreement with the factorizability \eqref{eq-disent-decomp-pure-2} and \eqref{eq-ssa-structrepeat} (from \cite{he-vidal-2015} and \cite{gour-guo-2018}) as the ground state $|\Psi_a\rangle$ obeys the disentangling condition \eqref{eq-disent-logneg}. On the other hand, we show that the ground state $|\Psi_a\rangle$ [see \eqref{twistedGSZM}] of a non-Abelian twisted sector $a=\xi^{r+1/2}$ of the Moore-Read state fails to decompose. We demonstrate this by focusing on the zero mode sector and seeing that (i) the (partially traced) reduced density matrix \eqref{eq:prooftwistGSnondecomposed} is a mixed state and therefore the ground state must be entangled and (ii) the (partially transposed) density matrix \eqref{twistedrhodecomposition} in the $AB_1CB_2$ geometry does not factorize. Furthermore, we show that (iii) the reduced density matrix \eqref{ZMtwistedreduceddensitymatrix} in the $ABCD$ geometry also fails to disentangle. These results serve as concrete examples where \eqref{eq-disent-decomp-pure-2} and \eqref{eq-ssa-structrepeat} both fail to hold even though the disentangling condition \eqref{eq-disent-logneg} is satisfied.

\subsection{\texorpdfstring{$A B_1 C B_2$}{AB1CB2} Geometry}

We first consider the $AB_1CB_2$ torus geometry with $X_1=A$, $X_2=B_1$, $X_3=C$, and $X_4=B_2$.
Our discussion below will apply to both the Laughlin and Moore-Read states, with the understanding that Majorana fermion labels and projection operators are dropped for the Laughlin and untwisted Moore-Read states.

Since we are interested in measuring the entanglement $\mathcal{E}_{A:BC}(\rho_{ABC})$ between $A$ and its complement in \eqref{eq-disent-logneg}, we first show how the corresponding four-cylinder state can be viewed as a two-cylinder state on cylinders $A$ and $\bar A = B_1 \cup C \cup B_2$.
The torus state is given by 
\begin{align}
\begin{split}
    |\Psi\rangle 
    &=\sum_aP_a \psi_a\sum_{\mathcal{N}_{a,1}, \mathcal{N}_{a,2}}
   \frac{\lambda(\mathcal{N}_{a,1})\lambda(\mathcal{N}_{a,2})}{\sqrt{Z_{a,1}Z_{a,2}}}\\
   &\;\;\;\;\times|0_1\rangle'_{s_a}\otimes|\mathcal{N}_{a,1}\mathcal{N}_{a,2}\rangle_{X_1}\otimes|0_2\rangle'_{s_a}\\
   &\;\;\;\;\otimes|\mathcal{N}_{a,2}\rangle_{LX_2} \otimes |\hat{\Psi}_{\rm bulk}\rangle\otimes|\mathcal{N}_{a,1}\rangle_{RX_4}.
\end{split}\label{EquvilentTopoStateMR}
\end{align} 
The partition functions $Z_{a,i} = Z_a$ for all $i$ with $Z_a$ defined in \eqref{laughlinsectorapartitionfunction} for the Laughlin state and in \eqref{sectorapartitionfunction} and \eqref{twistednormsquared} for the untwisted and twisted sectors of the Moore-Read state; the additional $i$ indices are bookkeeping devices that associate these factors to their corresponding cylinders $X_i$.
The (unprojected) ``bulk" state is
\begin{align}
    \begin{split}
    |\hat{\Psi}_{\rm bulk}\rangle &=
    \sum_{\mathcal{N}_{a,3},\mathcal{N}_{a,4}}\prod_{i=3,4}\frac{\lambda(\mathcal{N}_{a,i})}{\sqrt{Z_{a,i}}}|\mathcal{N}_{a,3}\rangle_{RX_2}\otimes |0_3\rangle'_{s_a}\\
    &\;\;\;\;\otimes|\mathcal{N}_{a,3}\mathcal{N}_{a,4}\rangle_{X_3}
    \otimes|0_4\rangle'_{s_a}\otimes|\mathcal{N}_{a,4}\rangle_{LX_4}.\label{eq:defnBulk}
    \end{split}
\end{align} 
It is normalized: $\langle\hat{\Psi}_{\rm bulk}|\hat{\Psi}_{\rm bulk}\rangle=1$. 
Because of summing over all the internal indexes labeled by 
$\mathcal{N}_{a,3}$, $\mathcal{N}_{a,4}$, $|0_1\rangle'_{s_a}$, and $|0_4\rangle'_{s_a}$ in \eqref{eq:defnBulk}, $|\hat{\Psi}_{\rm bulk}\rangle\langle\hat{\Psi}_{\rm bulk}|$ acts as an identity operator when computing $\rho=|\Psi\rangle\langle\Psi|$. Therefore, $|\hat{\Psi}_{\mathrm{bulk}}\rangle$ makes no contribution to the measured entanglement.
This is the key observation for relating the four-cylinder and two-cylinder states. 
Note that $|0_i\rangle'_{s_a}=|0_i\rangle'$ when $s_a=1$ in the twisted sector where $|0_i\rangle'$ denotes the parity of the Majorana zero mode states at the interfaces between cylinders.
These states appear before the $F$-symbol basis change to states labeled by the parity of Majorana zero mode states on a given cylinder.
$|0_i\rangle'_{s_a}=1$ when $s_a=0$ in an Abelian or untwisted sector. 
One can then perform a basis transformation using the $F$-symbols and shift the labeling of Majorana fermion parity from the interfaces to the cylinders. 
The second identity in \eqref{EquvilentTopoStateMR} shows that the above four-cylinder torus state is equivalent to the the two-cylinder torus state. 
Thus, we may safely apply the results of the previous section for the entanglement negativity to conclude that only pure states is a specific topological sector, i.e., those states without long-range entanglement, satisfy the disentangling condition \eqref{eq-disent-logneg}.

In a specific sector $a$, the unprojected state $|\hat{\Psi}_a\rangle$ can be factorized as
\begin{align}
\begin{split}
    |\hat{\Psi}_a\rangle &=|\hat{\Psi}_{X_1(LX_2RX_4),a}\rangle\otimes|\hat{\Psi}_{(RX_2LX_4)X_3,a}\rangle,
    \label{Factorization1}
\end{split}
\end{align} 
where
\begin{align}
    \begin{split}
        &|\hat{\Psi}_{X_1(LX_2RX_4),a}\rangle=\\ &\;\;\;\;\sum_{\mathcal{N}_{a,1},\mathcal{N}_{a,2}}\frac{\lambda(\mathcal{N}_{a,1})}{\sqrt{Z_{a,1}}}\frac{\lambda(\mathcal{N}_{a,2})}{\sqrt{Z_{a,2}}}
        \hat{\Psi}_{X_1,a}^{(\mathcal{N}_{a,1}\mathcal{N}_{a,2})}
        \hat{\Psi}_{LX_2RX_4}^{\mathcal{N}_{a,2}\mathcal{N}_{a,1}},\\
        &|\hat{\Psi}_{(RX_2LX_4)X_3,a}\rangle=\\
        &\;\;\;\;\sum_{\mathcal{N}_{a,3},\mathcal{N}_{a,4}}\frac{\lambda(\mathcal{N}_{a,3})}{\sqrt{Z_{a,3}}}\frac{\lambda(\mathcal{N}_{a,4})}{\sqrt{Z_{a,4}}}
        \hat{\Psi}_{X_3,a}^{(\mathcal{N}_{a,3}\mathcal{N}_{a,4})}\hat{\Psi}_{RX_2LX_4}^{\mathcal{N}_{a,3}\mathcal{N}_{a,4}},
    \end{split}
\end{align} 
and
\begin{align}
\begin{split}
    \hat{\Psi}_{X_1,a}^{(\mathcal{N}_{a,1}\mathcal{N}_{a,2})}&= |0_1\rangle'_{s_a}\otimes 
    |\mathcal{N}_{a,1}\mathcal{N}_{a,2}\rangle_{X_1}\otimes|0_2\rangle'_{s_a},\\
    \hat{\Psi}_{X_3,a}^{(\mathcal{N}_{a,3}\mathcal{N}_{a,4})}&= |0_3\rangle'_{s_a}\otimes 
    |\mathcal{N}_{a,3}\mathcal{N}_{a,4}\rangle_{X_3}\otimes|0_4\rangle'_{s_a},\\
    \hat{\Psi}_{LX_2RX_4}^{\mathcal{N}_{a,2}\mathcal{N}_{a,1}} &=|\mathcal{N}_{a,1}\rangle_{RX_4}\otimes|\mathcal{N}_{a,2}\rangle_{LX_2},\\
    \hat{\Psi}_{RX_2LX_4}^{\mathcal{N}_{a,3}\mathcal{N}_{a,4}} &=|\mathcal{N}_{a,3}\rangle_{RX_2}\otimes|\mathcal{N}_{a,4}\rangle_{LX_4}.
\end{split}
\end{align} 
Here, $|\hat{\Psi}_{X_1(LX_2RX_4),a}\rangle$ and $|\hat{\Psi}_{(RX_2LX_4)X_3,a}\rangle$ are normalized.
\eqref{Factorization1} is the desired factorization for the Laughlin state, where there is no projection operator.
For an untwisted sector Moore-Read state, the projection operator $P_a$ can be decomposed into cylinder state projection operators $P_{a,X_i}$, which in turn decompose into left and right edge projection operators as $P_{a,X_i} = P_{a, i}P_{a,i + 1}$.
Including these factorized projection operators, we find the untwisted Moore-Read ground state wavefunction in a specific sector $a$ can be disentangled.

In the twisted sector of the Moore-Read phase, the projection operator $P_{X_i}$ does not factorize into left and right edge components.
Further division of a given cylinder into sub-cylinders does not appear to help, as the resulting Hilbert space is not a simple tensor product.
Thus, the corresponding manybody wavefunction does not factorize as \eqref{eq-disent-decomp-pure-2}.
Although the corresponding pure state density matrix can be written in a form similar to \eqref{eq-ssa-structrepeat}, we will show in the following that the factorization of $\mathcal{H}_B=\mathcal{H}_{B_L}\otimes\mathcal{H}_{B_R}$ fails.

By splitting cylinder $B_1$ into sub-cylinders $X_2, X_3$ and $B_2$ into sub-cylinders $X_5, X_6$ (see Fig.~\ref{fig:AB1CB2geom}), the torus ground state
of a fixed twisted sector $a=\xi^{r+1/2}$ is 
\begin{align}
    |\Psi_a\rangle &=\sum_{\vec{\mathcal{N}}_a,\vec{\gamma}}c_a(\vec{\mathcal{N}}_a) P_a(\vec{\mathcal{N}}_a,\vec{\gamma}) \bigotimes_{i=1}^6
    |\mathcal{N}_{a,i}\mathcal{N}_{a,i+1}\gamma_i\rangle_{X_i},\label{eq:AB1CB2torusState}
\end{align}
where $c_a(\vec{\mathcal{N}}_a) \equiv \prod_{i=1}^{6}\frac{\lambda(\mathcal{N}_{a,i})}{\sqrt{Z_{a}}}$ and the projection operator eigenvalues $P_a(\vec{\mathcal{N}}_a,\vec{\gamma}) = \prod_{i=1}^6  P_a(\mathcal{N}_{a,i}, \gamma_i)$ are defined in \eqref{eq:twistSecProjector0}.
\begin{figure}[h]
\centering
\includegraphics[width=0.4\textwidth]{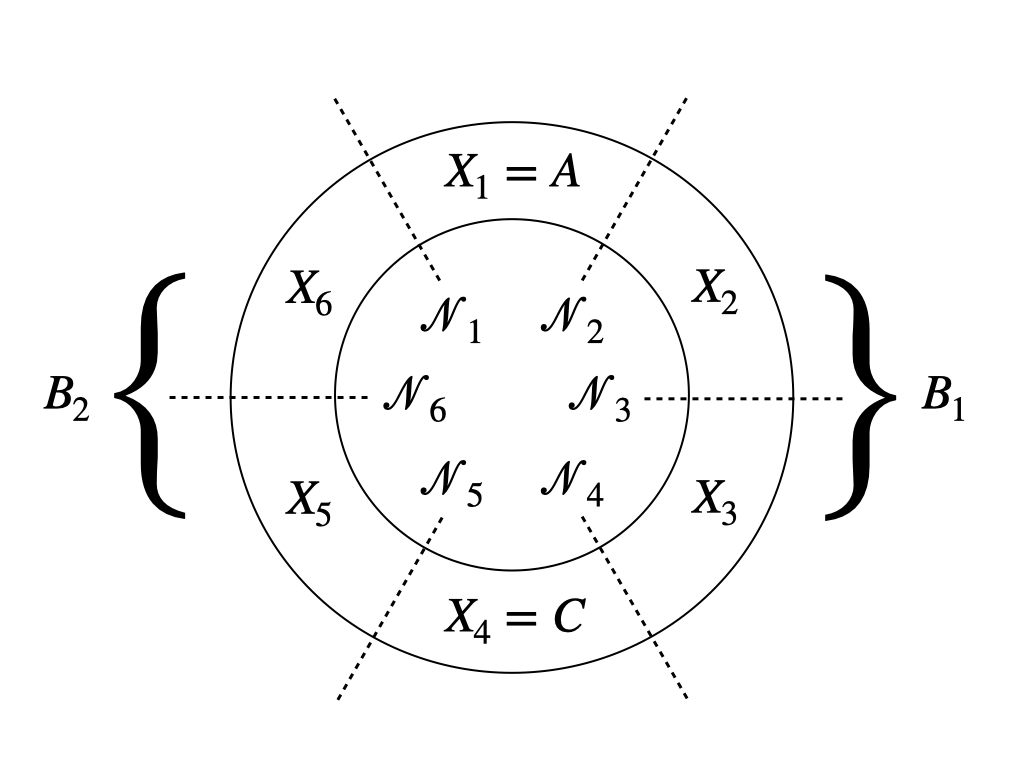}
\caption{A torus is divided by the dashed lines into the $AB_1CB_2$ geometry: Region $A$ is the $X_1$ cylinder while region $C$ the $X_4$ cylinder. Region $B_1$ is the union of $X_2$, and $X_3$; and $B_2$ the union of $X_5$ and $X_6$. The collective modes $\mathcal{N}_{i}$ are defined on the $i^{\mathrm{th}}$ interface as in \eqref{generalN} or \eqref{generalNwithfermionosc}.}\label{fig:AB1CB2geom}
\end{figure} 
We find that under the restricted sum $\sum_{i=1}^6\gamma_i=1$ mod 2 [see \eqref{totalparitycomparison}], we can split and restrict the sum over $\gamma_1, \gamma_2, \gamma_6$ and the sum over $\gamma_3, \gamma_4, \gamma_5$ according to the parity \begin{align}s=-(-1)^{\gamma_1+\gamma_2+\gamma_6}=(-1)^{\gamma_3+\gamma_4+\gamma_5}=\pm.\label{eq:parityindex}\end{align} The projections $P_{a,X_1}P_{a,X_2}P_{a,X_6}$ and $P_{a,X_3}P_{a,X_4}P_{a,X_5}$ from \eqref{eq:twistSecProjector0} both require $s=(-1)^{\mathcal{N}_{a,3}+\mathcal{N}_{a,6}}$.

The density matrix $\rho^{\rm twist}=|\Psi_a\rangle\langle\Psi_a|$ from the twisted sector $a=\xi^{r+1/2}$ can now be factorized as
\begin{align}
    \begin{split}
        \rho^{\rm twist} &= \sum_j p_j\rho_{X_1(X_2X_6)}^j\otimes \rho_{(X_3X_5)X4}^j.\label{eq:non-localindexintwistedSec}
    \end{split}
\end{align} 
Here the summation index $j$ is an abbreviation for the collection of quantities \begin{align}j=\left\{\left.\begin{array}{*{20}l}\mathcal{N}_{a,3},\mathcal{N}_{a,6},s\\\mathcal{N}'_{a,3},\mathcal{N}'_{a,6},s'\end{array}\right|\begin{array}{*{20}l} s=(-1)^{\mathcal{N}_{a,3},\mathcal{N}_{a,6}}\\ s'=(-1)^{\mathcal{N}'_{a,3},\mathcal{N}'_{a,6}}\end{array}\right\}.\end{align}
The probabilities in the density matrix are \begin{align}p_j&\equiv \prod_{i=3,6}\frac{\lambda(\mathcal{N}_{a,i})\lambda(\mathcal{N'}_{a',i})}{\sqrt{Z_{a,i}Z_{a,i}}}\end{align} so $\sum_jp_j=1$. The density matrix components are
\begin{widetext}
\begin{align}
    \begin{split}
        \rho^j_{X_1(X_2X_6)}&=\sum_{\{h\}}
        \prod_{i=1,2}\frac{\lambda(\mathcal{N}_{a,i})\lambda(\mathcal{N'}_{a,i})}{\sqrt{Z_{a,i}Z_{a,i}}}
        \sum_{\substack{\gamma_1,\gamma_2,\gamma_6\\s=-(-1)^{\gamma_1+\gamma_2+\gamma_6}}}
        \sum_{\substack{\gamma'_1,\gamma'_2,\gamma'_6\\s'=-(-1)^{\gamma'_1+\gamma'_2+\gamma'_6}}}
        \\
        &\;\;\;\;\times
        P_{a,X_6}|\mathcal{N}_{a,6}\mathcal{N}_{a,1}\gamma_6\rangle\langle\mathcal{N}'_{a,6},\mathcal{N}'_{a,1},\gamma'_6|P_{a,X_6}\\
        &\;\;\;\;\otimes
        P_{a,X_1}|\mathcal{N}_{a,1}\mathcal{N}_{a,2}\gamma_1\rangle\langle
        \mathcal{N}'_{a,1}\mathcal{N}'_{a,2}\gamma_1'|_{X_1}P_{a,X_1}\\
        &\;\;\;\;\otimes
        P_{a,X_2}|\mathcal{N}_{a,2}\mathcal{N}_{a,3}\gamma_2\rangle\langle\mathcal{N}'_{a,2},\mathcal{N}'_{a,3},\gamma'_2|P_{a,X_2}
    \end{split}
\end{align} 
with $\{h\}\equiv\{\mathcal{N}_{a,1}\mathcal{N}'_{a,1},\mathcal{N}_{a,2},\mathcal{N}'_{a,2}\}$, and         
\begin{align}
    \begin{split}       
        \rho^j_{(X_3X_5)X_4}&=\sum_{\{e\}}
        \prod_{i=4,5}\frac{\lambda(\mathcal{N}_{a,i})\lambda(\mathcal{N'}_{a,i})}{\sqrt{Z_{a,i}Z_{a,i}}}
        \sum_{\substack{\gamma_3,\gamma_4,\gamma_5\\s=(-1)^{\gamma_3+\gamma_4+\gamma_5}}}
        \sum_{\substack{\gamma'_3,\gamma'_4,\gamma'_5\\s'=(-1)^{\gamma'_3+\gamma'_4+\gamma'_5}}}\\
        &\;\;\;\;\times 
        P_{a,X_3}|\mathcal{N}_{a,3}\mathcal{N}_{a,4}\gamma_3\rangle\langle\mathcal{N}'_{a,3},\mathcal{N}'_{a,4},\gamma'_3|P_{a,X_3}\\
        &\;\;\;\;\otimes
        P_{a,X_4}|\mathcal{N}_{a,4}\mathcal{N}_{a,5}\gamma_4\rangle\langle
        \mathcal{N}'_{a,4}\mathcal{N}'_{a,5}\gamma_4'|P_{a,X_4}\\
        &\;\;\;\;\otimes
        P_{a,X_5}|\mathcal{N}_{a,5}\mathcal{N}_{a,6}\gamma_5\rangle\langle\mathcal{N}'_{a,6},\mathcal{N}'_{a,6},\gamma'_5|P_{a,X_5},
    \end{split}
\end{align} 
with $\{e\}\equiv\{\mathcal{N}_{a,4}\mathcal{N}'_{a,4},\mathcal{N}_{a,5},\mathcal{N}'_{a,5}\}$.
All density matrix components $\rho$'s are Hermitian and have unit trace.

We notice that the ground state \eqref{eq:AB1CB2torusState} and the density matrix \eqref{eq:non-localindexintwistedSec} of any of the twisted sectors are {\em not} factorizable and cannot be expressed in \eqref{eq-disent-decomp-pure-2} and \eqref{eq-ssa-structrepeat}. This is because the summation index $j$ involves the parities $s,s'$, which specify the fermion parity of half of the torus $AB_{\rm top}=X_1(X_2X_6)$ and $B_{\rm bottom}C=(X_3X_5)X_4$. These parity indices cannot be absorbed entirely into $B_{\rm top}$ and $B_{\rm bottom}$, and therefore the Hilbert space decomposition $\mathcal{H}_B=\bigoplus_j\mathcal{H}_{B_{\rm top}^j}\otimes\mathcal{H}_{B_{\rm bottom}^j}$ is {\em not} satisfied. We show this failure of decomposition of the ground state below by focusing on the zero mode sector \begin{align}|{\rm GS}\rangle = \frac{1}{4\sqrt{2}} \sum_{\{\gamma_i\}|_C}|\gamma_1\gamma_2\gamma_3\gamma_4\gamma_5\gamma_6\rangle\end{align} where the constraint $C$ requires that $\sum_{i=1}^6\gamma_i =1$ mod 2.

\begin{proof}
 We first see that the ground state can be re-expressed as \begin{align}|{\rm GS}\rangle=\frac{1}{\sqrt{2}}\sum_{s=\pm}\left(\frac{1}{2}\sum_{\substack{\gamma_1,\gamma_2,\gamma_6\\s=-(-1)^{\gamma_1+\gamma_2+\gamma_6}}}|\gamma_1\gamma_2\gamma_6\rangle\right)\otimes\left(\frac{1}{2}\sum_{\substack{\gamma_3,\gamma_4,\gamma_5\\s=(-1)^{\gamma_3+\gamma_4+\gamma_5}}}|\gamma_3\gamma_4\gamma_5\rangle\right)\label{twistedGSZM}.\end{align} To show that the above ground state does not decompose according to \eqref{eq-disent-decomp-pure-2}, we assume the contrary that $|{\rm GS}\rangle=|{\rm GS}\rangle_{AB_{\rm top}}\otimes|{\rm GS}\rangle_{B_{\rm bottom}C}$. This would imply the reduced density matrix $\rho_{AB_{\rm top}}=\Tr_{B_{\rm bottom}C}(\rho)$ is a pure state, where $\rho=|{\rm GS}\rangle\langle {\rm GS}|$. However, from \eqref{twistedGSZM}, \begin{align}\Tr_{B_{\rm bottom}C}|{\rm GS}\rangle\langle {\rm GS}|=\frac{1}{8}\sum_{\substack{(-1)^{\gamma_1+\gamma_2+\gamma_6}\\=(-1)^{\gamma'_1+\gamma'_2+\gamma'_6}}}|\gamma_1\gamma_2\gamma_6\rangle\langle\gamma'_1\gamma'_2\gamma'_6|\label{eq:prooftwistGSnondecomposed}
 \end{align} has spectrum $\{1/2,1/2,0,0,0,0,0,0\}$ and is a mixed state. Therefore the assumption $|{\rm GS}\rangle=|{\rm GS}\rangle_{AB_{\rm top}}\otimes|{\rm GS}\rangle_{B_{\rm bottom}C}$ must be false, and the ground state does not disentangle according to \eqref{eq-disent-decomp-pure-2}.

Furthermore, we consider the density matrix $\rho=|{\rm GS}\rangle\langle {\rm GS}|$, 
\begin{align}\begin{split}
        \rho &= \frac{1}{32} \sum_{\vec{\gamma}|_C}\sum_{\vec{\gamma}'|_C}
        |\gamma_1\gamma_2\gamma_6\rangle
        \langle\gamma'_1\gamma'_2\gamma'_6|\otimes
        |\gamma_3\gamma_4\gamma_5\rangle
        \langle\gamma'_3\gamma'_4\gamma'_5|\\
        &=\frac{1}{32}\sum_{s,s'=\pm}\left(
        \sum_{\substack{s=-(-1)^{\gamma_1+\gamma_2+\gamma_6}\\
        s'=-(-1)^{\gamma'_1+\gamma'_2+\gamma'_6}
        }}
        |\gamma_1\gamma_2\gamma_6\rangle
        \langle\gamma'_1\gamma'_2\gamma'_6|\right)\otimes\left(
        \sum_{\substack{s=(-1)^{\gamma_3+\gamma_4+\gamma_5}\\
        s'=(-1)^{\gamma'_3+\gamma'_4+\gamma'_5}
        }}
        |\gamma_3\gamma_4\gamma_5\rangle
        \langle\gamma'_3\gamma'_4\gamma'_5|\right).
    \end{split}\label{twistedrhodecomposition}
\end{align} We define the density matrix components
\begin{align}\begin{split}
    \rho_{AB_{\rm top}}^{ss'} = \frac{1}{4}\sum_{\substack{s=-(-1)^{\gamma_1+\gamma_2+\gamma_6}\\
    s'=-(-1)^{\gamma'_1+\gamma'_2+\gamma'_6}
    }}|\gamma_1\gamma_2\gamma_6\rangle\langle\gamma'_1\gamma'_2\gamma'_6|,\\ \rho_{B_{\rm bottom}C}^{ss'} = \frac{1}{4}\sum_{\substack{s=(-1)^{\gamma_3+\gamma_4+\gamma_5}\\
        s'=(-1)^{\gamma'_3+\gamma'_4+\gamma'_5}
        }}
        |\gamma_3\gamma_4\gamma_5\rangle
        \langle\gamma'_3\gamma'_4\gamma'_5|.
\end{split}\label{twistedrhocomponents}\end{align} These components have unit trace only when $s=s'$, and have vanishing trace when $s\neq s'$. Therefore, \eqref{twistedrhodecomposition} does {\em not} admit a density matrix decomposition  \eqref{eq-ssa-structrepeat}. Moreover, even when $s=s'$, the parity index cannot be absorbed entirely in $B_{\rm top}$ and $B_{\rm bottom}$. To see this, we assume the contrary that the density matrix components decompose,  $\rho_{AB_{\rm top}}^{s}=\rho_{A}\otimes\rho_{B_{
\rm top}}^{s}$, where $\rho_A$ and $\rho_{B_{\rm top}}$ have unit trace. This implies $\rho_{A} = \Tr_{B_{\rm top}}(\rho_{AB_{\rm top}}^{s})$ and $\rho_{B_{\rm top}}^{s} = \Tr_{A}(\rho_{AB_{\rm top}}^{s})$. By taking the partial traces in \eqref{twistedrhocomponents},
\begin{align}
    \begin{split}
         \rho_{A}  =\Tr_{B_{\rm top}}(\rho_{AB_{\rm top}}^{s})=\frac{1}{2}\sum_{\gamma_1=0,1}
         |\gamma_1\rangle\langle\gamma_1|\quad\mbox{and}\quad
         \rho_{B_{\rm top}}^s =\Tr_{A}(\rho_{AB_{\rm top}}^{s})=\frac{1}{4}\sum_{\substack{(-1)^{\gamma_2+\gamma_6}\\=(-1)^{\gamma'_2+\gamma'_6}}}|\gamma_2\gamma_6\rangle\langle\gamma'_2\gamma'_6|.
    \end{split}
\end{align}
The product $\rho_A\otimes\rho_{B_{\rm top}}^s=\Tr_{B_{\rm top}}(\rho_{AB_{\rm top}}^{s})\otimes\Tr_{A}(\rho_{AB_{\rm top}}^{s})$ is \begin{align}\frac{1}{8}\sum_{\gamma_1=0,1}\sum_{\substack{(-1)^{\gamma_2+\gamma_6}\\=(-1)^{\gamma'_2+\gamma'_6}}}|\gamma_1\gamma_2\gamma_6\rangle\langle\gamma'_1\gamma'_2\gamma'_6|,
\label{eq:AB1CB2desityproofondecomposed}
\end{align} which contradicts \eqref{twistedrhocomponents}. Therefore, the assumption $\rho_{AB_{\rm top}}^{s}=\rho_{A}\otimes\rho_{B_{
\rm top}}^{s}$ must be false. Similarly, $\rho_{B_{\rm bottom}C}^{s}$ is also not factorizable.

\end{proof}

\end{widetext}

\subsection{\texorpdfstring{$ABCD$}{ABCD} Geometry}.

As we found in the proof of Lemma \ref{lemma2}, the trace over degrees of freedom in cylinder $D$ results in a reduced density matrix $\rho_{ABC}$ that is a direct sum over each of the topological sectors.
We may therefore consider the decomposition \eqref{eq-ssa-structrepeat} for the Laughlin and untwisted Moore-Read states separately from that of the twisted Moore-Read states. 
We will show how the latter set of states admits a refinement of the decomposition \eqref{eq-ssa-structrepeat}.
In both cases, the reduced density matrices $\rho_{ABC}$ saturate the strong subadditivity relation of the entanglement entropy.

\subsubsection{Laughlin and Untwisted Sector Moore-Read States}

We begin with a fixed pure torus state with anyon flux $a$,
\begin{align}
\begin{split}
    |\Psi\rangle &= \sum_{\vec{\mathcal{N}}_a}P_a(\vec{\mathcal{N}}_a)
    \prod_{i=1}^{4}\frac{\lambda(\mathcal{N}_{a,i})}{\sqrt{Z_{a, i}}}\\
    &\;\;\;\;\times|\mathcal{N}_{a,1}\mathcal{N}_{a,2}\rangle_{X_1}\otimes
    |\mathcal{N}_{a,2}\mathcal{N}_{a,3}\rangle_{X_2}\\
    &\;\;\;\;\otimes|\mathcal{N}_{a,3}\mathcal{N}_{a,4}\rangle_{X_3}\otimes
    |\mathcal{N}_{a,4}\mathcal{N}_{a,1}\rangle_{X_4}.
\end{split}
\end{align} 
Here we are taking $X_1=A$, $X_2=B$, $X_3=C$, and $X_4=D$.
$P_a(\vec{\mathcal{N}}_a)$ is the product of projection operator eigenvalue for the four cylinders. 
For the Laughlin state, $P_a(\vec{\mathcal{N}}_a) = 1$.
For the untwisted sectors of the Moore-Read state, a cylinder projection operator eigenvalue can be factorized into a product of left and right edge projection operator eigenvalues for each cylinder, i.e., $P_{a,X_i}({\cal N}_{a,i}, {\cal N}_{a,i+1})=P_{a,i}({\cal N}_{a,i}) P_{a,i+1}({\cal N}_{a,i+1})$ with $P_{a,i}({\cal N}_{a,i})$ given in \eqref{untwistedprojectioneigenvalues}.
The density matrix $\rho_{ABC} = \tr_{X_4}|\Psi\rangle\langle\Psi|$ obtained by tracing out $X_4$ is
\begin{align}
    \begin{split}
        \rho_{ABC}&=
        \sum_{\mathcal{N}_{a,1},\mathcal{N}_{a,2},\mathcal{N}'_{a,2}}
        \frac{\lambda^2(\mathcal{N}_{a,1})}{Z_{a,1}}\frac{\lambda(\mathcal{N}_{a,2})}{\sqrt{Z_{a,2}}}\frac{\lambda(\mathcal{N}'_{a,2})}{\sqrt{Z_{a,2}}}\\
        &\;\;\;\;\times P_{a,X_1}|\mathcal{N}_{a,1}\mathcal{N}_{a,2}\rangle\langle\mathcal{N}_{a,1}\mathcal{N}'_{a,2}|_{X_1}P_{a,X_1}\\
        &\;\;\;\;\otimes 
        P_{a,2}|\mathcal{N}_{a,2}\rangle\langle
        \mathcal{N}'_{a,2}|_{LX_2}P_{a,2}\\
        &\;\;\;\;\otimes
        \sum_{\mathcal{N}_{a,3},\mathcal{N}'_{a,3},\mathcal{N}_{a,4}}
        \frac{\lambda^2(\mathcal{N}_{a,4})}{Z_{a,4}}\frac{\lambda(\mathcal{N}_{a,3})}{\sqrt{Z_{a,3}}}\frac{\lambda(\mathcal{N}'_{a,3})}{\sqrt{Z_{a,3}}}\\
        &\;\;\;\;\times P_{a,3}|\mathcal{N}_{a,3}\rangle\langle
        \mathcal{N}'_{a,3}|_{RX_2}P_{a,3}\\
        &\;\;\;\;\otimes
        P_{a,X_3}|\mathcal{N}_{a,3}\mathcal{N}_{a,4}\rangle\langle\mathcal{N}_{a,3}\mathcal{N}'_{a,4}|_{X_1}P_{a,X_3}\label{ABCdensitymatrixuntwist},
    \end{split}
\end{align} 
where the partition functions $Z_{a,i} = Z_a$ for all $i$ (for the uniform states we consider) with $Z_a$ defined in \eqref{laughlinsectorapartitionfunction} for the Laughlin state and in \eqref{sectorapartitionfunction} for the untwisted sectors of the Moore-Read state. 
By inspection, this admits the decomposition \eqref{eq-ssa-structrepeat}:
\begin{align}
    \rho_{ABC}=\rho_{A(LB)}^a\otimes\rho_{(RB)C}^a\label{Factorization2},
\end{align}
The density matrices associate with \eqref{Factorization2} are,
\begin{align}
    \begin{split}
        &\rho_{A(LB)}^a=\\
        &\;\;\;\;\sum_{\mathcal{N}_{a,2},\mathcal{N}'_{a,2}}
        \frac{\lambda(\mathcal{N}_{a,2})\lambda(\mathcal{N}'_{a,2})}{Z_{a,2}}\rho^{\mathcal{N}_{a,2}\mathcal{N}'_{a,2}}_{A,a}\rho^{\mathcal{N}_{a,2}\mathcal{N}'_{a,2}}_{LB,a}
    \end{split}\label{ZMtwistedreduceddensitymatrix}
\end{align}
where
\begin{align}
    \begin{split}
        &\rho_{A,a}^{\mathcal{N}_{a,2}\mathcal{N}'_{a,2}}=\\
        &\;\;\;\;\sum_{\mathcal{N}_{a,1}}\frac{\lambda^2(\mathcal{N}_{a,1})}{Z_{a,1}}P_{a,X_1}|\mathcal{N}_{a,1}\mathcal{N}_{a,2}\rangle\langle\mathcal{N}_{a,1}\mathcal{N}'_{a,2}|_{X_1}P_{a,X_1},\\
        &\rho^{\mathcal{N}_{a,2}\mathcal{N}'_{a,2}}_{LB,a}=
        P_{a,2}|\mathcal{N}_{a,2}\rangle\langle\mathcal{N}'_{a,2}|_{LX_2}P_{a,2}.
    \end{split}
\end{align}
$\rho_{(RB)C,a}$ has a similar decomposition,
\begin{align}
    \begin{split}
        &\rho_{(RB)C}^a=\\
        &\;\;\;\;\sum_{\mathcal{N}_{a,3},\mathcal{N}'_{a,3}}\frac{\lambda(\mathcal{N}_{a,3})\lambda(\mathcal{N}'_{a,3})}{Z_{a,3}}\rho^{\mathcal{N}_{a,3}\mathcal{N}'_{a,3}}_{RB,a}\rho^{\mathcal{N}_{a,3}\mathcal{N}'_{a,3}}_{C,a}.
    \end{split}
\end{align}

\subsubsection{Twisted Sector Moore-Read States}
For the twisted sectors of the Moore-Read state, we need to split $B$ into two consecutive cylinders $B_I$ and $B_{II}$. Specifically, we take $X_1=A$, $X_2=B_I$, $X_3=B_{II}$, $X_4=C$, $X_5=D_1$ and $X_6=D_2$ (see Fig.~\ref{fig:ABCDgeom}) and begin with the generic twisted sector pure state,
\begin{figure}[h]
\centering
\includegraphics[width=0.4\textwidth]{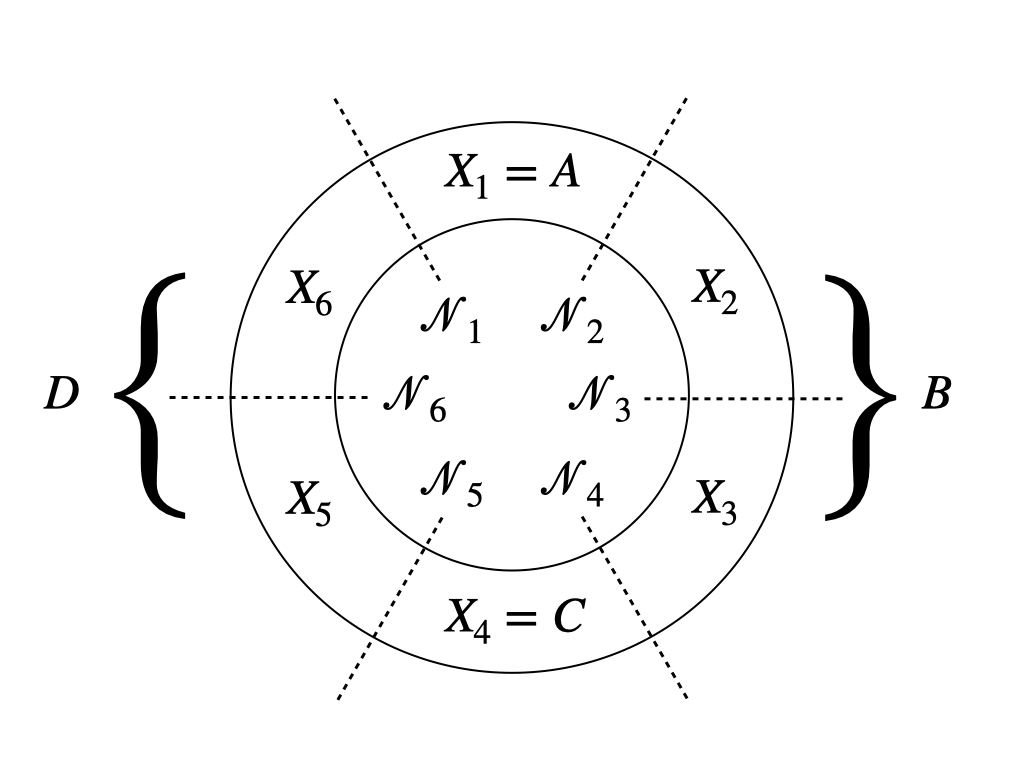}
\caption{A torus is divided by the dashed lines into the $ABCD_1D_2$ geometry: Region $A$ is the $X_1$ cylinder while region $C$ the $X_4$ cylinder. Region $B$ is the union of $X_2$ and $X_3$, and $D$ the union of $X_5$ and $X_6$.}\label{fig:ABCDgeom}
\end{figure} 

The ground state of a fixed twisted sector $a=\xi^{r+1/2}$ is again described by \eqref{eq:AB1CB2torusState}.
We will show the reduced density matrix $\rho^{\rm twist}_{ABC}$, after tracing out subsystem $D$, cannot factorize according \eqref{eq-ssa-structrepeat}. Similar to the previous $AB_1CB_2$ geometry in the last subsection, it suffice to focus on the zero mode sector. The ground state in the zero mode sector is $|{\rm GS}\rangle=\sum_{\{\gamma_i\}|_C}|\gamma_1\ldots\gamma_6\rangle/(4\sqrt{2})$, where the sum is restricted by $\sum_{i=1}^6\gamma_i=1$ mod 2. 

The reduced density matrix is \begin{align}\begin{split}\rho^{\rm twist}_{ABC}&=\Tr_D\left(|{\rm GS}\rangle\langle {\rm GS}|\right)\\&=\frac{1}{16}\sum_{\substack{\gamma_1+\ldots+\gamma_4\\=\gamma'_1+\ldots+\gamma'_4}}|\gamma_1\ldots\gamma_4\rangle\langle\gamma'_1\ldots\gamma'_4|.\label{twistedABCDrdm}\end{split}\end{align} To show that it does not decompose, we follow a similar procedure to before and assume the contrary that $\rho^{\rm twist}_{ABC}=\sum_j p_j \rho_{AB_I^j}\otimes\rho_{B_{II}^jC}$, where $p_j$ are probabilities satisfying $\sum_jp_j=1$. 
Tracing over subsystem $B$, the factorization would imply $\Tr_B(\rho_{ABC}^{\rm twist})=\rho_A\otimes\rho_C$. At the same time, from \eqref{twistedABCDrdm}, \begin{align}\begin{split}
&\Tr_B(\rho_{ABC}^{\rm twist})=\frac{1}{4}\sum_{\substack{\gamma_1+\gamma_4\\=\gamma'_1+\gamma'_4}}|\gamma_1\gamma_4\rangle\langle\gamma'_1\gamma'_4|,\\
&\rho_A=\Tr_{BC}\left(\rho^{\rm twist}_{ABC}\right)=\frac{1}{2}\sum_{\gamma_1=0,1}|\gamma_1\rangle\langle\gamma_1|,\\&\rho_C=\Tr_{AB}\left(\rho^{\rm twist}_{ABC}\right)=\frac{1}{2}\sum_{\gamma_4=0,1}|\gamma_4\rangle\langle\gamma_4|.\end{split}\end{align} However, this would lead to a contradiction because $\rho_A\otimes\rho_C=\frac{1}{4}\sum_{\gamma_1\gamma_4}|\gamma_1\gamma_4\rangle\langle\gamma_1\gamma_4|$, which disagrees with $\Tr_B(\rho_{ABC}^{\rm twist})$ in the equation above. Therefore, the assumption that the reduced density matrix decomposes, $\rho^{\rm twist}_{ABC}=\sum_j p_j \rho_{AB_I^j}\otimes\rho_{B_{II}^jC}$, must be false.

\section{Discussion and Conclusion}
\label{sectionfive}

In this paper, we studied multipartite entanglement in the Laughlin and Moore-Read ground state wavefunctions.
Our main results for the entanglement negativity of these states are summarized in Eqs.~\eqref{torusresult} and \eqref{cylinderresult}.
From these entanglement negativities, we constructed a disentangling condition \eqref{longrangeentanglementgeneral} for whether states can be disentangled, i.e., decomposed according to either \eqref{eq-disent-decomp-pure} or \eqref{eq-ssa-struct}.
The disentangling condition is only satisfied by states in a definite topological sector.
We found the disentangling condition to be a necessary and sufficient condition to disentangle the Laughlin and untwisted sector Moore-Read states.

Despite satisfying the disentangling condition, a twisted sector Moore-Read ground state wavefunction on the torus cannot be disentangled.
The obstruction is due to the lack of a tensor product decomposition of the twisted sector torus Hilbert space into appropriate subspaces. It would be interesting to find a generalization of the disentangling condition, perhaps one that involves the partial time-reversal \cite{2017PhRvB..95p5101S} or anyonic partial transpose \cite{2020arXiv201202222S}, that is sensitive to this particular obstruction to wavefunction disentanglement.

Our results rely on the cut and glue construction of topological ground states.
In this approach, the correlation length is zero.
With finite correlation length, we expect exponentially suppressed corrections to appear in the disentangling condition. 
It would also be interesting to consider the disentangling condition at phase transitions where the correlation length is infinite.

We focused on the Laughlin and Moore-Read topological states.
We expect that our entanglement negativity results hold for more general topological states in $2+1$ dimensions, such as those phases hosting metaplectic anyons \cite{2013PhRvB..87p5421H} and Fibonacci anyons \cite{ChetanSimonSternFreedmanDasSarma}. It is unclear to us whether the corresponding wavefunctions for such states might disentangle, as the fusion rule structure of general states is more intricate than the Laughlin and Moore-Read states.
Fracton orders in $3+1$ dimensions \cite{doi:10.1146/annurev-conmatphys-031218-013604} have similar entanglement signatures as their lower-dimensional ``conventional" topologically ordered counterparts \cite{2018arXiv180310426S}.
Recent work has shown how certain types of fracton order obtain from coupled-wire constructions \cite{PhysRevResearch.3.023123,PhysRevB.103.205301} or from infinite-component (2+1)-dimensional Chern-Simons gauge theory \cite{2020arXiv201008917M}.
The multipartite entanglement characteristics of this order are yet to be understood.

\section*{Acknowledgments}

We thank Sudip Chakravarty for useful conversations and correspondence.
This material is based upon work supported by the U.S. Department of Energy, Office of Science, Office of Basic Energy Sciences under Award No. DE-SC0020007. J.C.Y.T is supported by the National Science Foundation under Grant No. DMR-1653535.

\appendix

\section{Modular and Character Functions}\label{ModularFunctions}
\subsection{Modular functions} 
For all the modular functions in this paper, we will follow directly the notation from Sohal {\it et~al.} \cite{sohal-2020}.
The nome $q$ is defined by 
\begin{align}
    q=e^{2\pi\tau}.
\end{align}
For the fictitious inverse temperature $\beta=1/T$, the modular parameter $\tau$ for our physical systems is defined by 
\begin{align}
    \tau = i\tau_2 = \frac{i\beta v_e}{L},
\end{align} where $\tau \in \mathbb{C}$.
The Dedekind's $\eta$ function is defined as 
\begin{align}
    \eta(\tau)=q^{1/24}\prod_{n=1}^\infty (1-q^n).\label{EtaFunctionForm}
\end{align} The modular transformations (with $\tau\rightarrow -1/\tau$ and $\tau\rightarrow \tau +1$) of the eta function gives us the following relations
\begin{align}
    \eta(-1/\tau) &= \sqrt{-i\tau}\eta(\tau),\\
    \eta(\tau+1)&=e^{i\pi/12}\eta(\tau).
\end{align}
The more general Jacobi theta functions are defined by 
\begin{align}
    \theta^\alpha_\beta(\tau)&=\sum_{n\in\mathbb{Z}}q^{\frac{1}{2}(n+\alpha)^2}e^{2\pi i(n+\alpha)\beta}.\label{AlmostGeneralThetaFunc}
\end{align} Under modular transformations, the theta functions satisfy the following relations
\begin{align}
    \theta^\alpha_\beta(-1/\tau) &= \sqrt{-i\tau}e^{2\pi i\alpha\beta}\theta^\beta_{-\alpha}(\tau),\\
    \theta^\alpha_\beta(\tau+1) &= e^{-\pi i\alpha(\alpha-1)}\theta^\alpha_{\alpha+\beta-\frac{1}{2}}(\tau).
\end{align}
Recall that $\tau = i\tau_2$ where $\tau_2\in\mathbb{R}^+$. As $\tau_2\rightarrow \infty$, the modular functions approach asymptomatic values
\begin{align}
    \lim_{\tau_2\rightarrow\infty}\eta(i/\tau_2)&=q^{1/24},\\
    \lim_{\tau_2\rightarrow\infty}\theta^\alpha_\beta(i/\tau_2)&=\delta_{\alpha,0}.
\end{align}

The standard theta functions can be written in terms of the above general theta functions in the following form,
\begin{align}
    \theta_2(\tau)&=\sum_{n\in\mathbb{Z}}q^{\left(n+\frac{1}{2}\right)^2}/2=\theta^{1/2}_0(\tau),\\
    \theta_3(\tau)&=\sum_{n\in\mathbb{Z}}q^{n^2/2}=\theta^0_0(\tau),\\
    \theta_4(\tau)&=\sum_{n\in\mathbb{Z}}(-1)^nq^{n^2/2}=\theta^0_{1/2}(\tau).
\end{align}

Include here also three product representation of theta functions \begin{align}
    \theta_2(\tau)&=2\sqrt[4]{q}\prod_{j>0}(1-q^{2j})(1+q^{2j})^2\label{Theta2},\\
    \theta_3(\tau)&=\prod_{j>0}(1-q^{2j})(1+q^{2j-1})^2,\label{ProductTheta3}\\
    \theta_4(\tau)&=\prod_{j>0}(1-q^{2j})(1-q^{2j-1})^2\label{ProductTheta4}.
\end{align}

\subsection{Character functions} 
From the entanglement Hamiltonian for fermions, its partition function under anti-periodic boundary condition ($k=\frac{2\pi}{L}(j+\frac{1}{2})$), the fermionic partition function is instead
\begin{align}
    \begin{split}
        &\;\;\;\;\sum_{n_{i,k}=0,1}e^{-\sum_{k>0}\beta \tilde{v}_ek(n_{i,k}+1/2)}\\
          &=\tilde{q}^{\frac{1}{24}}\prod_{j>0}\sum_{n_{i,k}=0,1}(\tilde{q}^{j+1/2})^{n_{i,k}}\\
        &=\tilde{q}^{-\frac{1}{48}}\prod_{j>0}(1+\tilde{q}^{j+1/2}).
    \end{split}
\end{align}
Under the action of parity $(-1)^{\sum_{k>0}n_{i,k}}$, the partition functions now become \begin{align}
    \begin{split}
    &\;\;\;\;\sum_{n_{i,k}=0,1}(-1)^{\sum_{k>0}n_{i,k}}e^{-\sum_{k>0}\beta \tilde{v}_ek(n_{i,k}+1/2)}\\
    &=
    \tilde{q}^{\frac{1}{24}}\prod_{j>0}
    \sum_{n_{i,k}=0,1}e^{i\pi n_{i,k}}
    e^{-\frac{2\pi i\beta\tilde{v}_e }{L}\left(j+\frac{1}{2}\right)n_{i,k}}\\
    &=\tilde{q}^{-\frac{1}{48}}\prod_{j>0}\left(1-\tilde{q}^{j+1/2}\right).
    \end{split}
\end{align}
Using \eqref{ProductTheta3}, the partition function can be recast in terms of modular function as \begin{align}
    \begin{split}
        &\;\;\;\;\tilde{q}^{-\frac{1}{48}}\prod_{j>0}\left(1+\tilde{q}^{j+1/2}\right)\\ &=\sqrt{\tilde{q}^{-\frac{1}{24}}\prod_{j>0}(1+\tilde{q}^{\frac{1}{2}(2j-1)})^2}\\
        &=\sqrt{\frac{
        \prod_{j>0}\left(1+\tilde{q}^{\frac{1}{2}(2j-1)}\right)^2\left(1-\tilde{q}^{\frac{1}{2}(2j)}\right)}{
        \tilde{q}^{\frac{1}{24}}\prod_{j>0}(1-\tilde{q}^j)}}\\
        &=\sqrt{\frac{\theta^0_0(\tilde{\tau})}{\eta(\tilde{\tau})}}.
    \end{split}
\end{align} Similarly using \eqref{ProductTheta4}, 
\begin{align}
    \begin{split}
        &\;\;\;\;\tilde{q}^{-\frac{1}{48}}\prod_{j>0}\left(1-\tilde{q}^{j+1/2}\right)\\
        &=\sqrt{\tilde{q}^{-\frac{1}{24}}\prod_{j>0}\left(1-\tilde{q}^{\frac{1}{2}\left(2j-1\right)}\right)^2}\\
        &=\sqrt{\frac{
        \prod_{j>0}\left(1-\tilde{q}^{\frac{1}{2}(2j-1)}\right)^2\left(1-\tilde{q}^{\frac{1}{2}(2j)}\right)}{
        \tilde{q}^{\frac{1}{24}}\prod_{j>0}(1-\tilde{q}^j)}}\\
        &=\sqrt{\frac{\theta^0_{1/2}(\tilde{\tau})}{\eta(\tilde{\tau})}}.
    \end{split}
\end{align} Thus the character functions in the untwisted sector have the forms
\begin{align}
    \begin{split}
        \chi^{\mathrm{Ising}}_0(\tilde{q}) &=\frac{1}{2}\sqrt{\frac{\theta^0_0(\tilde{\tau})}{\eta(\tilde{\tau})}} +\frac{1}{2}\sqrt{\frac{\theta^0_{1/2}(\tilde{\tau})}{\eta(\tilde{\tau})}},\\
        \chi^{\mathrm{Ising}}_{1/2}(\tilde{q}) &=\frac{1}{2}\sqrt{\frac{\theta^0_0(\tilde{\tau})}{\eta(\tilde{\tau})}} -\frac{1}{2}\sqrt{\frac{\theta^0_{1/2}(\tilde{\tau})}{\eta(\tilde{\tau})}}.
    \end{split}
\end{align}

On the other hand, for periodic boundary condition ($k=\frac{2\pi j}{L}, j\in \mathbb{Z})$, the fermionic partition functions is instead 
\begin{align}
    \begin{split}
        &\;\;\;\;\sum_{n_{i,k}=0,1}e^{-\sum_{k>0}\beta \tilde{v}_ek(n_{i,k}+1/2)}\\
        &=\tilde{q}^{\frac{1}{24}}\prod_{j>0}\sum_{n_{i,k}=0,1}(\tilde{q}^j)^{n_{i,k}}\\
        &=\tilde{q}^{\frac{1}{24}}\prod_{j>0}(1+\tilde{q}^j).\label{Ising16Partition}
    \end{split}
\end{align}  We now rewrite \eqref{Ising16Partition} in terms of modular function using \eqref{Theta2}:
\begin{align}
    \begin{split}
        &\;\;\;\;\tilde{q}^{\frac{1}{24}}\prod_{j>0}(1+\tilde{q}^j)\\
        &=\sqrt{\tilde{q}^{1/12}\prod_{j>0}(1+\tilde{q}^j)^2}\\
        &=\frac{2
        \tilde{q}^{\frac{1}{8}\prod_{j>0}\left(1+\tilde{q}^{\frac{1}{2}(2j)}\right)^2\big(1-\tilde{q}\big)^{\frac{1}{2}(2j)}}
        }{2\tilde{q}^{\frac{1}{24}}\prod_{j>0}(1-\tilde{q}^j)}\\
        &=\sqrt{\frac{\theta^{1/2}_0(\tilde{\tau})}{2\eta(\tilde{\tau})}}.
    \end{split}
\end{align} So the character function in the twisted sector has the form \begin{align}
  \chi^{\text{Ising}}_{1/16}(\tilde{q})=  \sqrt{\frac{\theta^{1/2}_0(\tilde{\tau})}{2\eta(\tilde{\tau})}}.
\end{align} The rest of character functions in this paper can be read off from \eqref{EtaFunctionForm} and \eqref{AlmostGeneralThetaFunc}.


\begin{thebibliography}{77}
\expandafter\ifx\csname natexlab\endcsname\relax\def\natexlab#1{#1}\fi
\expandafter\ifx\csname bibnamefont\endcsname\relax
  \def\bibnamefont#1{#1}\fi
\expandafter\ifx\csname bibfnamefont\endcsname\relax
  \def\bibfnamefont#1{#1}\fi
\expandafter\ifx\csname citenamefont\endcsname\relax
  \def\citenamefont#1{#1}\fi
\expandafter\ifx\csname url\endcsname\relax
  \def\url#1{\texttt{#1}}\fi
\expandafter\ifx\csname urlprefix\endcsname\relax\def\urlprefix{URL }\fi
\providecommand{\bibinfo}[2]{#2}
\providecommand{\eprint}[2][]{\url{#2}}


\bibitem[{\citenamefont{Wen and Niu}(1990)}]{wen-niu-1990}
\bibinfo{author}{\bibfnamefont{X.-G.}~\bibnamefont{Wen}} \bibnamefont{and}
  \bibinfo{author}{\bibfnamefont{Q.}~\bibnamefont{Niu}},
  \bibinfo{journal}{Phys. Rev. B} \textbf{\bibinfo{volume}{41}},
  \bibinfo{pages}{9377}
  (\bibinfo{year}{1990}).

\bibitem[{\citenamefont{{Oshikawa} et~al.}(2007)\citenamefont{{Oshikawa},
  {Kim}, {Shtengel}, {Nayak}, and {Tewari}}}]{2007AnPhy.322.1477O}
\bibinfo{author}{\bibfnamefont{M.}~\bibnamefont{{Oshikawa}}},
  \bibinfo{author}{\bibfnamefont{Y.~B.} \bibnamefont{{Kim}}},
  \bibinfo{author}{\bibfnamefont{K.}~\bibnamefont{{Shtengel}}},
  \bibinfo{author}{\bibfnamefont{C.}~\bibnamefont{{Nayak}}}, \bibnamefont{and}
  \bibinfo{author}{\bibfnamefont{S.}~\bibnamefont{{Tewari}}},
  \bibinfo{journal}{Annals of Physics} \textbf{\bibinfo{volume}{322}},
  \bibinfo{pages}{1477} (\bibinfo{year}{2007}), \eprint{cond-mat/0607743}.

\bibitem[{\citenamefont{{Wen}}(2012)}]{2012arXiv1210.1281W}
\bibinfo{author}{\bibfnamefont{X.-G.} \bibnamefont{{Wen}}},
  \bibinfo{journal}{arXiv e-prints} \bibinfo{eid}{arXiv:1210.1281}
  (\bibinfo{year}{2012}), \eprint{1210.1281}.

\bibitem[{\citenamefont{Chen et~al.}(2013)\citenamefont{Chen, Gu, Liu, and
  Wen}}]{PhysRevB.87.155114}
\bibinfo{author}{\bibfnamefont{X.}~\bibnamefont{Chen}},
  \bibinfo{author}{\bibfnamefont{Z.-C.} \bibnamefont{Gu}},
  \bibinfo{author}{\bibfnamefont{Z.-X.} \bibnamefont{Liu}}, \bibnamefont{and}
  \bibinfo{author}{\bibfnamefont{X.-G.} \bibnamefont{Wen}},
  \bibinfo{journal}{Phys. Rev. B} \textbf{\bibinfo{volume}{87}},
  \bibinfo{pages}{155114} (\bibinfo{year}{2013}),
  \urlprefix\url{https://link.aps.org/doi/10.1103/PhysRevB.87.155114}.

\bibitem[{\citenamefont{Lu and Vishwanath}(2012)}]{PhysRevB.86.125119}
\bibinfo{author}{\bibfnamefont{Y.-M.} \bibnamefont{Lu}} \bibnamefont{and}
  \bibinfo{author}{\bibfnamefont{A.}~\bibnamefont{Vishwanath}},
  \bibinfo{journal}{Phys. Rev. B} \textbf{\bibinfo{volume}{86}},
  \bibinfo{pages}{125119} (\bibinfo{year}{2012}),
  \urlprefix\url{https://link.aps.org/doi/10.1103/PhysRevB.86.125119}.

\bibitem[{\citenamefont{Wen}(1991)}]{wengaplessboundary}
\bibinfo{author}{\bibfnamefont{X.~G.} \bibnamefont{Wen}},
  \bibinfo{journal}{Phys. Rev B} \textbf{\bibinfo{volume}{43}},
  \bibinfo{pages}{11025} (\bibinfo{year}{1991}).

\bibitem[{\citenamefont{Hamma et~al.}(2005)\citenamefont{Hamma, Ionicioiu, and
  Zanardi}}]{hamma-2005-kitaev}
\bibinfo{author}{\bibfnamefont{A.}~\bibnamefont{Hamma}},
  \bibinfo{author}{\bibfnamefont{R.}~\bibnamefont{Ionicioiu}},
  \bibnamefont{and} \bibinfo{author}{\bibfnamefont{P.}~\bibnamefont{Zanardi}},
  \bibinfo{journal}{Physics Letters A} \textbf{\bibinfo{volume}{337}},
  \bibinfo{pages}{22 } (\bibinfo{year}{2005}), ISSN \bibinfo{issn}{0375-9601},
  \urlprefix\url{http://www.sciencedirect.com/science/article/pii/S0375960105001544}.

\bibitem[{\citenamefont{Levin and Wen}(2006)}]{levin-wen-2006}
\bibinfo{author}{\bibfnamefont{M.}~\bibnamefont{Levin}} \bibnamefont{and}
  \bibinfo{author}{\bibfnamefont{X.-G.} \bibnamefont{Wen}},
  \bibinfo{journal}{Phys. Rev. Lett.} \textbf{\bibinfo{volume}{96}},
  \bibinfo{pages}{110405} (\bibinfo{year}{2006}),
  \urlprefix\url{https://link.aps.org/doi/10.1103/PhysRevLett.96.110405}.

\bibitem[{\citenamefont{Kitaev and Preskill}(2006)}]{kitaev-preskill-2006}
\bibinfo{author}{\bibfnamefont{A.}~\bibnamefont{Kitaev}} \bibnamefont{and}
  \bibinfo{author}{\bibfnamefont{J.}~\bibnamefont{Preskill}},
  \bibinfo{journal}{Phys. Rev. Lett.} \textbf{\bibinfo{volume}{96}},
  \bibinfo{pages}{110404} (\bibinfo{year}{2006}),
  \urlprefix\url{https://link.aps.org/doi/10.1103/PhysRevLett.96.110404}.

\bibitem[{\citenamefont{Kitaev}(2003)}]{kitaev-z2}
\bibinfo{author}{\bibfnamefont{A.~Y.} \bibnamefont{Kitaev}},
  \bibinfo{journal}{Annals of Physics} \textbf{\bibinfo{volume}{303}},
  \bibinfo{pages}{2} (\bibinfo{year}{2003}).

\bibitem[{\citenamefont{Moore and Read}(1991)}]{Moore1991}
\bibinfo{author}{\bibfnamefont{G.}~\bibnamefont{Moore}} \bibnamefont{and}
  \bibinfo{author}{\bibfnamefont{N.}~\bibnamefont{Read}},
  \bibinfo{journal}{Nuclear Physics B} \textbf{\bibinfo{volume}{360}},
  \bibinfo{pages}{362} (\bibinfo{year}{1991}), ISSN \bibinfo{issn}{0550-3213},
  \urlprefix\url{http://www.sciencedirect.com/science/article/pii/055032139190407O}.

\bibitem[{\citenamefont{Cano et~al.}(2015)\citenamefont{Cano, Hughes, and
  Mulligan}}]{cano-2015}
\bibinfo{author}{\bibfnamefont{J.}~\bibnamefont{Cano}},
  \bibinfo{author}{\bibfnamefont{T.~L.} \bibnamefont{Hughes}},
  \bibnamefont{and} \bibinfo{author}{\bibfnamefont{M.}~\bibnamefont{Mulligan}},
  \bibinfo{journal}{Phys. Rev. B} \textbf{\bibinfo{volume}{92}},
  \bibinfo{pages}{075104} (\bibinfo{year}{2015}),
  \urlprefix\url{https://link.aps.org/doi/10.1103/PhysRevB.92.075104}.

\bibitem[{\citenamefont{Santos et~al.}(2018)\citenamefont{Santos, Cano,
  Mulligan, and Hughes}}]{santos-2018}
\bibinfo{author}{\bibfnamefont{L.~H.} \bibnamefont{Santos}},
  \bibinfo{author}{\bibfnamefont{J.}~\bibnamefont{Cano}},
  \bibinfo{author}{\bibfnamefont{M.}~\bibnamefont{Mulligan}}, \bibnamefont{and}
  \bibinfo{author}{\bibfnamefont{T.~L.} \bibnamefont{Hughes}},
  \bibinfo{journal}{Phys. Rev. B} \textbf{\bibinfo{volume}{98}},
  \bibinfo{pages}{075131} (\bibinfo{year}{2018}),
  \urlprefix\url{https://link.aps.org/doi/10.1103/PhysRevB.98.075131}.

\bibitem[{\citenamefont{{Ohmori} and {Tachikawa}}(2015)}]{2015JSMTE..04..010O}
\bibinfo{author}{\bibfnamefont{K.}~\bibnamefont{{Ohmori}}} \bibnamefont{and}
  \bibinfo{author}{\bibfnamefont{Y.}~\bibnamefont{{Tachikawa}}},
  \bibinfo{journal}{Journal of Statistical Mechanics: Theory and Experiment}
  \textbf{\bibinfo{volume}{2015}}, \bibinfo{eid}{04010} (\bibinfo{year}{2015}),
  \eprint{1406.4167}.

\bibitem[{\citenamefont{Chandran et~al.}(2014)\citenamefont{Chandran, Khemani,
  and Sondhi}}]{ChandranKhemaniSondhi}
\bibinfo{author}{\bibfnamefont{A.}~\bibnamefont{Chandran}},
  \bibinfo{author}{\bibfnamefont{V.}~\bibnamefont{Khemani}}, \bibnamefont{and}
  \bibinfo{author}{\bibfnamefont{S.L.}~\bibnamefont{Sondhi}},
  \bibinfo{journal}{Phys. Rev. Lett.} \textbf{\bibinfo{volume}{113}},
  \bibinfo{pages}{060501} (\bibinfo{year}{2014}).

\bibitem[{\citenamefont{{Dong} et~al.}(2008)\citenamefont{{Dong}, {Fradkin},
  {Leigh}, and {Nowling}}}]{2008JHEP...05..016D}
\bibinfo{author}{\bibfnamefont{S.}~\bibnamefont{{Dong}}},
  \bibinfo{author}{\bibfnamefont{E.}~\bibnamefont{{Fradkin}}},
  \bibinfo{author}{\bibfnamefont{R.~G.} \bibnamefont{{Leigh}}},
  \bibnamefont{and}
  \bibinfo{author}{\bibfnamefont{S.}~\bibnamefont{{Nowling}}},
  \bibinfo{journal}{Journal of High Energy Physics}
  \textbf{\bibinfo{volume}{2008}}, \bibinfo{eid}{016} (\bibinfo{year}{2008}),
  \eprint{0802.3231}.

\bibitem[{\citenamefont{{Zhang} et~al.}(2012)\citenamefont{{Zhang}, {Grover},
  {Turner}, {Oshikawa}, and {Vishwanath}}}]{2012PhRvB..85w5151Z}
\bibinfo{author}{\bibfnamefont{Y.}~\bibnamefont{{Zhang}}},
  \bibinfo{author}{\bibfnamefont{T.}~\bibnamefont{{Grover}}},
  \bibinfo{author}{\bibfnamefont{A.}~\bibnamefont{{Turner}}},
  \bibinfo{author}{\bibfnamefont{M.}~\bibnamefont{{Oshikawa}}},
  \bibnamefont{and}
  \bibinfo{author}{\bibfnamefont{A.}~\bibnamefont{{Vishwanath}}},
  \bibinfo{journal}{\prb} \textbf{\bibinfo{volume}{85}}, \bibinfo{eid}{235151}
  (\bibinfo{year}{2012}), \eprint{1111.2342}.

\bibitem[{\citenamefont{Lee and Vidal}(2013)}]{lee-vidal-2013}
\bibinfo{author}{\bibfnamefont{Y.~A.} \bibnamefont{Lee}} \bibnamefont{and}
  \bibinfo{author}{\bibfnamefont{G.}~\bibnamefont{Vidal}},
  \bibinfo{journal}{Phys. Rev. A} \textbf{\bibinfo{volume}{88}},
  \bibinfo{pages}{042318} (\bibinfo{year}{2013}),
  \urlprefix\url{https://link.aps.org/doi/10.1103/PhysRevA.88.042318}.

\bibitem[{\citenamefont{Castelnovo}(2013)}]{castel-2013}
\bibinfo{author}{\bibfnamefont{C.}~\bibnamefont{Castelnovo}},
  \bibinfo{journal}{Phys. Rev. A} \textbf{\bibinfo{volume}{88}},
  \bibinfo{pages}{042319} (\bibinfo{year}{2013}),
  \urlprefix\url{https://link.aps.org/doi/10.1103/PhysRevA.88.042319}.

\bibitem[{\citenamefont{Wen et~al.}(2016)\citenamefont{Wen, Matsuura, and
  Ryu}}]{wen-matsuura-ryu-2016}
\bibinfo{author}{\bibfnamefont{X.}~\bibnamefont{Wen}},
  \bibinfo{author}{\bibfnamefont{S.}~\bibnamefont{Matsuura}}, \bibnamefont{and}
  \bibinfo{author}{\bibfnamefont{S.}~\bibnamefont{Ryu}},
  \bibinfo{journal}{Phys. Rev. B} \textbf{\bibinfo{volume}{93}},
  \bibinfo{pages}{245140} (\bibinfo{year}{2016}),
  \urlprefix\url{https://link.aps.org/doi/10.1103/PhysRevB.93.245140}.

\bibitem[{\citenamefont{Vidal and Werner}(2002)}]{vidal-werner-2002}
\bibinfo{author}{\bibfnamefont{G.}~\bibnamefont{Vidal}} \bibnamefont{and}
  \bibinfo{author}{\bibfnamefont{R.~F.} \bibnamefont{Werner}},
  \bibinfo{journal}{Phys. Rev. A} \textbf{\bibinfo{volume}{65}},
  \bibinfo{pages}{032314} (\bibinfo{year}{2002}),
  \urlprefix\url{https://link.aps.org/doi/10.1103/PhysRevA.65.032314}.

\bibitem[{\citenamefont{{Plenio} and {Virmani}}(2005)}]{2005quant.ph..4163P}
\bibinfo{author}{\bibfnamefont{M.~B.} \bibnamefont{{Plenio}}} \bibnamefont{and}
  \bibinfo{author}{\bibfnamefont{S.}~\bibnamefont{{Virmani}}},
  \bibinfo{journal}{arXiv e-prints} \bibinfo{eid}{quant-ph/0504163}
  (\bibinfo{year}{2005}), \eprint{quant-ph/0504163}.

\bibitem[{\citenamefont{Groisman et~al.}(2005)\citenamefont{Groisman, Popescu,
  and Winter}}]{PhysRevA.72.032317}
\bibinfo{author}{\bibfnamefont{B.}~\bibnamefont{Groisman}},
  \bibinfo{author}{\bibfnamefont{S.}~\bibnamefont{Popescu}}, \bibnamefont{and}
  \bibinfo{author}{\bibfnamefont{A.}~\bibnamefont{Winter}},
  \bibinfo{journal}{Phys. Rev. A} \textbf{\bibinfo{volume}{72}},
  \bibinfo{pages}{032317} (\bibinfo{year}{2005}),
  \urlprefix\url{https://link.aps.org/doi/10.1103/PhysRevA.72.032317}.

\bibitem[{\citenamefont{Peres}(1996)}]{peres-1996}
\bibinfo{author}{\bibfnamefont{A.}~\bibnamefont{Peres}},
  \bibinfo{journal}{Phys. Rev. Lett.} \textbf{\bibinfo{volume}{77}},
  \bibinfo{pages}{1413} (\bibinfo{year}{1996}),
  \urlprefix\url{https://link.aps.org/doi/10.1103/PhysRevLett.77.1413}.

\bibitem[{\citenamefont{D\"ur et~al.}(2000)\citenamefont{D\"ur, Vidal, and
  Cirac}}]{dur-vidal-2000}
\bibinfo{author}{\bibfnamefont{W.}~\bibnamefont{D\"ur}},
  \bibinfo{author}{\bibfnamefont{G.}~\bibnamefont{Vidal}}, \bibnamefont{and}
  \bibinfo{author}{\bibfnamefont{J.~I.} \bibnamefont{Cirac}},
  \bibinfo{journal}{Phys. Rev. A} \textbf{\bibinfo{volume}{62}},
  \bibinfo{pages}{062314} (\bibinfo{year}{2000}),
  \urlprefix\url{https://link.aps.org/doi/10.1103/PhysRevA.62.062314}.

\bibitem[{\citenamefont{Calabrese et~al.}(2012)\citenamefont{Calabrese, Cardy,
  and Tonni}}]{calabrese-2012}
\bibinfo{author}{\bibfnamefont{P.}~\bibnamefont{Calabrese}},
  \bibinfo{author}{\bibfnamefont{J.}~\bibnamefont{Cardy}}, \bibnamefont{and}
  \bibinfo{author}{\bibfnamefont{E.}~\bibnamefont{Tonni}},
  \bibinfo{journal}{Phys. Rev. Lett.} \textbf{\bibinfo{volume}{109}},
  \bibinfo{pages}{130502} (\bibinfo{year}{2012}),
  \urlprefix\url{https://link.aps.org/doi/10.1103/PhysRevLett.109.130502}.

\bibitem[{\citenamefont{{Rangamani} and {Rota}}(2014)}]{2014JHEP...10..060R}
\bibinfo{author}{\bibfnamefont{M.}~\bibnamefont{{Rangamani}}} \bibnamefont{and}
  \bibinfo{author}{\bibfnamefont{M.}~\bibnamefont{{Rota}}},
  \bibinfo{journal}{Journal of High Energy Physics}
  \textbf{\bibinfo{volume}{2014}}, \bibinfo{eid}{60} (\bibinfo{year}{2014}),
  \eprint{1406.6989}.

\bibitem[{\citenamefont{{Dong} et~al.}(2021)\citenamefont{{Dong}, {Qi}, and
  {Walter}}}]{2021arXiv210111029D}
\bibinfo{author}{\bibfnamefont{X.}~\bibnamefont{{Dong}}},
  \bibinfo{author}{\bibfnamefont{X.-L.} \bibnamefont{{Qi}}}, \bibnamefont{and}
  \bibinfo{author}{\bibfnamefont{M.}~\bibnamefont{{Walter}}},
  \bibinfo{journal}{arXiv e-prints} \bibinfo{eid}{arXiv:2101.11029}
  (\bibinfo{year}{2021}), \eprint{2101.11029}.

\bibitem[{\citenamefont{{Calabrese} et~al.}(2015)\citenamefont{{Calabrese},
  {Cardy}, and {Tonni}}}]{2015JPhA...48a5006C}
\bibinfo{author}{\bibfnamefont{P.}~\bibnamefont{{Calabrese}}},
  \bibinfo{author}{\bibfnamefont{J.}~\bibnamefont{{Cardy}}}, \bibnamefont{and}
  \bibinfo{author}{\bibfnamefont{E.}~\bibnamefont{{Tonni}}},
  \bibinfo{journal}{Journal of Physics A Mathematical General}
  \textbf{\bibinfo{volume}{48}}, \bibinfo{eid}{015006} (\bibinfo{year}{2015}),
  \eprint{1408.3043}.

\bibitem[{\citenamefont{{Shapourian} and {Ryu}}(2019)}]{2019JSMTE..04.3106S}
\bibinfo{author}{\bibfnamefont{H.}~\bibnamefont{{Shapourian}}}
  \bibnamefont{and} \bibinfo{author}{\bibfnamefont{S.}~\bibnamefont{{Ryu}}},
  \bibinfo{journal}{Journal of Statistical Mechanics: Theory and Experiment}
  \textbf{\bibinfo{volume}{4}}, \bibinfo{pages}{043106} (\bibinfo{year}{2019}),
  \eprint{1807.09808}.

\bibitem[{\citenamefont{Lu and Grover}(2020)}]{PhysRevResearch.2.043345}
\bibinfo{author}{\bibfnamefont{T.-C.} \bibnamefont{Lu}} \bibnamefont{and}
  \bibinfo{author}{\bibfnamefont{T.}~\bibnamefont{Grover}},
  \bibinfo{journal}{Phys. Rev. Research} \textbf{\bibinfo{volume}{2}},
  \bibinfo{pages}{043345} (\bibinfo{year}{2020}),
  \urlprefix\url{https://link.aps.org/doi/10.1103/PhysRevResearch.2.043345}.

\bibitem[{\citenamefont{Cornfeld et~al.}(2018)\citenamefont{Cornfeld,
  Goldstein, and Sela}}]{PhysRevA.98.032302}
\bibinfo{author}{\bibfnamefont{E.}~\bibnamefont{Cornfeld}},
  \bibinfo{author}{\bibfnamefont{M.}~\bibnamefont{Goldstein}},
  \bibnamefont{and} \bibinfo{author}{\bibfnamefont{E.}~\bibnamefont{Sela}},
  \bibinfo{journal}{Phys. Rev. A} \textbf{\bibinfo{volume}{98}},
  \bibinfo{pages}{032302} (\bibinfo{year}{2018}),
  \urlprefix\url{https://link.aps.org/doi/10.1103/PhysRevA.98.032302}.

\bibitem[{\citenamefont{Hart and Castelnovo}(2018)}]{PhysRevB.97.144410}
\bibinfo{author}{\bibfnamefont{O.}~\bibnamefont{Hart}} \bibnamefont{and}
  \bibinfo{author}{\bibfnamefont{C.}~\bibnamefont{Castelnovo}},
  \bibinfo{journal}{Phys. Rev. B} \textbf{\bibinfo{volume}{97}},
  \bibinfo{pages}{144410} (\bibinfo{year}{2018}),
  \urlprefix\url{https://link.aps.org/doi/10.1103/PhysRevB.97.144410}.

\bibitem[{\citenamefont{Lu et~al.}(2020)\citenamefont{Lu, Hsieh, and
  Grover}}]{PhysRevLett.125.116801}
\bibinfo{author}{\bibfnamefont{T.-C.} \bibnamefont{Lu}},
  \bibinfo{author}{\bibfnamefont{T.~H.} \bibnamefont{Hsieh}}, \bibnamefont{and}
  \bibinfo{author}{\bibfnamefont{T.}~\bibnamefont{Grover}},
  \bibinfo{journal}{Phys. Rev. Lett.} \textbf{\bibinfo{volume}{125}},
  \bibinfo{pages}{116801} (\bibinfo{year}{2020}),
  \urlprefix\url{https://link.aps.org/doi/10.1103/PhysRevLett.125.116801}.

\bibitem[{\citenamefont{Coser et~al.}(2014)\citenamefont{Coser, Tonni, and
  Calabrese}}]{Coser_2014}
\bibinfo{author}{\bibfnamefont{A.}~\bibnamefont{Coser}},
  \bibinfo{author}{\bibfnamefont{E.}~\bibnamefont{Tonni}}, \bibnamefont{and}
  \bibinfo{author}{\bibfnamefont{P.}~\bibnamefont{Calabrese}},
  \bibinfo{journal}{Journal of Statistical Mechanics: Theory and Experiment}
  \textbf{\bibinfo{volume}{2014}}, \bibinfo{pages}{P12017}
  (\bibinfo{year}{2014}),
  \urlprefix\url{https://doi.org/10.1088/1742-5468/2014/12/p12017}.

\bibitem[{\citenamefont{Eisler and Zimbor{\'{a}}s}(2014)}]{Eisler_2014}
\bibinfo{author}{\bibfnamefont{V.}~\bibnamefont{Eisler}} \bibnamefont{and}
  \bibinfo{author}{\bibfnamefont{Z.}~\bibnamefont{Zimbor{\'{a}}s}},
  \bibinfo{journal}{New Journal of Physics} \textbf{\bibinfo{volume}{16}},
  \bibinfo{pages}{123020} (\bibinfo{year}{2014}),
  \urlprefix\url{https://doi.org/10.1088/1367-2630/16/12/123020}.

\bibitem[{\citenamefont{Hoogeveen and Doyon}(2015)}]{HOOGEVEEN201578}
\bibinfo{author}{\bibfnamefont{M.}~\bibnamefont{Hoogeveen}} \bibnamefont{and}
  \bibinfo{author}{\bibfnamefont{B.}~\bibnamefont{Doyon}},
  \bibinfo{journal}{Nuclear Physics B} \textbf{\bibinfo{volume}{898}},
  \bibinfo{pages}{78 } (\bibinfo{year}{2015}), ISSN \bibinfo{issn}{0550-3213},
  \urlprefix\url{http://www.sciencedirect.com/science/article/pii/S0550321315002242}.

\bibitem[{\citenamefont{Wen et~al.}(2015)\citenamefont{Wen, Chang, and
  Ryu}}]{PhysRevB.92.075109}
\bibinfo{author}{\bibfnamefont{X.}~\bibnamefont{Wen}},
  \bibinfo{author}{\bibfnamefont{P.-Y.} \bibnamefont{Chang}}, \bibnamefont{and}
  \bibinfo{author}{\bibfnamefont{S.}~\bibnamefont{Ryu}},
  \bibinfo{journal}{Phys. Rev. B} \textbf{\bibinfo{volume}{92}},
  \bibinfo{pages}{075109} (\bibinfo{year}{2015}),
  \urlprefix\url{https://link.aps.org/doi/10.1103/PhysRevB.92.075109}.

\bibitem[{\citenamefont{{Shapourian}
  et~al.}(2020{\natexlab{a}})\citenamefont{{Shapourian}, {Liu}, {Kudler-Flam},
  and {Vishwanath}}}]{2020arXiv201101277S}
\bibinfo{author}{\bibfnamefont{H.}~\bibnamefont{{Shapourian}}},
  \bibinfo{author}{\bibfnamefont{S.}~\bibnamefont{{Liu}}},
  \bibinfo{author}{\bibfnamefont{J.}~\bibnamefont{{Kudler-Flam}}},
  \bibnamefont{and}
  \bibinfo{author}{\bibfnamefont{A.}~\bibnamefont{{Vishwanath}}},
  \bibinfo{journal}{arXiv e-prints} \bibinfo{eid}{arXiv:2011.01277}
  (\bibinfo{year}{2020}{\natexlab{a}}), \eprint{2011.01277}.

\bibitem[{\citenamefont{{Sang} et~al.}(2020)\citenamefont{{Sang}, {Li}, {Zhou},
  {Chen}, {Hsieh}, and {Fisher}}}]{2020arXiv201200031S}
\bibinfo{author}{\bibfnamefont{S.}~\bibnamefont{{Sang}}},
  \bibinfo{author}{\bibfnamefont{Y.}~\bibnamefont{{Li}}},
  \bibinfo{author}{\bibfnamefont{T.}~\bibnamefont{{Zhou}}},
  \bibinfo{author}{\bibfnamefont{X.}~\bibnamefont{{Chen}}},
  \bibinfo{author}{\bibfnamefont{T.~H.} \bibnamefont{{Hsieh}}},
  \bibnamefont{and} \bibinfo{author}{\bibfnamefont{M.~P.~A.}
  \bibnamefont{{Fisher}}}, \bibinfo{journal}{arXiv e-prints}
  \bibinfo{eid}{arXiv:2012.00031} (\bibinfo{year}{2020}), \eprint{2012.00031}.

\bibitem[{\citenamefont{{Shi} et~al.}(2020)\citenamefont{{Shi}, {Dai}, and
  {Lu}}}]{2020arXiv201200040S}
\bibinfo{author}{\bibfnamefont{B.}~\bibnamefont{{Shi}}},
  \bibinfo{author}{\bibfnamefont{X.}~\bibnamefont{{Dai}}}, \bibnamefont{and}
  \bibinfo{author}{\bibfnamefont{Y.-M.} \bibnamefont{{Lu}}},
  \bibinfo{journal}{arXiv e-prints} \bibinfo{eid}{arXiv:2012.00040}
  (\bibinfo{year}{2020}), \eprint{2012.00040}.

\bibitem[{\citenamefont{He and Vidal}(2015)}]{he-vidal-2015}
\bibinfo{author}{\bibfnamefont{H.}~\bibnamefont{He}} \bibnamefont{and}
  \bibinfo{author}{\bibfnamefont{G.}~\bibnamefont{Vidal}},
  \bibinfo{journal}{Phys. Rev. A} \textbf{\bibinfo{volume}{91}},
  \bibinfo{pages}{012339} (\bibinfo{year}{2015}),
  \urlprefix\url{https://link.aps.org/doi/10.1103/PhysRevA.91.012339}.

\bibitem[{\citenamefont{Coffman et~al.}(2000)\citenamefont{Coffman, Kundu, and
  Wootters}}]{coffman-2000}
\bibinfo{author}{\bibfnamefont{V.}~\bibnamefont{Coffman}},
  \bibinfo{author}{\bibfnamefont{J.}~\bibnamefont{Kundu}}, \bibnamefont{and}
  \bibinfo{author}{\bibfnamefont{W.~K.} \bibnamefont{Wootters}},
  \bibinfo{journal}{Phys. Rev. A} \textbf{\bibinfo{volume}{61}},
  \bibinfo{pages}{052306} (\bibinfo{year}{2000}),
  \urlprefix\url{https://link.aps.org/doi/10.1103/PhysRevA.61.052306}.

\bibitem[{\citenamefont{Osborne and Verstraete}(2006)}]{osborne-2006}
\bibinfo{author}{\bibfnamefont{T.~J.} \bibnamefont{Osborne}} \bibnamefont{and}
  \bibinfo{author}{\bibfnamefont{F.}~\bibnamefont{Verstraete}},
  \bibinfo{journal}{Phys. Rev. Lett.} \textbf{\bibinfo{volume}{96}},
  \bibinfo{pages}{220503} (\bibinfo{year}{2006}),
  \urlprefix\url{https://link.aps.org/doi/10.1103/PhysRevLett.96.220503}.

\bibitem[{\citenamefont{Ou and Fan}(2007)}]{ou-fan-2007}
\bibinfo{author}{\bibfnamefont{Y.-C.} \bibnamefont{Ou}} \bibnamefont{and}
  \bibinfo{author}{\bibfnamefont{H.}~\bibnamefont{Fan}},
  \bibinfo{journal}{Phys. Rev. A} \textbf{\bibinfo{volume}{75}},
  \bibinfo{pages}{062308} (\bibinfo{year}{2007}),
  \urlprefix\url{https://link.aps.org/doi/10.1103/PhysRevA.75.062308}.

\bibitem[{\citenamefont{Gour and Guo}(2018)}]{gour-guo-2018}
\bibinfo{author}{\bibfnamefont{G.}~\bibnamefont{Gour}} \bibnamefont{and}
  \bibinfo{author}{\bibfnamefont{Y.}~\bibnamefont{Guo}},
  \bibinfo{journal}{{Quantum}} \textbf{\bibinfo{volume}{2}},
  \bibinfo{pages}{81} (\bibinfo{year}{2018}), ISSN \bibinfo{issn}{2521-327X},
  \urlprefix\url{https://doi.org/10.22331/q-2018-08-13-81}.

\bibitem[{\citenamefont{Hayden et~al.}(2004)\citenamefont{Hayden, Jozsa, Petz,
  and Winter}}]{hayden-2004}
\bibinfo{author}{\bibfnamefont{P.}~\bibnamefont{Hayden}},
  \bibinfo{author}{\bibfnamefont{R.}~\bibnamefont{Jozsa}},
  \bibinfo{author}{\bibfnamefont{D.}~\bibnamefont{Petz}}, \bibnamefont{and}
  \bibinfo{author}{\bibfnamefont{A.}~\bibnamefont{Winter}},
  \bibinfo{journal}{Communications in Mathematical Physics}
  \textbf{\bibinfo{volume}{246}}, \bibinfo{pages}{359} (\bibinfo{year}{2004}),
  ISSN \bibinfo{issn}{1432-0916},
  \urlprefix\url{http://dx.doi.org/10.1007/s00220-004-1049-z}.

\bibitem[{\citenamefont{Elitzur et~al.}(1989)\citenamefont{Elitzur, Moore,
  Schwimmer, and Seiberg}}]{Elitzur:1989nr}
\bibinfo{author}{\bibfnamefont{S.}~\bibnamefont{Elitzur}},
  \bibinfo{author}{\bibfnamefont{G.~W.} \bibnamefont{Moore}},
  \bibinfo{author}{\bibfnamefont{A.}~\bibnamefont{Schwimmer}},
  \bibnamefont{and} \bibinfo{author}{\bibfnamefont{N.}~\bibnamefont{Seiberg}},
  \bibinfo{journal}{Nucl. Phys. B} \textbf{\bibinfo{volume}{326}},
  \bibinfo{pages}{108} (\bibinfo{year}{1989}).

\bibitem[{\citenamefont{WEN}(1992)}]{doi:10.1142/S0217979292000840}
\bibinfo{author}{\bibfnamefont{X.-G.} \bibnamefont{WEN}},
  \bibinfo{journal}{International Journal of Modern Physics B}
  \textbf{\bibinfo{volume}{06}}, \bibinfo{pages}{1711} (\bibinfo{year}{1992}),
  \eprint{https://doi.org/10.1142/S0217979292000840},
  \urlprefix\url{https://doi.org/10.1142/S0217979292000840}.


\bibitem[{\citenamefont{Qi et~al.}(2012)\citenamefont{Qi, Katsura, and Ludwig}}]{Qi-Katsura-Ludwig-2012}
\bibinfo{author}{\bibfnamefont{X.-L.} \bibnamefont{Qi}},
  \bibinfo{author}{\bibfnamefont{H.}~\bibnamefont{Katsura}}, \bibnamefont{and}
  \bibinfo{author}{\bibfnamefont{A.~W.~W.} \bibnamefont{Ludwig}},
  \bibinfo{journal}{Phys. Rev. Lett.} \textbf{\bibinfo{volume}{108}},
  \bibinfo{pages}{196402}
  (\bibinfo{year}{2012}).

\bibitem[{\citenamefont{Lundgren et~al.}(2013)\citenamefont{Lundgren, Fuji,
  Furukawa, and Oshikawa}}]{lundgrenentanglement}
\bibinfo{author}{\bibfnamefont{R.}~\bibnamefont{Lundgren}},
  \bibinfo{author}{\bibfnamefont{Y.}~\bibnamefont{Fuji}},
  \bibinfo{author}{\bibfnamefont{S.}~\bibnamefont{Furukawa}}, \bibnamefont{and}
  \bibinfo{author}{\bibfnamefont{M.}~\bibnamefont{Oshikawa}},
  \bibinfo{journal}{Phys. Rev. B} \textbf{\bibinfo{volume}{88}},
  \bibinfo{pages}{245137} (\bibinfo{year}{2013}).

\bibitem[{\citenamefont{Teo and Kane}(2014)}]{Teo-Kane-2014}
\bibinfo{author}{\bibfnamefont{J.~C.~Y.} \bibnamefont{Teo}} \bibnamefont{and}
  \bibinfo{author}{\bibfnamefont{C.L.}~\bibnamefont{Kane}},
  \bibinfo{journal}{PHYSICAL REVIEW B} \textbf{\bibinfo{volume}{89}},\bibinfo{pages}{085101}
  (\bibinfo{year}{2014}).

\bibitem[{\citenamefont{Li and Haldane}(2008)}]{lihaldane}
\bibinfo{author}{\bibfnamefont{H.}~\bibnamefont{Li}} \bibnamefont{and}
  \bibinfo{author}{\bibfnamefont{F.~D.~M.} \bibnamefont{Haldane}},
  \bibinfo{journal}{Phys. Rev. Lett.} \textbf{\bibinfo{volume}{101}},
  \bibinfo{pages}{010504} (\bibinfo{year}{2008}).

\bibitem[{\citenamefont{Regnault et~al.}(2009)\citenamefont{Regnault, Bernevig,
  and Haldane}}]{Regnault2009}
\bibinfo{author}{\bibfnamefont{N.}~\bibnamefont{Regnault}},
  \bibinfo{author}{\bibfnamefont{B.~A.} \bibnamefont{Bernevig}},
  \bibnamefont{and} \bibinfo{author}{\bibfnamefont{F.~D.~M.}
  \bibnamefont{Haldane}}, \bibinfo{journal}{Phys. Rev. Lett.}
  \textbf{\bibinfo{volume}{103}}, \bibinfo{pages}{016801}
  (\bibinfo{year}{2009}).

\bibitem[{\citenamefont{Thomale et~al.}(2010)\citenamefont{Thomale, Sterdyniak,
  Regnault, and Bernevig}}]{Thomale2010}
\bibinfo{author}{\bibfnamefont{R.}~\bibnamefont{Thomale}},
  \bibinfo{author}{\bibfnamefont{A.}~\bibnamefont{Sterdyniak}},
  \bibinfo{author}{\bibfnamefont{N.}~\bibnamefont{Regnault}}, \bibnamefont{and}
  \bibinfo{author}{\bibfnamefont{B.~A.} \bibnamefont{Bernevig}},
  \bibinfo{journal}{Phys. Rev. Lett.} \textbf{\bibinfo{volume}{104}},
  \bibinfo{pages}{180502} (\bibinfo{year}{2010}).

\bibitem[{\citenamefont{L{\"a}uchli et~al.}(2010)\citenamefont{L{\"a}uchli,
  Bergholtz, Suorsa, and Haque}}]{lauchli2010}
\bibinfo{author}{\bibfnamefont{A.~M.} \bibnamefont{L{\"a}uchli}},
  \bibinfo{author}{\bibfnamefont{E.~J.} \bibnamefont{Bergholtz}},
  \bibinfo{author}{\bibfnamefont{J.}~\bibnamefont{Suorsa}}, \bibnamefont{and}
  \bibinfo{author}{\bibfnamefont{M.}~\bibnamefont{Haque}},
  \bibinfo{journal}{Physical review letters} \textbf{\bibinfo{volume}{104}},
  \bibinfo{pages}{156404} (\bibinfo{year}{2010}).

\bibitem[{\citenamefont{Papi{\'c} et~al.}(2011)\citenamefont{Papi{\'c},
  Bernevig, and Regnault}}]{papic2011}
\bibinfo{author}{\bibfnamefont{Z.}~\bibnamefont{Papi{\'c}}},
  \bibinfo{author}{\bibfnamefont{B.A.}~\bibnamefont{Bernevig}}, \bibnamefont{and}
  \bibinfo{author}{\bibfnamefont{N.}~\bibnamefont{Regnault}},
  \bibinfo{journal}{Physical Review Letters} \textbf{\bibinfo{volume}{106}},
  \bibinfo{pages}{056801} (\bibinfo{year}{2011}).

\bibitem[{\citenamefont{Chandran et~al.}(2011)\citenamefont{Chandran, Hermanns,
  Regnault, and Bernevig}}]{chandran2011}
\bibinfo{author}{\bibfnamefont{A.}~\bibnamefont{Chandran}},
  \bibinfo{author}{\bibfnamefont{M.}~\bibnamefont{Hermanns}},
  \bibinfo{author}{\bibfnamefont{N.}~\bibnamefont{Regnault}}, \bibnamefont{and}
  \bibinfo{author}{\bibfnamefont{B.~A.} \bibnamefont{Bernevig}},
  \bibinfo{journal}{Physical Review B} \textbf{\bibinfo{volume}{84}},
  \bibinfo{pages}{205136} (\bibinfo{year}{2011}).

\bibitem[{\citenamefont{Hermanns et~al.}(2011)\citenamefont{Hermanns, Chandran,
  Regnault, and Bernevig}}]{hermanns2011}
\bibinfo{author}{\bibfnamefont{M.}~\bibnamefont{Hermanns}},
  \bibinfo{author}{\bibfnamefont{A.}~\bibnamefont{Chandran}},
  \bibinfo{author}{\bibfnamefont{N.}~\bibnamefont{Regnault}}, \bibnamefont{and}
  \bibinfo{author}{\bibfnamefont{B.~A.} \bibnamefont{Bernevig}},
  \bibinfo{journal}{Physical Review B} \textbf{\bibinfo{volume}{84}},
  \bibinfo{pages}{121309(R)} (\bibinfo{year}{2011}).

\bibitem[{\citenamefont{Rodr{\'\i}guez
  et~al.}(2013)\citenamefont{Rodr{\'\i}guez, Davenport, Simon, and
  Slingerland}}]{Simon2013}
\bibinfo{author}{\bibfnamefont{I.~D.} \bibnamefont{Rodr{\'\i}guez}},
  \bibinfo{author}{\bibfnamefont{S.~C.} \bibnamefont{Davenport}},
  \bibinfo{author}{\bibfnamefont{S.~H.} \bibnamefont{Simon}}, \bibnamefont{and}
  \bibinfo{author}{\bibfnamefont{J.~K.} \bibnamefont{Slingerland}},
  \bibinfo{journal}{Phys. Rev. B} \textbf{\bibinfo{volume}{88}},
  \bibinfo{pages}{155307} (\bibinfo{year}{2013}).

\bibitem[{\citenamefont{Pollmann et~al.}(2010)\citenamefont{Pollmann, Turner,
  Berg, and Oshikawa}}]{Pollman2010}
\bibinfo{author}{\bibfnamefont{F.}~\bibnamefont{Pollmann}},
  \bibinfo{author}{\bibfnamefont{A.~M.} \bibnamefont{Turner}},
  \bibinfo{author}{\bibfnamefont{E.}~\bibnamefont{Berg}}, \bibnamefont{and}
  \bibinfo{author}{\bibfnamefont{M.}~\bibnamefont{Oshikawa}},
  \bibinfo{journal}{Phys. Rev. B} \textbf{\bibinfo{volume}{81}},
  \bibinfo{pages}{064439} (\bibinfo{year}{2010}).

\bibitem[{\citenamefont{Fidkowski}(2010)}]{Fidkowski2010}
\bibinfo{author}{\bibfnamefont{L.}~\bibnamefont{Fidkowski}},
  \bibinfo{journal}{Phys. Rev. Lett.} \textbf{\bibinfo{volume}{104}},
  \bibinfo{pages}{130502} (\bibinfo{year}{2010}).

\bibitem[{\citenamefont{Prodan et~al.}(2010)\citenamefont{Prodan, Hughes, and
  Bernevig}}]{Prodan2010}
\bibinfo{author}{\bibfnamefont{E.}~\bibnamefont{Prodan}},
  \bibinfo{author}{\bibfnamefont{T.~L.} \bibnamefont{Hughes}},
  \bibnamefont{and} \bibinfo{author}{\bibfnamefont{B.~A.}
  \bibnamefont{Bernevig}}, \bibinfo{journal}{Phys. Rev. Lett.}
  \textbf{\bibinfo{volume}{105}}, \bibinfo{pages}{115501}
  (\bibinfo{year}{2010}).

\bibitem[{\citenamefont{Fang et~al.}(2013)\citenamefont{Fang, Gilbert, and
  Bernevig}}]{fang2013}
\bibinfo{author}{\bibfnamefont{C.}~\bibnamefont{Fang}},
  \bibinfo{author}{\bibfnamefont{M.~J.} \bibnamefont{Gilbert}},
  \bibnamefont{and} \bibinfo{author}{\bibfnamefont{B.~A.}
  \bibnamefont{Bernevig}}, \bibinfo{journal}{Physical Review B}
  \textbf{\bibinfo{volume}{87}}, \bibinfo{pages}{035119}
  (\bibinfo{year}{2013}).

\bibitem[{\citenamefont{Sohal et~al.}(2020)\citenamefont{Sohal, Han, Santos,
  and Teo}}]{sohal-2020}
\bibinfo{author}{\bibfnamefont{R.}~\bibnamefont{Sohal}},
  \bibinfo{author}{\bibfnamefont{B.}~\bibnamefont{Han}},
  \bibinfo{author}{\bibfnamefont{L.~H.} \bibnamefont{Santos}},
  \bibnamefont{and} \bibinfo{author}{\bibfnamefont{J.~C.~Y.}
  \bibnamefont{Teo}}, \bibinfo{journal}{Phys. Rev. B}
  \textbf{\bibinfo{volume}{102}}, \bibinfo{pages}{045102}
  (\bibinfo{year}{2020}),
  \urlprefix\url{https://link.aps.org/doi/10.1103/PhysRevB.102.045102}.

\bibitem[{\citenamefont{Milovanovi{\'c} and
  Read}(1996)}]{Milovanovic-Read-1996}
\bibinfo{author}{\bibfnamefont{M.}~\bibnamefont{Milovanovi{\'c}}}
  \bibnamefont{and} \bibinfo{author}{\bibfnamefont{N.}~\bibnamefont{Read}},
  \bibinfo{journal}{Phy. Rev. B} \textbf{\bibinfo{volume}{53}}, 
  \bibinfo{pages}{13559}
  (\bibinfo{year}{1996}).

\bibitem[{\citenamefont{{Kitaev}}(2001)}]{2001PhyU...44..131K}
\bibinfo{author}{\bibfnamefont{A.~Y.} \bibnamefont{{Kitaev}}},
  \bibinfo{journal}{Physics Uspekhi} \textbf{\bibinfo{volume}{44}},
  \bibinfo{pages}{131} (\bibinfo{year}{2001}), \eprint{cond-mat/0010440}.

\bibitem[{\citenamefont{Kitaev}(2006)}]{kiaev-2006}
\bibinfo{author}{\bibfnamefont{A.}~\bibnamefont{Kitaev}},
  \bibinfo{journal}{Annals of Physics} \textbf{\bibinfo{volume}{321}}
  (\bibinfo{year}{2006}).

\bibitem[{\citenamefont{Casini et~al.}(2014)\citenamefont{Casini, Huerta, and
  Rosabal}}]{PhysRevD.89.085012}
\bibinfo{author}{\bibfnamefont{H.}~\bibnamefont{Casini}},
  \bibinfo{author}{\bibfnamefont{M.}~\bibnamefont{Huerta}}, \bibnamefont{and}
  \bibinfo{author}{\bibfnamefont{J.~A.} \bibnamefont{Rosabal}},
  \bibinfo{journal}{Phys. Rev. D} \textbf{\bibinfo{volume}{89}},
  \bibinfo{pages}{085012} (\bibinfo{year}{2014}),
  \urlprefix\url{https://link.aps.org/doi/10.1103/PhysRevD.89.085012}.

\bibitem[{\citenamefont{{Shapourian} et~al.}(2017)\citenamefont{{Shapourian},
  {Shiozaki}, and {Ryu}}}]{2017PhRvB..95p5101S}
\bibinfo{author}{\bibfnamefont{H.}~\bibnamefont{{Shapourian}}},
  \bibinfo{author}{\bibfnamefont{K.}~\bibnamefont{{Shiozaki}}},
  \bibnamefont{and} \bibinfo{author}{\bibfnamefont{S.}~\bibnamefont{{Ryu}}},
  \bibinfo{journal}{\prb} \textbf{\bibinfo{volume}{95}}, \bibinfo{eid}{165101}
  (\bibinfo{year}{2017}), \eprint{1611.07536}.

\bibitem[{\citenamefont{{Shapourian}
  et~al.}(2020{\natexlab{b}})\citenamefont{{Shapourian}, {Mong}, and
  {Ryu}}}]{2020arXiv201202222S}
\bibinfo{author}{\bibfnamefont{H.}~\bibnamefont{{Shapourian}}},
  \bibinfo{author}{\bibfnamefont{R.~S.~K.} \bibnamefont{{Mong}}},
  \bibnamefont{and} \bibinfo{author}{\bibfnamefont{S.}~\bibnamefont{{Ryu}}},
  \bibinfo{journal}{arXiv e-prints} \bibinfo{eid}{arXiv:2012.02222}
  (\bibinfo{year}{2020}{\natexlab{b}}), \eprint{2012.02222}.

\bibitem[{\citenamefont{{Hastings} et~al.}(2013)\citenamefont{{Hastings},
  {Nayak}, and {Wang}}}]{2013PhRvB..87p5421H}
\bibinfo{author}{\bibfnamefont{M.~B.} \bibnamefont{{Hastings}}},
  \bibinfo{author}{\bibfnamefont{C.}~\bibnamefont{{Nayak}}}, \bibnamefont{and}
  \bibinfo{author}{\bibfnamefont{Z.}~\bibnamefont{{Wang}}},
  \bibinfo{journal}{\prb} \textbf{\bibinfo{volume}{87}}, \bibinfo{eid}{165421}
  (\bibinfo{year}{2013}), \eprint{1210.5477}.

\bibitem[{\citenamefont{Nayak et~al.}(2008)\citenamefont{Nayak, Simon, Stern,
  Freedman, and Das~Sarma}}]{ChetanSimonSternFreedmanDasSarma}
\bibinfo{author}{\bibfnamefont{C.}~\bibnamefont{Nayak}},
  \bibinfo{author}{\bibfnamefont{S.~H.} \bibnamefont{Simon}},
  \bibinfo{author}{\bibfnamefont{A.}~\bibnamefont{Stern}},
  \bibinfo{author}{\bibfnamefont{M.}~\bibnamefont{Freedman}}, \bibnamefont{and}
  \bibinfo{author}{\bibfnamefont{S.}~\bibnamefont{Das~Sarma}},
  \bibinfo{journal}{Rev. Mod. Phys.} \textbf{\bibinfo{volume}{80}},
  \bibinfo{pages}{1083} (\bibinfo{year}{2008}),
  \urlprefix\url{http://link.aps.org/doi/10.1103/RevModPhys.80.1083}.

\bibitem[{\citenamefont{Nandkishore and
  Hermele}(2019)}]{doi:10.1146/annurev-conmatphys-031218-013604}
\bibinfo{author}{\bibfnamefont{R.~M.} \bibnamefont{Nandkishore}}
  \bibnamefont{and} \bibinfo{author}{\bibfnamefont{M.}~\bibnamefont{Hermele}},
  \bibinfo{journal}{Annual Review of Condensed Matter Physics}
  \textbf{\bibinfo{volume}{10}}, \bibinfo{pages}{295} (\bibinfo{year}{2019}),
  \eprint{https://doi.org/10.1146/annurev-conmatphys-031218-013604},
  \urlprefix\url{https://doi.org/10.1146/annurev-conmatphys-031218-013604}.

\bibitem[{\citenamefont{{Shirley} et~al.}(2018)\citenamefont{{Shirley},
  {Slagle}, and {Chen}}}]{2018arXiv180310426S}
\bibinfo{author}{\bibfnamefont{W.}~\bibnamefont{{Shirley}}},
  \bibinfo{author}{\bibfnamefont{K.}~\bibnamefont{{Slagle}}}, \bibnamefont{and}
  \bibinfo{author}{\bibfnamefont{X.}~\bibnamefont{{Chen}}},
  \bibinfo{journal}{arXiv e-prints} \bibinfo{eid}{arXiv:1803.10426}
  (\bibinfo{year}{2018}), \eprint{1803.10426}.

\bibitem[{\citenamefont{Sullivan et~al.}(2021)\citenamefont{Sullivan, Dua, and
  Cheng}}]{PhysRevResearch.3.023123}
\bibinfo{author}{\bibfnamefont{J.}~\bibnamefont{Sullivan}},
  \bibinfo{author}{\bibfnamefont{A.}~\bibnamefont{Dua}}, \bibnamefont{and}
  \bibinfo{author}{\bibfnamefont{M.}~\bibnamefont{Cheng}},
  \bibinfo{journal}{Phys. Rev. Research} \textbf{\bibinfo{volume}{3}},
  \bibinfo{pages}{023123} (\bibinfo{year}{2021}),
  \urlprefix\url{https://link.aps.org/doi/10.1103/PhysRevResearch.3.023123}.

\bibitem[{\citenamefont{{Ma} et~al.}(2020)\citenamefont{{Ma}, {Shirley},
  {Cheng}, {Levin}, {McGreevy}, and {Chen}}}]{2020arXiv201008917M}
\bibinfo{author}{\bibfnamefont{X.}~\bibnamefont{{Ma}}},
  \bibinfo{author}{\bibfnamefont{W.}~\bibnamefont{{Shirley}}},
  \bibinfo{author}{\bibfnamefont{M.}~\bibnamefont{{Cheng}}},
  \bibinfo{author}{\bibfnamefont{M.}~\bibnamefont{{Levin}}},
  \bibinfo{author}{\bibfnamefont{J.}~\bibnamefont{{McGreevy}}},
  \bibnamefont{and} \bibinfo{author}{\bibfnamefont{X.}~\bibnamefont{{Chen}}},
  \bibinfo{journal}{arXiv e-prints} \bibinfo{eid}{arXiv:2010.08917}
  (\bibinfo{year}{2020}), \eprint{2010.08917}.

\bibitem[{\citenamefont{Sullivan et~al.}(2021)\citenamefont{Sullivan, Joseph and Iadecola, Thomas and Williamson, Dominic J.}}]{PhysRevB.103.205301}
\bibinfo{author}{\bibfnamefont{J.}~\bibnamefont{Sullivan}},
  \bibinfo{author}{\bibfnamefont{T.}~\bibnamefont{Iadecola}},
  \bibnamefont{and}
  \bibinfo{author}{\bibfnamefont{D. J.}~\bibnamefont{Williamson}},
  \bibinfo{journal}{Phys. Rev. B} \textbf{\bibinfo{volume}{103}},
  \bibinfo{pages}{205301} (\bibinfo{year}{2021}).


\end{thebibliography}

\end{document}